\documentclass[
    reprint,
    amsmath,
    amssymb,
    aps,
    prx,
    longbibliography,
    nofootinbib,
    notitlepage,
]{revtex4-2}

\usepackage{physics}

\usepackage{graphicx}
\usepackage{dcolumn}
\usepackage{bm}
\usepackage[colorlinks=true,urlcolor=blue,linkcolor=blue,citecolor=magenta]{hyperref}
\usepackage[normalem]{ulem}
\usepackage{tikz}
\usepackage{amsmath, amssymb, amsthm, mathtools}
\usepackage[T1]{fontenc}
\usepackage{bbold}
\usepackage{standalone}

\usepackage{algorithm2e}

\def\({\left(}
\def\){\right)}

\newcommand{\ovl}{\overline}%
\newcommand{\avg}[1]{\langle#1\rangle}

\newcommand{\customref}[2]{\hyperref[#1]{\ref*{#1}#2}}
\newcommand{\nn}{\nonumber\\}

\usepackage{xcolor}

\definecolor{Ured}{HTML}{cc0000}
\definecolor{Ublue}{HTML}{1f65cf}
\definecolor{Ugreen}{HTML}{198a11}
\definecolor{appBlue}{RGB}{68, 114, 196}
\definecolor{appOrange}{RGB}{214, 116, 50}

\newcommand{\YL}[1]{{\color{black} #1}}

\newcommand{\PRXQ}[1]{{\color{black} #1}}

\definecolor{OliveGreen}{cmyk}{0.84, 0, 0.95, 0.30}

\newcommand{\pbf}{p_{\rm bf}}
\newcommand{\kbf}{\mathsf{K}_{\rm bf}}
\newcommand{\km}{\mathsf{K}_{\rm m}}
\newcommand{\pmeas}{p_{\rm m}}
\newcommand{\qs}{\mathbf{S}}
\newcommand{\qe}{\mathbf{E}}
\newcommand{\bfr}{\mathbf{r}}
\newcommand{\ti}{{t_{\rm 0}}}
\newcommand{\tf}{{t_{\rm f}}}
\DeclareMathOperator*{\argmax}{argmax}
\DeclareMathOperator{\sgn}{sgn}
\DeclareMathOperator{\wgt}{wgt}

\newcommand{\decoder}{\mathcal{D}}

\newcommand{\cycle}{\mathsf{C}}
\newcommand{\vertex}{\mathsf{v}}
\newcommand{\error}{\mathsf{e}}
\newcommand{\cc}{\mathsf{c}}
\newcommand{\synd}{\mathsf{s}}

\newcommand{\paulix}{\mathsf{X}}
\newcommand{\pauliz}{\mathsf{Z}}
\newcommand{\pauliy}{\mathsf{Y}}

\newtheorem{theorem}{Theorem}[section]
\newtheorem{corollary}{Corollary}[section]
\newtheorem{lemma}{Lemma}[section]
\newtheorem{proposition}{Proposition}[section]
\theoremstyle{definition}
\newtheorem{definition}{Definition}[section]

\usepackage{enumitem}

\makeatletter
\def\l@subsection#1#2{}
\def\l@subsubsection#1#2{}
\makeatother

\begin{document}

\title{Perturbative stability and error correction thresholds of quantum codes}

\author{Yaodong Li}
\thanks{Y.L. and N.O.D. contributed equally to this work.}
\affiliation{Department of Physics, Stanford University, Stanford, CA 94305}
\author{Nicholas O'Dea}
\thanks{Y.L. and N.O.D. contributed equally to this work.}
\affiliation{Department of Physics, Stanford University, Stanford, CA 94305}
\author{Vedika Khemani}
\affiliation{Department of Physics, Stanford University, Stanford, CA 94305}

\date{February 10, 2025}

\begin{abstract}

Topologically-ordered phases are stable to local perturbations, and topological quantum error-correcting codes enjoy thresholds to local errors.
We connect the two notions of stability by constructing classical statistical mechanics models for decoding general CSS codes and classical linear codes.
Our construction encodes correction success probabilities under uncorrelated bit-flip and phase-flip errors, and simultaneously describes a generalized $\mathbb{Z}_2$ 
lattice gauge theory with quenched disorder.
We observe that the clean limit of the latter is precisely the \textit{discretized} imaginary time path integral of the corresponding quantum code Hamiltonian when the errors are turned into a perturbative $\paulix$ or $\pauliz$ magnetic field.
Motivated by error correction considerations, we define general order parameters for all such generalized $\mathbb{Z}_2$ 
lattice gauge theories, and show that they are generally lower bounded by success probabilities of error correction.
For CSS codes satisfying the LDPC condition and with a sufficiently large code distance, we prove the existence of a low temperature ordered phase of the corresponding lattice gauge theories, particularly for those lacking Euclidean spatial locality and/or when there is a nonzero code rate.
We further argue that these results provide evidence to stable phases in the corresponding perturbed quantum Hamiltonians, obtained in the limit of continuous imaginary time. 
To do so, we distinguish space- and time-like defects in the lattice gauge theory.
A high free-energy cost of space-like defects corresponds to a successful ``memory experiment'' and suppresses the 
energy splitting among the ground states, 
while a high free-energy cost of time-like defects corresponds to a successful ``stability experiment'' and points to a nonzero gap to local excitations.

\end{abstract}

\maketitle

\tableofcontents

\section{Introduction}

Topologically-ordered phases are stable to local perturbations~\cite{Klich_2010, Bravyi_2010, Bravyi_2011}, and topological quantum error-correcting codes are stable to local errors~\cite{DKLP2001topologicalQmemory}.
The former stability allows their observation in nature, and the latter stability promises the possibility of scalable quantum computing. 
Although the connection between topologically-ordered phases and topological codes has long been appreciated~\cite{Kitaev_1997, knill1997theory}, the two notions of stability appear physically distinct.
The former is the property of a non-equilibrium process (described by a quantum channel), whereas the latter describes properties of the Hamiltonian in equilibrium, at zero temperature.
It is unclear how one might relate the two notions despite certain commonalities, including the requirement of a large code distance.

There has been significant recent progress in the search for quantum low-density parity check (LDPC) codes\footnote{The code is said to be LDPC when each check has a constant qubit degree and each qubit has a constant check degree.} with good parameters~\cite{hastings2020fiber,Breuckmann_2021, dinur2021locally, panteleev2021asymptotically, leverrier2022quantum, dinur2022good}\footnote{Throughout the paper we denote the code parameters as $[N, K(N), D(N)]$, where $N$ is the number of physical qubits, $K(N)$ the number of logical qubits, and $D(N)$ the code distance.
The parameters are ``good'' if both $K(N)$ and $D(N)$ scale linearly with $N$ in the limit $N \to \infty$ within a family of codes.} beyond topological codes in finite dimensions~\cite{bpt2010}.
Initially motivated by efficient fault-tolerant computing~\cite{gottesman2013faulttolerant}, these codes also lead to explicit constructions of exactly solvable Hamiltonians on higher-dimensional lattices, in non-Euclidean geometries, or even on graphs.
These models can exhibit exotic properties such as \YL{exponential} ground state degeneracies or even the lack of low-energy ``trivial'' states~\cite{Anshu_2023,freedman2013quantum}.
They challenge the conventional notion of topological order and call for a significant extension.
For example, it is unclear if thermodynamic limits are well defined for Hamiltonians of general quantum codes without geometric locality~\cite{Kovalev_2018, Jiang_2019}.
It is also unclear if the order parameter used for defining topological phases in low dimensions, namely the Wilson loop, still provides meaningful distinction between topologically ordered and disordered phases without Euclidean spatial locality (\YL{e.g. on expander graphs where no distinction between perimeter law and area law can be made}).
A basic question is whether the exactly solvable models represent fixed points of ``topologically ordered phases'' of matter or are fine-tuned and isolated in the parameter space.
\PRXQ{Therefore, a notion of perturbative stability of these code Hamiltonians is required.
Previously, there have been rigorous results for perturbative stability of codes in Euclidean lattices when the interactions are short ranged~\cite{Klich_2010, Bravyi_2010, Bravyi_2011, Michalakis_2013}. These approaches have only been extended to LDPC codes without geometric locality very recently~\cite{lavasani2024klocal}.}

While an understanding of phases associated to such codes is in the process of being fleshed out~\cite{sagar2016fracton, williamson2016fractal,  kubica2018ungauging, shirley2019foliated, rakovszky2023gauge, rakovszky2024product}, the codes' error correction thresholds are generally and unambiguously defined~\cite{Aharonov_1999}.\footnote{\label{fn:intro_memory_stability_expr}By error correction, we have in mind ``fault-tolerant memory experiments''~\cite{DKLP2001topologicalQmemory} and ``stability experiments''~\cite{Gidney_2022} \YL{--- discussed in detail in Sec.~\ref{sec:decoding_success_physical_consequences} ---} with a phenomenological error model including local qubit errors and syndrome errors, and where measurements are repeated for many rounds in time.
By threshold, we mean that when all the local error rates are below a constant value, and when the number of measurement rounds is below a ``memory time'' $\tf(N)$ (growing exponentially with $D(N)$), the probability of success of the memory experiment converges to 1 in the limit $N \to \infty$.}
Recall that a nonzero threshold is necessary for the code to be useful as a quantum memory.
In the case of topological codes in low dimensions, it follows from the existence of an ``ordered phase'' of a statistical mechanics (stat mech) model at low temperatures with weak uncorrelated random disorder~\cite{DKLP2001topologicalQmemory}.
Remarkably, quantum LDPC codes with a sufficiently large code distance $D(N) = \Omega(\log N)$\footnote{That is, $D(N) \geq c \ln N$ asymptotically for some $c>0$.} are also known~\cite{gottesman2013faulttolerant, pryadko2014} to have a nonzero threshold.
Comparing the two notions of stability, it is conceivable that the existence of error thresholds might suggest certain stability of LDPC codes as Hamiltonians, where the LDPC condition generalizes the notion of spatial locality.

\begin{figure*}
    \centering
    \includegraphics[width=.95\linewidth]{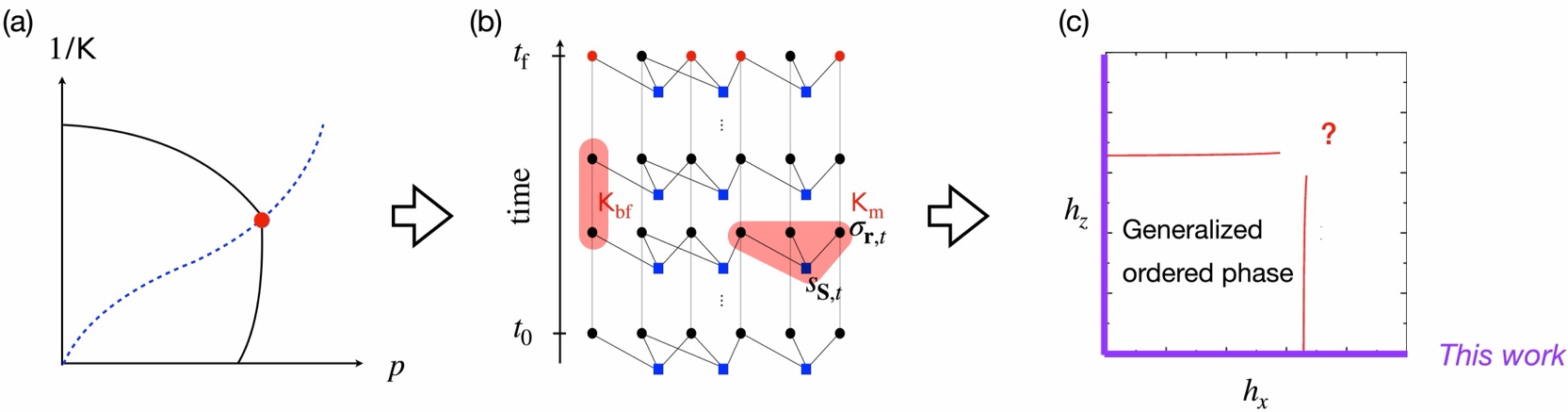}
    \caption{Summary of our argument.
    (a) The decoding of CSS codes under pure bitflip (Pauli $\paulix$) noise can be formulated as a statistical mechanics model with quenched disorder. (Similarly for pure Pauli $\pauliz$ noise, which we omit here.)
    Here, $p$ is the disorder strength, and $1/\mathsf{K}$ is the temperature of the model.
    The two parameters are constrained by the Nishimori condition, represented by the dashed blue line.
    The location of the critical point determines the optimal error threshold, below which decoding succeeds with probability $1$.
    \YL{The phase diagram is schematic and suppresses two dimensions; we consider both bitflip errors with rate $p_{\rm bf}$ and measurement errors with rate $p_{\rm m}$, corresponding to coupling strengths $\kbf$ and $\km$ via generalized Nishimori conditions, see Eqs.~(\ref{eq:tempphyserr},\ref{eq:measerr}).}
    (b) The statistical mechanics model \YL{for decoding bit-flip errors} takes a simple layered structure with one layer for each time step, see Eq.~\eqref{eq:H(sigma,s)}.
    There is one spin $\sigma_{\bfr, t}$ for each qubit $\bfr$ at each time step, and there is one spin $s_{\qs, t}$ for each $\pauliz$-stabilizer $\qs$ at each time step.
    The in-plane spins are coupled precisely by the $\pauliz$-stabilizers, and the inter-plane couplings are simple Ising terms.
    Successful decoding (i.e. ordering of the model on the Nishimori line) implies that \textit{all} $2^K$ logical operators $L_J$  at the final temporal boundary $t=\tf$ (represented in red) attain expectation value $1$ in the clean limit ($p=0$).
    (c) Schematic phase diagram of CSS codes under perturbations of a local magnetic field.
    \YL{The origin denotes the fixed point code Hamiltonian in Eq.~\eqref{eq:stabilizer_code_hamiltonian}.}
    We focus on the case where the perturbation is purely in the $\paulix$ or $\pauliz$ direction, where the discretized imaginary time path integral takes the form of the stat mech model in (b) in the clean limit.
    The generalized ordered phase has the logical operators $L_J$ as the order parameters.
    Therefore, decoding success provides evidence to the existence of a generalized ordered phase in the perturbed quantum Hamiltonian.}
    \label{fig:summary}
\end{figure*}

\PRXQ{Motivated by these developments, we make a physicists' attempt in directly connecting the two notions of stability, and argue that the existence of an error correction threshold of general Calderbank-Shor-Steane (CSS) codes~\cite{CalderbankShor, steane1996multiple} under simple noise models provides evidence to the perturbative stability of the corresponding code Hamiltonians under the simplest types of local perturbations.
Our argument is summarized in Fig.~\ref{fig:summary}.
En route, we formulate generalized stat mech models and their order parameters, describing discretized imaginary-time path integrals of the perturbed Hamiltonians.}

\subsection{Summary of results}

Our overall approach to demonstrating the connection is through stat mech of decoding for general CSS codes.
As detailed in Sec.~\ref{sec:stat_mech_model}, we formulate decoding as a problem of statistical inference.
Specifically, we collect error probabilities into weights of classical stat mech models with quenched disorder, see Eq.~\eqref{eq:H(sigma,s)}.
\PRXQ{Our derivation is equivalent to those in the literature~\cite{DKLP2001topologicalQmemory, pryadko2013, ChubbFlammia2018}.
However, it is agnostic of microscopic details, and
its simplicity permits immediate generalization to arbitrary CSS codes.}

The simple stat mech model with quenched disorder has a direct and exact correspondence 
with the \textit{discretized} imaginary time path integral of the corresponding quantum code Hamiltonian in a transverse field.
\YL{Here, the physical circuit time directly corresponds to the imaginary time of the path integral.}\footnote{We note that similar observations were made recently in Ref.~\cite{Andi_Bauer_2024}, in a different context.
Outside the context of quantum error correction, Lindbladian dynamics when post-selected (a form of feedback) on ``no-jump'' trajectories is known to simulate imaginary time evolution, and quantum trajectories of random circuits with measurements~\cite{nahum2018hybrid, li2018hybrid} are known to admit stat mech descriptions where the circuit time becomes a Euclidean time~\cite{nahum2018hybrid, andreas2019hybrid, choi2019spin, li2020cft, chenxiao2020nonunitary}.}
The discretized path integral takes the form of a clean stat mech model (more precisely, a generalized $\mathbb{Z}_2$ lattice gauge theory~\cite{wegner1971}) \textit{without quenched disorder}, and it is \textit{generically} the clean limit of the stat mech model for decoding. 
This correspondence is our main observation, and motivates a direct \PRXQ{implication from the existence of ``decodable phases'' of the code to the perturbative stability of the quantum Hamiltonian.
It would also be interesting to explore whether the implication in the other direction also holds, which we leave for future work.}

By virtue of the inference problem, we define a series of ``logical'' boundary order parameters of the generalized $\mathbb{Z}_2$ lattice gauge theories.
\YL{There are $2^{K(N)}$ of them where $K(N)$ can be as large as $O(N)$, but remarkably, all are lower bounded by a single quantity, namely the success probability.}
Therefore, invoking a Peierls style proof of nonzero threshold~\cite{DKLP2001topologicalQmemory, pryadko2014, hastings2023graph, breuckmann2018phd, griffiths1964peierls, peierls1936ising} for LDPC CSS codes,
we show in Sec.~\ref{sec:ordering_of_LGT} the existence of a low temperature ``ordered phase'' of the corresponding $\mathbb{Z}_2$ lattice gauge theories, as characterized by the condition that all logical boundary order parameters attain expectation value $1$ in the limit $N \to \infty$.
These boundary order parameters are generalizations of Wilson loop operators which detect the existence of homologically nontrivial disordering defects \textit{in the bulk} that change its sign.
\PRXQ{They can therefore be viewed as order parameters for spontaneous breaking of higher form or generalized symmetries.
However, unlike conventional Wilson loops which decay to $0$ in the thermodynamic limit for all nonzero temperatures (with different phases only distinguished by the form of the decay, e.g. perimeter versus area law),
our order parameters remain $1$ in the low temperature phase.}
This follows from boundary conditions natural to error correction (c.f. Eq.~\eqref{eq:SPAM_boundary_condition}), where all boundary couplings are set to infinity, similar to so-called Dobrushin boundary conditions in the literature~\cite{friedli_velenik_2017}.

In Sec.~\ref{sec:decoding_success_physical_consequences}, we make the non-rigorous but intuitive jump from \textit{discrete} to \textit{continuous} imaginary time path integrals, so as to translate boundary observables of the classical stat mech model into properties of the perturbed quantum Hamiltonian.
In particular, the small failure probability of the ``memory experiment''~\cite{DKLP2001topologicalQmemory} upper bounds the tunneling rate between \textit{any pair} of different ground states under discrete imaginary time evolution, which is exponentially suppressed by the code distance.
This is conventionally associated with an exponentially small splitting of the ground space of the perturbed Hamiltonian (Sec.~\ref{sec:memory_expr}).
Meanwhile, the small failure probability of the ``stability experiment''~\cite{Gidney_2022} upper bounds two-point correlation functions of excitation-creating operators in a slightly modified Hamiltonian; these correlations are shown to decay exponentially in their separation in imaginary time (Sec.~\ref{sec:stability_hard_wall} and \ref{sec:stability_hard_wall_two_punctures}).
This is conventionally interpreted as evidence for a finite gap in the excitation spectrum.
In the spacetime picture, the success of the memory experiment corresponds to the suppression of space-like defects that form logical errors, and the success of the stability experiment corresponds to the suppression of time-like defects.
The latter can be thought of as worldlines of elementary excitations of the Hamiltonian. Our aim in this section is translating the consequences of successful decoding into the language of perturbative stability of quantum phases, particularly arguing for the existence of an ordered phase given a non-vanishing threshold for error correction. The relationship between the two is most straightforward for the case of completely biased bit-flip or phase-flip noise corresponding to either an $\paulix$ or $\pauliz$ field perturbation to the code Hamiltonian (Fig.~\ref{fig:summary}(c)). We note that questions about the topology of the phase diagram and universality classes of the phase transitions, particularly in the bulk of the phase diagram, have a storied history~\cite{fradkinshenker1979} and are an active area of research for the case of the toric code and related gauge-Higgs models~\cite{vidal2009phase, vidal2011phase, fengcheng2012phase, nahum2020selfdual, nahum2024patching}. However, our focus in the current work remains on the axes of the phase diagram (and for general CSS LDPC codes, not just the toric code) and on the characteristics of their ordered phases, rather than transitions out of their ordered phases. Similarly, we do not directly comment on related problems of direct deformations of ground states~\cite{guoyi2019gapless} (relevant for teleportation of logical qubits in the toric code~\cite{guoyi2024robust}).

In Sec.~\ref{sec:discussion},
we discuss limitations of our arguments in the continuous imaginary time limit.
In particular, we point out that it is nontrivial to assume stability persists within this limit, and that additional evidence (beyond the existence of decoding thresholds) is needed in order to justify this assumption, see Appendix~\ref{sec:continuous_time_limit} for further details.
Therefore, our approach to the continuous time limit is not rigorous, and our results do not allow direct access to the spectrum of the perturbed quantum Hamiltonian. Our results in Sec.~\ref{sec:decoding_success_physical_consequences} should instead be viewed as evidence to the general perturbative stability of such Hamiltonians.
\PRXQ{On the other hand, they describe measurable properties of the code and are of experimental relevance.}
Our framework also allows one to build intuition for QEC from stat mech, and vice versa.
This is exemplified by a few examples and various different decoding experiments we proposed.
We anticipate future examples along these lines.

\YL{We also point out general difficulties in making the analogy between quantum error correction and zero temperature phases of the code Hamiltonian beyond CSS codes under biased noise models, particularly when the path integral might have a sign problem.}

This paper also has a few appendices.
In Appendix~\ref{sec:correlated_bitflip}, we discuss stat mech models with correlated bitflip errors and non-uniform error rates.
In Appendix~\ref{sec:nishimori}, we present a few supplementary mathematical results on the disordered stat mech model.
In Appendix~\ref{sec:GKS} we review correlation inequalities and present a corollary particularly useful when removing sign disorder from models.
In Appendix~\ref{sec:zbasis} we share more details on a link between boundary observables in a classical model and tunneling amplitudes in a corresponding quantum model.
In Appendices~\ref{sec:two_punctures_two_points} and~\ref{sec:stability_experiment_special_to_surface_codes}, we provide further details on the ``stability experiments'', including a formulation specialized for 2D surface codes in Euclidean and hyperbolic spaces that do not involve modifying the Hamiltonian, making use of a global redundancy.
In Appendix~\ref{sec:pfail_upperbound_MW}, we provide detailed upper bounds to decoding failure probabilities for LDPC CSS codes, following previous works~\cite{DKLP2001topologicalQmemory, pryadko2014, breuckmann2018phd}.
In Appendix~\ref{sec:continuous_time_limit}, we examine the continuous time limit, and we discuss whether decoding thresholds guarantee stability in this limit.
We also conjecture the asymptotic shape of the phase boundary of disordered models describing maximal-likelihood decoding for generic codes.
In Appendix~\ref{sec:examples_beyond}, we discuss two examples that are slightly beyond LDPC CSS codes under uncorrelated $\paulix$ and $\pauliz$ errors but whose stat mech models can be easily written down and analyzed.
\YL{They are (i) the toric code under pure $\pauliy$ noise and (ii) the Bacon-Shor code (which is a subsystem code with noncommuting checks).}

\section{Statistical mechanics for maximal-likelihood decoding \label{sec:stat_mech_model}}

When the code is a classical linear code or a quantum CSS code, the stability to errors can often be captured by classical stat mech models with quenched disorder~\cite{Rujan_1993, Sourlas_1994, nishimori2001book, mezardmontanari2009book, DKLP2001topologicalQmemory, pryadko2013, Zdeborova_2016, ChubbFlammia2018}. %
The decoding threshold of the code is identified with ``disordering'' phase transitions as driven by the error rate.
In particular, a maximal-likelihood decoder computes probabilities of syndrome-compatible errors within different equivalence classes, which are proportional to partition functions of the spin model within different homology classes~\cite{pryadko2013}.
In practice, such mappings have been most useful in giving numerical estimates of upper bounds for error thresholds for all decoders without the need to perform actual decoding (which can often be computationally prohibitive), but instead by analyzing thermodynamic functions of the disordered stat mech models across the phase transition~\cite{wang2003, Arakawa_2005,Kovalev_2018, kubica-2018-PRL-statmech-3Dcolorcode, andrist2011tricolored, andrist2015pra, haosong2022fractonstatmech}.

In Sec.~\ref{sec:stat_mech_model_derivation_subsection} we generalize these results 
to arbitrary CSS codes under uncorrelated $\paulix$ and $\pauliz$ noise.
This is a setup as close to classical codes as possible, and where connections with classical stat mech are the most transparent.
We focus on the \textit{fault tolerant memory experiment} setting where syndrome measurements are faulty and are repeated in time. (The case of perfect measurements is a special case to this more general situation, as we briefly discuss in Sec.~\ref{sec:perfect_measurements}.)
The main observation is that when the error model is a Markovian process, the transition probabilities are local and can be encoded into local positive Boltzmann weights of a classical stat mech model through generalized Nishimori conditions~\cite{nishimori1981internal,nishimori1986geometry}.

Our mapping aligns more closely with~\cite{Rujan_1993, Sourlas_1994} and are motivated by Bayesian considerations,
and is different from Dennis \textit{et.\,al.}~\cite{DKLP2001topologicalQmemory} which requires microscopic descriptions of excitations and error chains.
These descriptions are nevertheless very useful in estimating upper bounds to failure probabilities, see Appendix~\ref{sec:pfail_upperbound_MW}.
Ours also differs from recent works~\cite{fan2023diagnostics, colmenarez2023accurate, su2024tapestry, jongyeonlee2024exact, lyons2024understanding, yimu2024} deriving stat mech models from coherent information under a noisy quantum channel, invoking a general theorem by Schumacher and Nielsen~\cite{schumacher_nielsen_1996}.
These general approaches are presumably of wider applicability. Ours also differs from recent works~\cite{fan2023diagnostics,lyons2024understanding,zhao2023errorthresholds} deriving stat mech models from quantum relative entropy in the context of approximate Knill-Laflamme~\cite{knill1997theory} error-correction conditions.

The resultant stat mech models have a simple and intuitive form, see Eq.~\eqref{eq:H(sigma,s)} and Fig.~\ref{fig:summary}: they have a layered structure, with each layer representing a time step.
Spins at the same location but on neighboring layers are coupled by an Ising-like term, coming from bitflip errors; 
and spins within the same layer are coupled by the $\pauliz$ stabilizers (c.f. Eq.~\eqref{eq:stabilizer_code_hamiltonian}), coming from measurement errors.

In Sec.~\ref{sec:def_ML_decoding} we define maximal-likelihood decoding, and represent its success probability with disorder-averaged observables within the stat mech models.
As a technical point, the conventional definition of decoding success (Inference Problem~\ref{IP1}) does not immediately lend itself to any particular order parameter.
This motivates the definition of a series of sub-problems (Inference Problem~\ref{IP2}), which directly correspond to a series of generalized boundary Wilson loop operators with support on logical operators of the quantum code.
Therefore, successful decoding leads to ordered phases of the Ising gauge theories with the boundary Wilson loops as order parameters, as we discuss later in Sec.~\ref{sec:ordering_of_LGT}.

We note that disorder appears in Eq.~\eqref{eq:H(sigma,s)} in a slightly different form than the models in the literature; this is resolved by a simple change of basis, as we explain in Sec.~\ref{sec:uncorrelated_disorder}.
In Sec.~\ref{sec:gauge_invariance} we point out gauge invariance of the stat mech models for any quantum codes with both constant weight $\pauliz$ and $\paulix$ stabilizers, therefore justifying the name of ``gauge theories.''
In Sec.~\ref{sec:examples_and_gauge_invariance}, we show that several
known examples can be immediately read off from our derivation, either with or without measurement errors.

\subsection{Notation}
We are particularly interested in CSS stabilizer codes.
We denote the spins by their location $\bfr$.
The $\pauliz$ stabilizers of the code are uniquely specified by $\qs$ --- the subset of qubits within that stabilizer --- as
\begin{align}
    \pauliz_\qs \equiv \prod_{\bfr \in \qs} \pauliz_\bfr.
\end{align}
(In the case of the 2D toric code, these operators are conventionally known as $B_p$, \YL{but the notation here works for general CSS codes.})
Similarly, the $\paulix$ stabilizers of the code are uniquely specified by subsets denoted $\widetilde{\qs}$, where
\begin{align}
    \paulix_{\widetilde{\qs}} \equiv \prod_{\bfr \in \widetilde{\qs}} \paulix_\bfr.
\end{align}
(In the case of the 2D toric code, these operators are conventionally known as $A_v$.)
The unperturbed quantum code Hamiltonian is therefore
\begin{align}
\label{eq:stabilizer_code_hamiltonian}
    H_{\rm q} \equiv 
    - \sum_\qs \pauliz_{\qs} 
    -\sum_{\widetilde{\qs}} \paulix_{\widetilde{\qs}}.
\end{align}
\YL{The $\pauliz$ and $\paulix$ stabilizers are guaranteed to pairwise commute.}

\subsection{Stat mech models from transfer matrices \label{sec:stat_mech_model_derivation_subsection}}

CSS codes are capable of correcting $\paulix$- and $\pauliz$-type errors separately, and in the following, we imagine for concreteness a process with biased $\paulix$ noise (bitflip), whose effects will be picked up by the $\pauliz$-stabilizers.
After encoding logical information into the code space (we assume this step to be perfect), we have in mind an ``error model;'' namely, a time evolution where qubit errors and measurements of stabilizers take place at a finite rate.
As is customary~\cite{DKLP2001topologicalQmemory}, we model this in terms of timesteps; each timestep consists of the accumulation of physical errors and the faulty readout of stabilizer measurements.

We represent physical errors and measurement outcomes with classical Ising spins $\sigma$ and $s$.
There is a $\sigma$ spin for each qubit at each timestep, and there is an $s$ spin for each $\pauliz$-stabilizer $\qs$ (as noted, we are considering bitflip noise for concreteness) at each timestep.
In particular, we define 
\begin{align}
\begin{split}
    \label{eq:def_sigma_s}
    \sigma_{\bfr,t} =&\ (-1)^{\text{\# of bitflip errors on site }\bfr\text{ by time }t}, \\
    s_{\qs,t} =&\ \text{measurement outcome of stabilizer }\qs\text{ at time } t.
\end{split}
\end{align}
Here, we assume that a syndrome error only affects the syndrome bit, but not the physical qubits.
As there are no errors in the initial state by assumption, a measurement succeeds if $s_{\qs,t} = \prod_{\bfr \in \qs} \sigma_{\bfr,t}$, and fails if they do not agree.
We use the notation $\sigma_{t}$ and $s_{t}$ to refer to the sets of all $\sigma_{\bfr, t}$ and $s_{\qs, t}$ at time $t$, and we will use the notation $\sigma$ and $s$ to refer to the sets of all $\sigma_{\bfr, t}$ and $s_{\qs, t}$. %

For simplicity, we restrict our consideration to stochastic error processes whose parameters, like measurement error rates, are independent of time.\footnote{\label{fn:final_perfect_measurements}The exception being the final time-slice of perfect measurements.
This assumption models the experimental operation where measurement of all qubits in the computational basis are made at the end of a computation, which are often of much higher fidelity and assumed to be perfect.}
That is, we have that 
\begin{align}\label{eq:transition_probabilities}
\mathbb{P}(\sigma_{t}, s_{t}|\sigma_{t-1}, s_{t-1}) 
    = \mathbb{P}(s_{t}|\sigma_{t}) \cdot \mathbb{P}(\sigma_{t}|\sigma_{t-1}).
\end{align} 
is independent of time (for $t<\tf$).
$\mathbb{P}(\sigma_{t}|\sigma_{t-1})$ will be set by the model of physical qubit errors, and $\mathbb{P}(s_{t}|\sigma_{t})$ is the probability of seeing stabilizer outcomes $s_{t}$ given cumulative physical errors $\sigma_{t}$.
\YL{As is customary (c.f. footnote~\ref{fn:final_perfect_measurements})},
at the final time step $t = \tf$, we assume all stabilizer measurement outcomes are perfect; this allows one to define error correction success and thresholds in a natural manner.
With these, we have perfect ``state preparation and measurement'' (SPAM) temporal boundary conditions
\begin{equation}
\label{eq:SPAM_boundary_condition}
\begin{split}
    \forall \bfr,\quad \sigma_{\bfr, \ti} =&\, +1, \\
    \forall \qs, \quad s_{\qs,\tf} =&\, \prod_{\bfr \in \qs} \sigma_{\bfr,\tf}.
\end{split}
\end{equation}

For simplicity, we focus on i.i.d. bitflip noise and i.i.d. measurement errors. (Correlated bitflip errors are discussed in Appendix~\ref{sec:correlated_bitflip}.)
A bitflip error occurring with probability $\pbf$ will flip $\sigma_{\bfr,t+1}$ relative to $\sigma_{\bfr,t}$, and can be written suggestively as 
\begin{equation}\label{eq:tempphyserr}
\begin{split}
 \mathbb{P}(\sigma_{t+1}|\sigma_{t}) & \propto e^{ \kbf \sum_{\bfr} \sigma_{\bfr, t+1} \sigma_{\bfr, t}}, \\
 \kbf &\equiv \frac{1}{2}\log(\frac{1-\pbf}{\pbf}).
\end{split}
\end{equation}
If measurement errors occur with probability $\pmeas$, then we can similarly write 
\begin{equation}
\begin{split}\label{eq:measerr}
 \mathbb{P}(s_{t}|\sigma_{t}) & \propto e^{ \km \sum_{\qs} s_{\qs, t} \prod_{\bfr \in \qs} \sigma_{\bfr, t}}, \\
 \km &\equiv 
 \begin{cases}
     \frac{1}{2}\log(\frac{1-\pmeas}{\pmeas}), &t < \tf \\
     \infty, &t = \tf
 \end{cases}.
\end{split}
\end{equation}
The relation between the error rates $p_{\rm bf}, p_{\rm m}$ and the couplings $\kbf, \km$ in Eqs.~(\ref{eq:tempphyserr},\ref{eq:measerr}) are known as Nishimori conditions~\cite{nishimori1986geometry}.
This leads to
\begin{align}
\begin{split}
    \label{eq:H(sigma,s)}
    \mathbb{P}(\sigma, s) =&\, \prod_{t=t_0+1}^{\tf} \mathbb{P}(\sigma_{t}, s_{t}|\sigma_{t-1}, s_{t-1}) 
    \quad\propto\quad \exp{-H(\sigma, s)}, \\
    H(\sigma, s) =&\, -\sum_{t=t_0+1}^{\tf} \left[ \km \sum_{\qs} s_{\qs, t} \prod_{\bfr \in \qs} \sigma_{\bfr, t} + \kbf \sum_{\bfr} \sigma_{\bfr, t} \sigma_{\bfr, t-1} \right].
\end{split}
\end{align}
This defines the following partition function
\begin{align}
    \label{eq:Z[s]}
    Z[s] = \tr_{\sigma} \exp{-H(\sigma, s)},
\end{align}
which describes a spin model with quenched disorder $s$.
It is evident that $Z[s] \propto \mathbb{P}(s) \equiv \sum_\sigma \mathbb{P}(\sigma, s)$.
The boundary conditions in Eq.~\eqref{eq:SPAM_boundary_condition} are always assumed, when not stated explicitly.
Eq.~\eqref{eq:Z[s]} is the disordered stat mech model encoding \textit{all} error probabilities,
from which successful decoding can be defined, as we now discuss.

\subsection{Maximal-likelihood decoding \label{sec:def_ML_decoding}}

For a given error history $\sigma$ and syndrome history $s$, define
\begin{align}
    \label{eq:def_LJ}
    L_j(\sigma_{\tf}) = \prod_{\bfr \in L_j} \sigma_{\bfr, \tf},
\end{align}
where $L_j (j \in [K])$\footnote{In this paper we denote by $[N] (N \in \mathbb{Z})$ the set of integers $\{1, 2, \ldots, N\}$.}
is a logical representative of the $\pauliz$-logical acting on the $j$-th logical qubit.\footnote{$L_j(\sigma_\tf)$ is a logical representative made up of the classical auxiliary spins $\sigma$. Note that knowing $L_j(\sigma_\tf)$ does not mean that $L_j$ in the quantum code is known. Instead, correctly predicting $L_j(\sigma_\tf)$ corresponds to successful decoding by allowing errors to be removed without inducing a logical error through the following steps. 1. Clean up the final syndrome by acting a product of $\paulix$ operators; any arbitrary product that removes the syndrome works. 2. Do the following for each $j$. Add the number of times the cleanup operator had support on $L_j$ to the predicted value of $(1-L_j(\sigma_\tf))/2$. If the resulting sum has even parity, do nothing, and if the resulting sum has odd parity, act the $\paulix$-logical conjugate to $L_j$ on the quantum code. (In the case that the perfect readout at the final time step was from measuring in the computational basis to extract the $\pauliz$-logical information, note that the cleanup and conjugate logicals can be performed in classical processing.) This decoding process works because it ensures that the resulting syndrome is trivial, and the number of times a given logical representative has been hit by cumulative bitflip errors and clean-up and correction operators built out of $\paulix$ is even.}
Due to perfect syndrome measurements at $\tf$, the choice of the logical representative $L_j$ is inessential, as different choices can be related by perfect stabilizer syndromes.
Decoding is defined to be the following inference problem:

\begin{definition}[Inference problem: Predicting all logicals]\label{IP1}
Given the syndrome $s$ as input, the decoder $\decoder$ outputs $\decoder_j(s) = \pm 1$ for $j \in [K]$ as its prediction for $L_j(\sigma_\tf$).
\end{definition}

We say the decoder succeeds if its prediction for all $j$ are simultaneously correct.
Its probability of success can be represented as
\begin{equation}
\begin{split}
    \mathbb{P}^{\rm all}_\decoder
    \equiv&\ 
    \mathbb{P}(\forall j\in [K],\ \decoder_j(s) = L_j(\sigma_{\tf})) \\
    =&\ \sum_{s} \mathbb{P}(s)  \sum_{\sigma} \mathbb{P}(\sigma|s) \prod_{j=1}^K \frac{1+\decoder_j(s) \cdot L_j(\sigma_{\tf})}{2}.
\end{split}
\end{equation}
The maximal-likelihood (ML) decoder is defined to maximize $\mathbb{P}_\decoder^{\rm all}$,
\begin{equation}
\begin{split}
    \label{eq:def_D_ML_all}
    \decoder^{\rm ML}(s) &\equiv \argmax_{\ell \in \{\pm 1\}^{[K]}} \mathbb{P}(\forall j,\ L_j(\sigma_{\tf}) = \ell_j  | s)
    \\&= \argmax_{\ell \in \{\pm 1\}^{[K]}} \sum_{\sigma} \mathbb{P}(\sigma | s) \prod_{j=1}^K \frac{1+\ell_j \cdot L_j(\sigma_{\tf})}{2}.
\end{split}
\end{equation}
We will define ``thresholds'' in terms of $\mathbb{P}_\decoder^{\rm all}$ as usual.
That is, we say that we are below threshold for some decoder $\decoder$, a family of codes of growing system size $N$, and a final time $\tf(N)$ potentially growing as a function of $N$, if $\lim_{N \to\infty} \mathbb{P}^{\rm all}_\decoder = 1$.
Being below threshold means that the fraction of times that all $K(N)$ encoded logical qubits are simultaneously successfully decoded with $\decoder$ goes to $1$ as $N \to \infty$. 

By definition, we have $\mathbb{P}^{\rm all}_\decoder \leq \mathbb{P}^{\rm all}_{\decoder^{\rm ML}}$ for any $\decoder$, and in this sense the ML decoder is optimal.

As alluded to earlier, for the purpose of defining order parameters,
we also define more specialized inference problems relevant to successfully predicting products of subsets of $\pauliz$-logicals without requiring successful prediction of all constituent $\pauliz$-logicals at the same time.
Denote such a product of $\pauliz$-logicals by $L_J(\sigma_{\tf}) = \prod_{j \in J} L_j(\sigma_{\tf})$ for an indexing set $J \subseteq [K]$.
We have

\begin{definition}[Inference problem: Predicting a particular logical $L_J$]  \label{IP2}
Given the syndrome $s$ as input, the decoder $\decoder$ outputs $\decoder(s) = \pm 1$ as its prediction for $L_J$.
\end{definition}

The probability of successfully predicting this product is
\begin{equation}
\label{eq:P_succ_specialized_deooder}
\begin{split}
    \mathbb{P}^J_{\decoder}
    \equiv&\ 
    \mathbb{P}(\decoder(s) = 
    L_J(\sigma_{\tf})
    ) \\
    =&\ \sum_{s} \mathbb{P}(s)  \sum_{\sigma} \mathbb{P}(\sigma|s) \frac{1+\decoder(s) \cdot 
    L_J(\sigma_{\tf})}{2}.
\end{split}
\end{equation}
A maximal-likelihood decoder for this specialized task can be similarly defined,
\begin{equation}
\label{eq:def_specialized_deooder_ML}
\begin{split}
    \decoder_J^{\rm ML}(s) &\equiv 
    \argmax_{\ell = \pm 1} \sum_{\sigma} \mathbb{P}(\sigma | s)  \frac{1+\ell \cdot L_J(\sigma_{\tf})}{2}
    \\&= \sgn \left( \sum_{\sigma} \mathbb{P}(\sigma | s) \cdot L_J(\sigma_{\tf}) \right).
\end{split}
\end{equation}
Similarly, we have $\mathbb{P}^{J}_\decoder \leq \mathbb{P}^{J}_{\decoder^{\rm ML}_J}$ for any $\decoder$.
Together, Eq.~\eqref{eq:P_succ_specialized_deooder} and Eq.~\eqref{eq:def_specialized_deooder_ML} allow us to write
\begin{align}
\label{eq:2Ps-1_abs_avg_LJ}
    \Delta \mathbb{P}^J_{\decoder_J^{\rm ML}}
    \equiv&\
    \mathbb{P}^J_{\decoder_J^{\rm ML}} - (1 - \mathbb{P}^J_{\decoder_J^{\rm ML}}) \nn
    =&\ 
    \sum_{s} \mathbb{P}(s) \cdot \decoder_J^{\rm ML}(s)  \cdot \sum_{\sigma} \mathbb{P}(\sigma|s) \cdot 
    L_J(\sigma_{\tf}) \nn
    =&\ \sum_s \mathbb{P}(s) 
    \left|
        \sum_{\sigma} \mathbb{P}(\sigma|s) \cdot 
        L_J(\sigma_{\tf}) 
    \right| \nn
    =&\ [| \avg{L_J}_{Z[s]} |]_s.
\end{align}
In the final step, we denote by $\avg{...}_{Z[s]}$ the expectation value taken with respect to the partition function $Z[s]$ (Eq.~\eqref{eq:Z[s]}), and $[...]_s$ the average over quenched disorder $s$.
We henceforth suppress the explicit dependence of $L_J$ on $\sigma_\tf$, which is always implied.
In the following, we will write $[\avg{...}]_s$ for $[ \avg{...}_{Z[s]} ]_s$, to avoid cluttering.
Such decoders were first proposed in the classical coding literature, in particular by Ruj\'{a}n~\cite{Rujan_1993}.

Similar to our definition above, 
we say that a decoder $\decoder$ is below threshold for Inference Problem~\ref{IP2} if 
$\lim_{N \to\infty} \mathbb{P}^{J}_\decoder = 1$
for a family of codes of growing system size $N$, and a time $\tf(N)$ potentially growing as a function of $N$.
In particular, when $\decoder = \decoder^{\rm ML}_J$, being below threshold implies $\lim_{N \to \infty} [| \avg{L_J} |]_s \to 1$.

The two Inference Problems are related by the following observation
\begin{proposition}
\label{lemma:IP1_IP2_equivalence}
    Inference Problem~\ref{IP1} is successful under $\decoder^{\rm ML}$ ($\mathbb{P}^{\rm all}_{\decoder^{\rm ML}} \to 1$) \YL{implies that} Inference Problem~\ref{IP2} is successful under $\decoder^{\rm ML}_J$ ($\mathbb{P}^{J}_{\decoder^{\rm ML}_J} \to 1$) for each $J \subseteq [K]$.
\end{proposition}
\begin{proof}
\YL{This is to formalize the intuition that Inference Problem~\ref{IP1} is harder than Inference Problem~\ref{IP2}.
Formally, we have
\begin{align}
    \label{eq:P_succ_inequalities}
    \forall J \subseteq [K], \quad \mathbb{P}^{\rm all}_{\decoder^{\rm ML}} \leq \mathbb{P}^{J}_{\decoder^{\rm ML}} \leq \mathbb{P}^{J}_{\decoder^{\rm ML}_J}.
\end{align}
This is because for $\decoder^{\rm ML}$, success at Inference Problem~\ref{IP2} is guaranteed by success at Problem~\ref{IP1}; and for Problem~\ref{IP2}, $\decoder_J^{\rm ML}$ is a better decoder than $\decoder^{\rm ML}$.}
\end{proof}

\PRXQ{
It is conceivable~\cite{pryadko2013, Jiang_2019, ChubbFlammia2018, fahimniya2023faulttolerant, plackebreuckmann2023hyperbolicIsing} that there can exist intermediate ``phases'' where Inference Problem~\ref{IP1} fails but Inference Problem~\ref{IP2} succeeds on certain $J \subseteq [K]$ (c.f. Eq.~\eqref{eq:P_succ_inequalities}) particularly for codes with $K$ growing with $N$.
In these phases, the code protects a subset of logical qubits and potentially serves as a partial quantum memory.
As the error rate is increased, there might exist multiple transitions corresponding to gradual degradation of the memory, as driven by proliferation of different types of defects, see Appendix~\ref{sec:divergence_of_free_energy}.
The order parameters $[|\avg{L_J}|]_s$ may then provide fine-grained probes of these transitions.
}

\begin{figure*}[t]
    \centering
    \includegraphics[width=.95\linewidth]{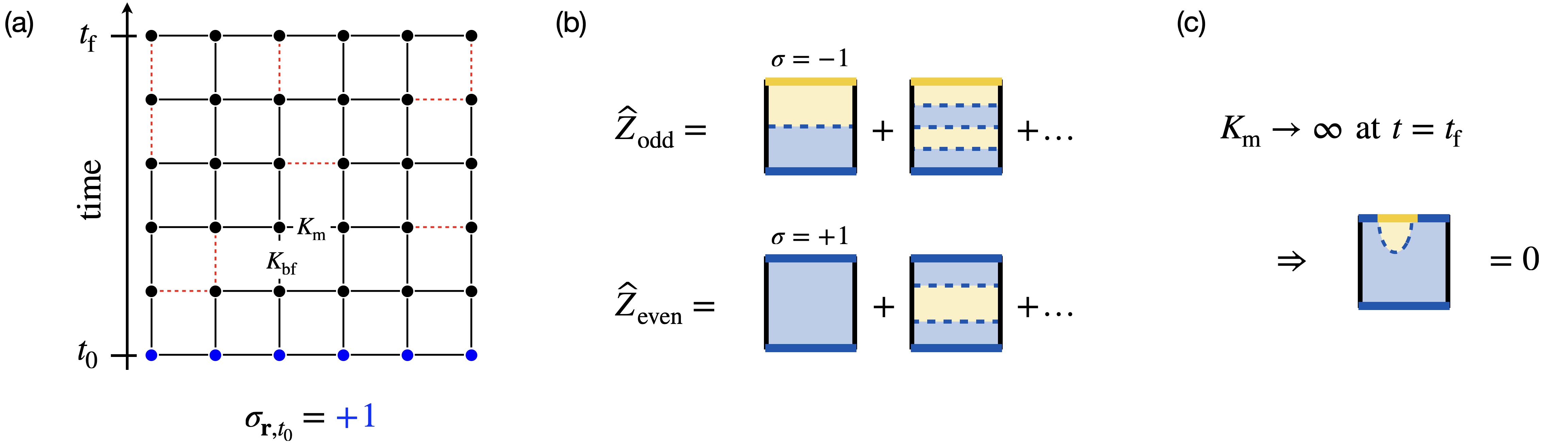}
    \caption{(a) Mapping the 1D repetition code to the 2D random bond Ising model, Eq.~\eqref{eq:H_RBIM}.
    The horizontal couplings are given by $\km$, and the vertical bond couplings are given by $\kbf$.
    The dashed bonds denote where $\error = -1$, occuring with probabilities $p_{\rm m}$ and $p_{\rm bf}$, as given by Nishimori conditions.
    The boundary conditions are $\sigma_{\bfr, \ti} = +1$, and $\km \to \infty$ at $t = \tf$, see Eq.~\eqref{eq:bc_for_RBIM}.
    (b) The order parameter $[\avg{\sigma_{\bfr, \tf}}]_\error$ depends on partition functions $\widehat{Z}_{\rm even, odd}$, as in Eq.~\eqref{eq:RBIM_sigma_partition_function}.
    The former has an even number of domain walls inserted between the two temporal boundaries $t=\ti, \tf$ (giving $\sigma_{\tf} = +1$), and the latter has an odd number, giving $\sigma_{\tf} = -1$.
    In the ordered phase of the RBIM we have $\widehat{Z}_{\rm odd}/\widehat{Z}_{\rm even} = O(e^{-L})$, due to a finite domain wall line tension, therefore $[\avg{\sigma_{\bfr, \tf}}]_\error \to 1$.
    (c) The boundary conditions forbid domain walls to terminate on the temporal edges defined by $t=\ti, \tf$. Therefore, any configuration with a wrong-signed $\sigma_\tf$ has a domain wall of length at least $L$.
    }
    \label{fig:Ising}
\end{figure*}

\subsection{Simplification: uncorrelated quenched disorder \label{sec:uncorrelated_disorder}}

The quenched ``syndrome'' disorder $s$ in Eq.~\eqref{eq:Z[s]} may be correlated through space and time, and therefore not directly comparable with disordered spin models considered in the literature.
We show in Appendix~\ref{sec:nishimori} that under a natural change of basis, a spin model with uncorrelated quenched disorder can be obtained,
\begin{align}
\begin{split}
    \label{eq:Z_eta_maintext}
    &\ \widehat{H}(\sigma, \error) \\
    &\quad =
    -\sum_{t} \left[ \km \sum_{\qs} \error_{\qs, t} \prod_{\bfr \in \qs} \sigma_{\bfr, t} + \kbf \sum_{\bfr} \error_{\bfr, t} \sigma_{\bfr, t+1} \sigma_{\bfr, t} \right], \\
    &\ \widehat{Z}[\error] = \tr_\sigma \exp {-\widehat{H}(\sigma, \error)}.
\end{split}
\end{align}
Intuitively, we go from a model with disorder configuration given by the syndrome $s$ to one where the disorder configuration is given by the bit-flip and syndrome error events $\error$.
Note that the quenched disorder $\error$ now couple to both time-like ``bonds'' and space-like    ``plaquettes''.
They are uncorrelated through spacetime, and obey the following probability distribution,
\begin{align}
    \mathbb{P}(\error) \propto \left(\frac{p_{\rm bf}}{1-p_{\rm bf}}\right)^{\wgt(\error_\bfr)} \cdot
    \left(\frac{p_{\rm m}}{1-p_{\rm m}}\right)^{\wgt(\error_\qs)},
\end{align}
where $\wgt(\cdot)$ counts the number of $-1$ components of a $\pm 1$ vector.
The two stat mech models are equivalent, and we have in particular
\begin{align}
    \label{eq:LJs=LJe}
    [| \avg{L_J} |]_s = [| \avg{L_J} |]_\error.
\end{align}
Furthermore, we have
\begin{align}
    \label{eq:LJ_avg_to_1}
    \mathbb{P}^J_{\decoder_J^{\rm ML}} \to 1
    \quad \Rightarrow \quad
    [\avg{L_J}]_\error \to 1,
\end{align}
where on the RHS no absolute value is taken, in contrast to Eq.~\eqref{eq:LJs=LJe}.
We refer the reader to Appendix~\ref{sec:divergence_of_free_energy} for more details.
For the remainder of this section, we focus on the model $\widehat{Z}[\error]$.

\subsection{Gauge invariance for quantum codes \label{sec:gauge_invariance}}

For any quantum code, the corresponding stat mech model generically takes the form of a lattice gauge theory.
Recall that a quantum code will have both $\paulix$ and $\pauliz$ stabilizers.
For each $\paulix$ stabilizer $\widetilde{\qs}$
\begin{align}
    \paulix_{\widetilde{\qs}} = \prod_{\bfr \in \widetilde{\qs}} \paulix_\bfr,
\end{align}
we introduce additional Ising variables $\sigma_{\widetilde{\qs}, t+1/2}$ between temporal layers, and consider a stat mech model with the following action
\begin{widetext}
\begin{align}
    \label{eq:H_gauge}
    \widehat{H}(\sigma, \error) \to \widehat{H}_{\rm gauge}(\sigma, \error) 
    =
    -\sum_{t} \left[ \km \sum_{\qs} \error_{\qs, t} \prod_{\bfr \in \qs} \sigma_{\bfr, t} + \kbf \sum_{\bfr} \error_{\bfr, t} \sigma_{\bfr, t+1} \sigma_{\bfr, t} \prod_{\widetilde{\qs} \ni \bfr} \sigma_{\widetilde{\qs}, t+1/2} \right].
\end{align}
\end{widetext}
The gauge transformations are generated by 
\begin{align}
    \label{eq:gauge_transformation}
    \forall \widetilde{\qs},   t \in [\ti, \tf],\quad
    \begin{cases}
        \forall \bfr \in \widetilde{\qs}, \quad &\sigma_{\bfr, t} \to - \sigma_{\bfr, t}, \\
        \text{if $t < \tf$},
        &\sigma_{\widetilde{\qs}, t+1/2} \to - \sigma_{\widetilde{\qs}, t+1/2}, \\
        \text{if $t > \ti$},
        &\sigma_{\widetilde{\qs}, t-1/2} \to - \sigma_{\widetilde{\qs}, t-1/2}.
    \end{cases}
\end{align}
As one can check, the invariance of the partition function is guaranteed by the commutation between $\paulix$ and $\pauliz$ stabilizers.
Under the gauge transformations it is always possible to fix $\sigma_{\widetilde{\qs}, t+1/2} = +1$ for all $\widetilde{\qs}$ and $t$ (the temporal gauge), thus establishing the equivalence between $\widehat{H}_{\rm gauge}(\sigma, \error)$ and $\widehat{H}(\sigma, \error)$.

The gauge invariance implies the lack of a local order parameter.
In general, gauge invariant quantities take the form of products of local terms in $\widehat{H}_{\rm gauge}(\sigma, \error)$, or their products with logical operators.
Gauge invariance also leads to the non-convergence of belief-propagation decoders~\cite{poulin2008bp_splitbelief}.

\subsection{Examples \label{sec:examples_and_gauge_invariance}}

Our derivation works for both classical linear codes and quantum CSS codes, as we illustrate with one familiar example from each class.
As we focus on bit-flip noise, the relevant stabilizers are those composed of Pauli $\pauliz$ operators.

\subsubsection{1D repetition code}

With the simplification in Sec.~\ref{sec:uncorrelated_disorder}, we consider the 1D repetition code.
Stabilizers of this code are two-body terms on nearest-neighbors, $\sigma_{\bfr} \sigma_{\bfr+\hat{\bf x}}$, and the corresponding partition function $\widehat{Z}[\error]$ takes the form of a random-bond Ising model (RBIM)~\cite{DKLP2001topologicalQmemory},\footnote{More generally, a repetition code in $d$ dimensions is described by the random-bond Ising model in $(d+1)$ dimensions.}
\begin{align}
    \label{eq:H_RBIM}
    \widehat{H}_{\rm RBIM}(\sigma, \error) 
    &=
    -\sum_{t} [ \km \sum_{\bfr} \error_{(\bfr, \bfr+\hat{\mathbf{x}}), t} \sigma_{\bfr, t} \sigma_{\bfr + \hat{\bf x}, t} \nn
    &\quad\quad\quad\quad
    + \kbf \sum_{\bfr} \error_{\bfr, t} \sigma_{\bfr, t+1} \sigma_{\bfr, t} ].
\end{align}
The perfect SPAM boundary conditions in Eq.~\eqref{eq:SPAM_boundary_condition} corresponds to the following
\begin{align}
\begin{split}
    \label{eq:bc_for_RBIM}
    \forall \bfr, \quad \sigma_{\bfr, \ti} =&\, +1 \\
    \forall \bfr, \bfr', \quad \sigma_{\bfr, \tf} =&\, \sigma_{\bfr', \tf}.
\end{split}    
\end{align}
Thus, in the stat mech model we can either have $\sigma_{\bfr, \tf} = +1$ or $\sigma_{\bfr, \tf} = -1$, corresponding to an even or odd number of domain walls between the two temporal boundaries, see Fig.~\ref{fig:Ising}.
Their difference can be detected by a local spin ``order parameter'' $\sigma_{\bfr, \tf}$ at an arbitrary site $\bfr$ on the $t=\tf$ boundary.
This is a logical representative of the repetition code; in fact, there is only one logical qubit of the repetition code, and Inference Problems~\ref{IP1} and \ref{IP2} coincide.
When the maximal likelihood decoder is below threshold, Eq.~\eqref{eq:LJ_avg_to_1} becomes
\begin{align}
    \label{eq:RBIM_sigma_partition_function}
    [\avg{L(\sigma_\tf)}]_\error = [\avg{\sigma_{\bfr, \tf}}]_\error = 
    \left[ \frac{\widehat{Z}_{\text{even}}[\error] - \widehat{Z}_{\text{odd}}[\error]}
    {\widehat{Z}_{\text{even}}[\error] + \widehat{Z}_{\text{odd}}[\error]}
    \right]_\error \to 1.
\end{align}
In this model, the decodable phase is the ordered phase of the RBIM, and $[\avg{\sigma_{\bfr, \tf}}]_\error \to 1$ is simultaneously the consequence of successful decoding and ferromagnetic ordering.
In this phase, the domain walls have finite surface tension in the ordered phase and are therefore suppressed.

\subsubsection{2D toric code}

\begin{figure}[t]
    \centering
    \includegraphics[width=1.0\linewidth]{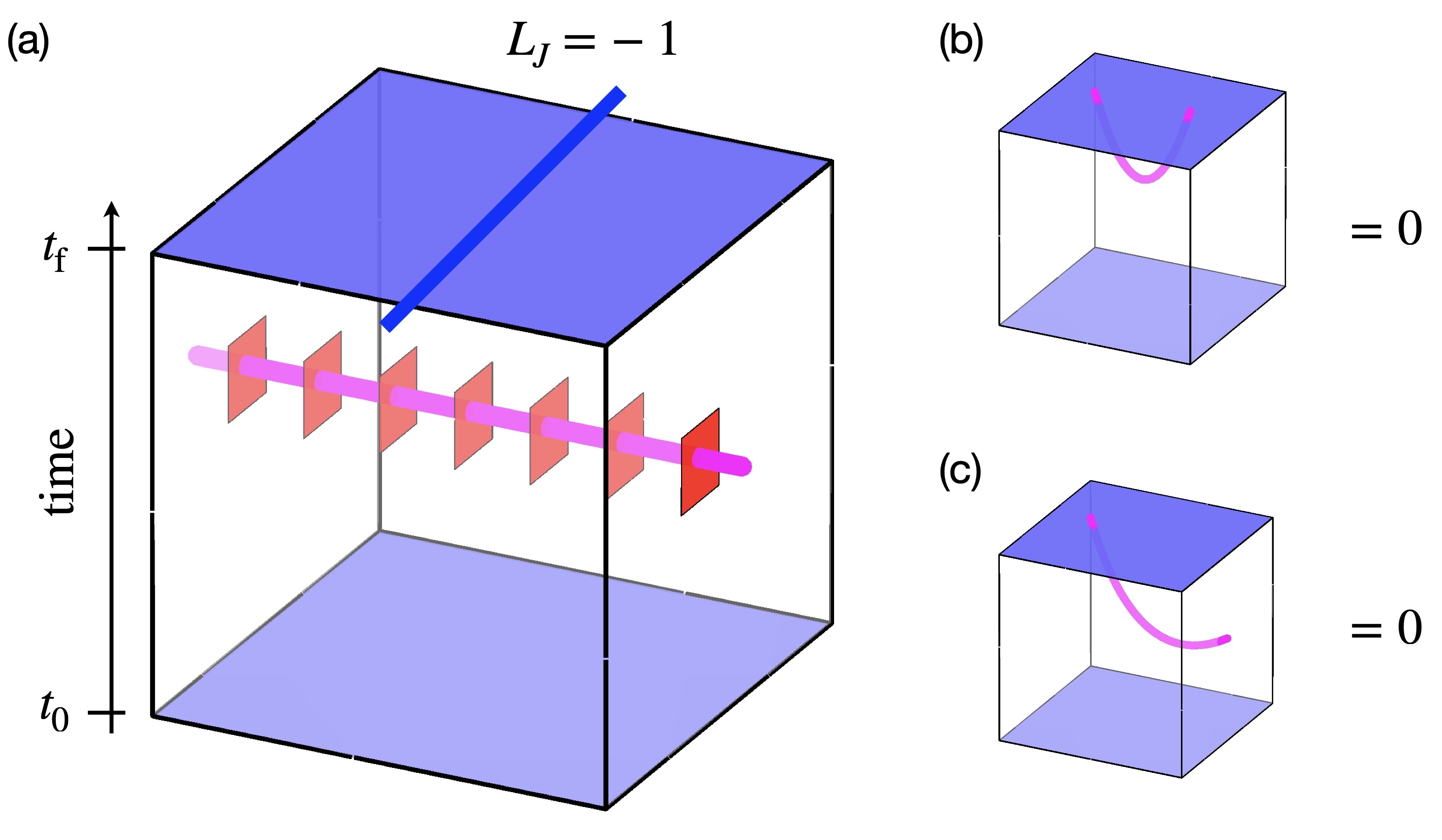}
    \caption{(a) Mapping the 2D toric code to the 3D random plaqutte gauge theory, Eq.~\eqref{eq:H_RPGM}.
    The boundary conditions are given in Eq.~\eqref{eq:SPAM_boundary_condition}, namely $\sigma_{\bfr, \ti} = +1$ and $\km \to \infty$ at $t = \tf$.
    The order parameter $L_J(\sigma_\tf)$ gets a wrong sign if an odd number of transversal flux loop (magenta) are inserted in the bulk, each piercing through a trail of wrong-signed plaquettes (colored in red, for which $\sigma_{\bfr, t} \sigma_{\bfr, t+1} = -1$).
    In the deconfined phase of the RPGM the flux tubes have a finite line tension, and $[\avg{L_J(\sigma_\tf)}]_\error \to 1$.
    (b,c) The boundary conditions forbid flux lines to terminate on the temporal boundaries defined by $t=\ti, \tf$, and any configuration with a wrong-signed $L_J(\sigma_\tf)$ has a flux loop of length at least $L$.}
    \label{fig:toric-code-flux-loop}
\end{figure}

For the 2D toric code~\cite{Kitaev_1997} under bitflip noise, we obtain a 3D stat mech model with intraplane 4-body couplings on plaquettes, and interplane two-body couplings, 
\begin{align}
    \label{eq:H_RPGM}
    \widehat{H}_{\rm RPGM}(\sigma, \error) 
    &=
    -\sum_{t} [ \km \sum_{\Box} \error_{\Box, t} \prod_{\bfr \in \Box} \sigma_{\bfr, t}
    \nn &\quad\quad\quad\quad
    + \kbf \sum_{\bfr} \error_{\bfr, t} \sigma_{\bfr, t+1} \sigma_{\bfr, t} ].
\end{align}
We recognize that this is the 3D classical $\mathbb{Z}_2$ gauge theory with random sign disorder in the couplings (``random-plaquette gauge model''~\cite{DKLP2001topologicalQmemory, wang2003, takeda2003selfdual}), after all time-like bonds are gauge-fixed to be $+1$ (known as the temporal gauge, see Sec.~\ref{sec:gauge_invariance}).
\PRXQ{Indeed, the Hamiltonian before gauge fixing takes the form of Eq.~\eqref{eq:H_gauge}, and is explicitly equivalent to the 3D random-plaquette gauge model.}
The logical operator for this code takes the form of a noncontractible ``Wilson loop'' of length $O(L)$ at final time $\tf$, namely $L_J(\sigma_\tf)$, see Fig.~\ref{fig:toric-code-flux-loop}.
The decodable phase of the toric code corresponds to the low-temperature deconfined phase of the RPGM, and we have according to Eq.~\eqref{eq:LJ_avg_to_1}
\begin{align}
    \mathbb{P}^J_{\decoder_J^{\rm ML}} \to 1 
    \quad \Rightarrow \quad
    [\avg{L_J}]_\error \to 1
\end{align}
Note that a Wilson loop in the bulk has an expectation value $O(e^{-L})$ (a ``perimeter law'' scaling) in the deconfined phase, due to constant-size flux loops that braid with it.
However, the infinite $\km$ couplings at $t=\tf$ forbid small flux loops to terminate on the upper boundary.
The leading contribution to a wrong-signed $L_J$ comes from a noncontractible flux loop of length $O(L)$ in the bulk, in the transverse direction of $L_J$, see Fig.~\ref{fig:toric-code-flux-loop}.
The flux loop has a finite line tension at low temperature and costs an extensive free energy.
In this way, the boundary condition effectively promotes the boundary order parameter $L_J(\sigma_\tf)$ from a perimeter-law to a ``1-law''.\footnote{We call this law a ``1-law'' because the expectation value tends to one. Under a strict analogy with perimeter and area laws, this might be better termed a ``0-law'' since these expectation values have zero scaling with their weights.}
Similar observations were exploited previously~\cite{YL2021z2decoing} in designing suboptimal decoders. We emphasize that in the below-threshold phase of all the models we consider, this ``1-law'' is a feature of all the order parameters $[\avg{L_J}]_\error$ we have proposed.

\subsubsection{Reduction to perfect measurements \label{sec:perfect_measurements}}

By setting $\km \to \infty$ and $p_{\rm m} \to 0$ everywhere, we recover stat mech models that describe perfect measurements~\cite{pryadko2013}.
For example, for $\widehat{H}_{\rm RBIM}$ in Eq.~\eqref{eq:H_RBIM}, $\km \to \infty$ imposes the constraint that $\sigma_{\bfr, t} = \sigma_{\bfr + \hat{\bf x}, t} \equiv \sigma_t$ for all $\bfr, t$. (Notice that $\error_{(\bfr, \bfr+\hat{\mathbf{x}}), t} = +1$ uniformly, due to the Nishimori condition.)
Under this reduction,
\begin{align}
    \lim_{\km \to \infty} \widehat{H}_{\rm RBIM}(\sigma, \error) 
    &=
    -\sum_{t} [ \kbf \sigma_{t+1} \sigma_{t} \sum_{\bfr} \error_{\bfr, t}].
\end{align}
Again due to the Nishimori condition, whenever $\kbf > 0$, the random variable $\sum_{\bfr} \error_{\bfr, t}$ becomes binomial with mean $L \cdot p_{\rm bf}$ and standard deviation $O(\sqrt{L})$.
Therefore, with probability $1$, the model becomes an Ising model with only ferromagnetic couplings, where the coupling strengths diverge with system size.
There is a unique ground state, namely $\sigma_t = +1$ for all $t$, whose weight dominates the partition function, as long as $\kbf > 0$ (i.e. no phase transitions at finite temperatures).
This is consistent with a $50\%$ threshold for repetition codes when measurements are perfect.

Similarly, sending $\km \to \infty$ for $\widehat{H}_{\rm RPGM}(\sigma, \error)$ in Eq.~\eqref{eq:H_RPGM} imposes the constraint  
\begin{align}
    \label{eq:perfect_measurement_contraint_RPGM}
    \forall \Box, t, \quad \prod_{\bfr \in \Box} \sigma_{\bfr, t} = 1.
\end{align}
Recall that the spins $\sigma_{\bfr, t}$ live on bonds of the lattice.
Writing $\bfr = \avg{ij}$, the constraint can be solved by $\sigma_{\avg{ij}, t} = \tau_{i,t} \cdot \tau_{j,t}$.
Therefore,
\begin{align}
    \label{eq:H_RPGM_perfect_measurements}
    &\lim_{\km \to \infty} 
    \widehat{H}_{\rm RPGM}(\sigma, \error)
    =
    -\sum_{t} [ \kbf \sum_{\avg{ij}} \error_{\avg{ij}, t} \widetilde{\tau}_{i, t} \widetilde{\tau}_{j, t} ],
\end{align}
where $\widetilde{\tau}_{i, t} \equiv \tau_{i, t}  \tau_{i, t+1}$.
This is the familiar RBIM, with disorder realization given by the bitflip errors.

In general, the perfect measurement ($p_{\rm m} = 0$, $\km \to \infty$) condition effectively imposes many hard constraints on the stat mech model, and we often get an ``ungauged''~\cite{kubica2018ungauging} stat mech model as compared to the $p_{\rm m} > 0$ one.
The success probability of decoding no longer admits the form of the expectation value of an order parameters; the original spins may not even appear in the stat mech model, as exemplified by Eq.~\eqref{eq:H_RPGM_perfect_measurements}.
Instead, the decoding transition is often signified by the vanishing of free energy cost when introducing ``twisted'' boundary conditions~\cite{pryadko2013}.

\section{Ordering of generalized lattice gauge theories \label{sec:ordering_of_LGT}}

Using decoding success, we show (Theorem~\ref{thm:ordering_of_clean_stat_mech_model}) that low-temperature ordered phases exist for generalized Ising gauge theories associated with the code, see Eq.~\eqref{eq:Z_clean_spin_model}.
We highlight that this result applies generally to all CSS codes with the LDPC condition, and with sufficiently large code distance (c.f. Eq.~\eqref{eq:LDPC_condition}).
In particular, the code may have higher form symmetries or other generalized symmetries.
The boundary observables with Dobrushin-like temporal boundary conditions (c.f. Eq.~\eqref{eq:SPAM_boundary_condition}) therefore serve generally as order parameters for spontaneous breaking of these symmetries.

We focus on CSS codes and error models satisfying the following conditions
\begin{align}
\label{eq:LDPC_condition}
    w = O(1), \quad D(N) = \Omega(\log(N)), \quad 
    \tf = O(e^{D(N)}),
\end{align}
where $w \equiv \max_{\qs} {\rm deg}(\qs)$ is the maximum degree of all $\pauliz$ stabilizers.

\begin{lemma}
\label{lemma:1-1}
The failure probability $\epsilon(N) \equiv 1 - \mathbb{P}^{\rm all}_{\decoder^{\rm ML}}$ satisfies
\begin{align}
\label{eq:pfail_upperbound_memory_maintext}
    \epsilon(N) = O\left( (N \cdot \tf)  \cdot [(w+1) \cdot \widetilde{q}]^{D(N)} \right),
\end{align}
where
\begin{align}
    \label{eq:def_q_tilde}
    \widetilde{q} =&\ 2\sqrt{q(1-q)}, \quad 
    q = \max \{p_{\rm bf}, p_{\rm m}\}.
\end{align}
In particular, if conditions in Eq.~\eqref{eq:LDPC_condition} are met,
there exists a threshold value $\widetilde{q}_{\rm th}$ such that
\begin{align}
    \label{eq:epsilon_N_to_zero}
    \widetilde{q} < \widetilde{q}_{\rm th} \quad \Rightarrow \quad
    \lim_{N \to \infty} \mathbb{P}^{\rm all}_{\decoder^{\rm ML}} = 1.
\end{align}
\end{lemma}
\begin{proof}
We prove in Appendix~\ref{sec:pfail_upperbound_MW} that for minimal weight decoders $\decoder^{\rm MW}$, $\epsilon^{\rm MW}(N) \equiv 1 - \mathbb{P}^{\rm all}_{\decoder^{\rm MW}}$
satisfy $\epsilon^{\rm MW}(N) = O\left( (N \cdot \tf)  \cdot [(w+1) \cdot \widetilde{q}]^{D(N)} \right)$.
Eq.~\eqref{eq:pfail_upperbound_memory_maintext} follows from that $\epsilon(N) \leq \epsilon^{\rm MW}(N)$, due to the optimality of $\decoder^{\rm ML}$ (see Eq.~\eqref{eq:def_D_ML_all}).
It is worth remarking that no $K(N)$-dependent entropic factor appears in Eq.~\eqref{eq:pfail_upperbound_memory_maintext}, as one might have naively expected from having to successfully decode every single logical simultaneously.
\end{proof}

\begin{lemma}
\label{lemma:1-2}
\begin{align}
    \forall J \subseteq [K], \quad
    [\left| \avg{L_J} \right|]_\error \geq 2 \mathbb{P}^{\rm all}_{\decoder^{\rm ML}} - 1.
\end{align}
\end{lemma}
\begin{proof}
By arguments in Proposition~\ref{lemma:IP1_IP2_equivalence}, we have
\begin{align}
    \label{eq:P_succ_inequalities}
    \forall J \subseteq [K], \quad \mathbb{P}^{\rm all}_{\decoder^{\rm ML}} \leq \mathbb{P}^{J}_{\decoder^{\rm ML}} \leq \mathbb{P}^{J}_{\decoder^{\rm ML}_J}.
\end{align}
Therefore, by Eqs.~(\ref{eq:2Ps-1_abs_avg_LJ}, \ref{eq:LJs=LJe}), we have
\begin{align}
\label{eq:abs_avg_LJ_decoding_lower_bound}
    [\left| \avg{L_J} \right|]_\error = [\left| \avg{L_J} \right|]_s = 2 \mathbb{P}^{J}_{\decoder^{\rm ML}_J} - 1 \geq 2 \mathbb{P}^{\rm all}_{\decoder^{\rm ML}} - 1.
\end{align}
\end{proof}

Consider the following stat mech model \emph{without disorder}
\begin{align}
    \label{eq:Z_clean_spin_model}
    Z = \tr_\sigma \exp {
        \sum_{t = t_0+1}^{\tf} \left[ \km \sum_{\qs} \prod_{\bfr \in \qs} \sigma_{\bfr, t} + \kbf \sum_{\bfr} \sigma_{\bfr, t} \sigma_{\bfr, t-1} \right]
    }.
\end{align}
This is the clean limit of Eq.~\eqref{eq:Z[s]} and Eq.~\eqref{eq:Z_eta_maintext}, when the quenched random-sign disorder is removed.
The boundary conditions, assumed even when not stated, are 
\begin{equation}
\label{eq:cleanSPAM_boundary_condition}
\begin{split}
    &\forall \bfr,\quad \sigma_{\bfr, \ti} =\, +1, \\
    &\forall \qs, \quad \prod_{\bfr \in \qs} \sigma_{\bfr,\tf} =\, +1.
\end{split}
\end{equation}
We will refer to this model as the (clean) ``lattice gauge theory'' associated with the code.
We have
\begin{lemma}
\label{lemma:1-3}
\begin{align}
    \forall J \subseteq [K], \quad
    [\left| \avg{L_J} \right|]_\error \leq \avg{L_J},
\end{align}
where the RHS denotes the ensemble average of $L_J$ with respect to the clean stat mech model Eq.~\eqref{eq:Z_clean_spin_model}.
\end{lemma}
\begin{proof}
In fact, for every disorder realization $\error$, $|\avg{L_J}_{Z[\error]}| \leq \avg{L_J}$. See Eq.~\eqref{eq:GKS_2}, which is a corollary of the ``GKS'' correlation inequalities due to Griffiths~\cite{griffiths1967a_correlations,griffiths1967b_correlations, griffiths1967c_correlations}, Kelly and Sherman~\cite{kelly1968general}.
\end{proof}

\PRXQ{While the step of removing quenched random sign disorder here is formal, we note that physical meaning can be associated to it via ``post-selected error correction'', as recently shown by English, Williamson, and Bartlett~\cite{english2024thresholds}.}

\begin{figure}[b]
    \includegraphics[width=\linewidth]{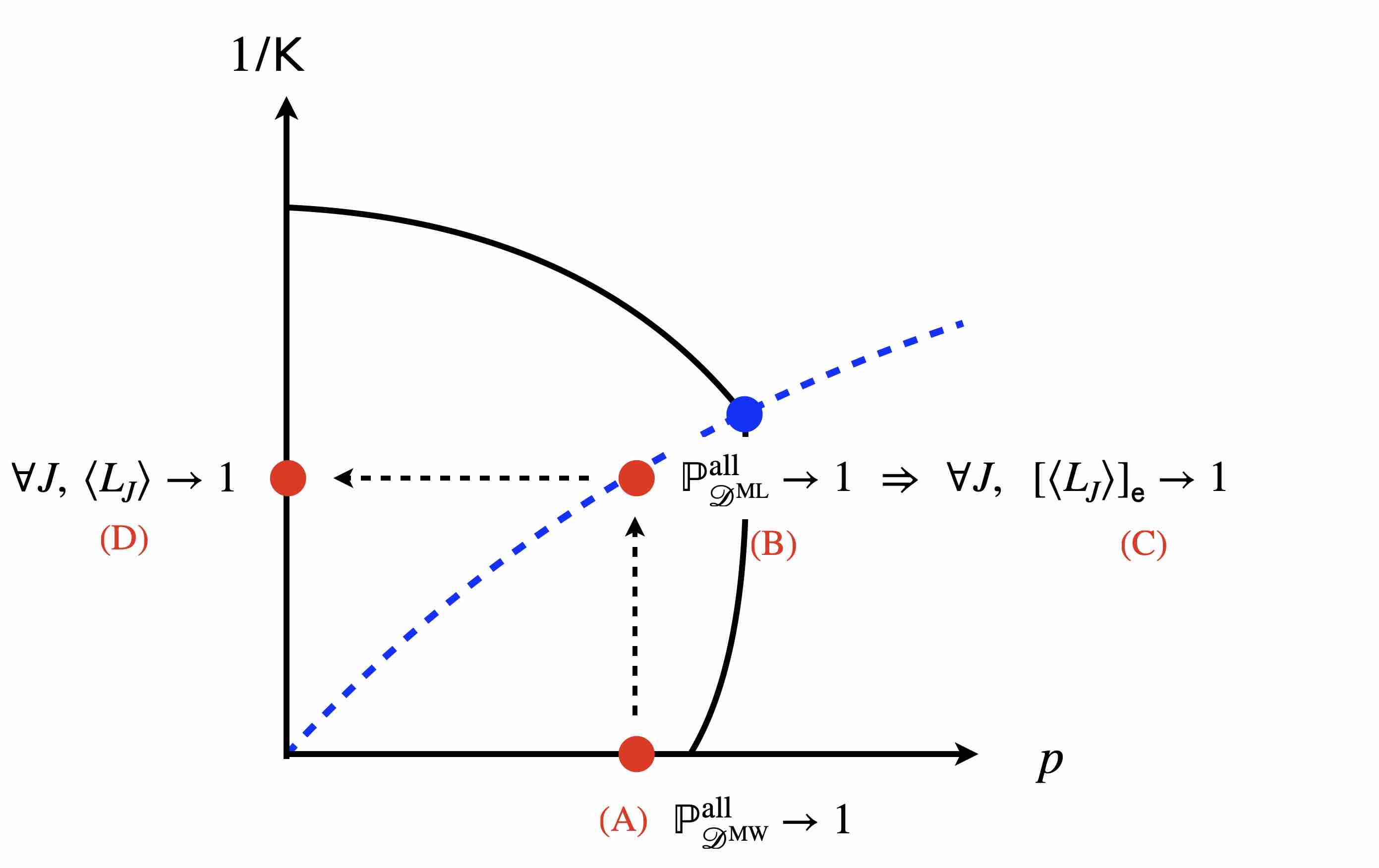}
    \caption{Proof of Theorem~\ref{thm:ordering_of_clean_stat_mech_model}.
    To establish $(\mathrm{D})$ the stability of the clean model at $(K, p=0)$, we turn to the corresponding point $(K, p)$ on the Nishimori line (blue dashed), which by GKS theorem (Lemma~\ref{lemma:1-3}) is always less ordered (therefore $(\mathrm{C}) \Rightarrow (\mathrm{D})$).
    $(\mathrm{C})$ refers to ordering of the disordered model at $(K, p)$ on the Nishimori line.
    The ordering can therefore be related to $(\mathrm{B})$ the success of decoding under the maximal likelihood decoder (which is optimal), where the Nishimori condition allows interpreting partition functions at this point as error probabilities, see Lemma~\ref{lemma:1-2}.
    $(\mathrm{B})$ follows from $(\mathrm{A})$ explicit upper bounds on the failure probabilities of the (suboptimal) minimal weight decoder, see Lemma~\ref{lemma:1-1} and Eq.~\eqref{eq:pfail_upperbound_memory_maintext}. The phase diagram is schematic and suppresses two dimensions; we consider both bitflip errors with rate $p_{\rm bf}$ and measurement errors with rate $p_{\rm m}$, corresponding to coupling strengths $\kbf$ and $\km$ via generalized Nishimori conditions, see Eqs.~(\ref{eq:tempphyserr},\ref{eq:measerr}).
    }
    \label{fig:proof_outline}
\end{figure}

\begin{theorem}
\label{thm:ordering_of_clean_stat_mech_model}
Let the stat mech model Eq.~\eqref{eq:Z_clean_spin_model} with boundary conditions Eq.~\eqref{eq:SPAM_boundary_condition} satisfy the conditions in Eq.~\eqref{eq:LDPC_condition}.
There exists a finite $\mathsf{K}_{\rm th}$ such that for $\kbf, \km > \mathsf{K}_{\rm th}$ we have
\begin{align}
\label{eq:ordering_of_clean_stat_mech_model}
    \forall J \subseteq [K(N)], \quad \lim_{N \to \infty} \langle L_J(\sigma_\tf) \rangle \to 1.
\end{align}
\end{theorem}
\begin{proof}
The conclusion follows from the previous Lemmas, which we outline in Fig.~\ref{fig:proof_outline} with the help of the phase diagram of the disordered stat mech model.
In particular, $(\mathrm{A})$ and $(\mathrm{A}) \Rightarrow (\mathrm{B})$ follow from Lemma~\ref{lemma:1-1}, $(\mathrm{B}) \Rightarrow (\mathrm{C})$ follows from Lemma~\ref{lemma:1-2}, and $(\mathrm{C}) \Rightarrow (\mathrm{D})$ follows from Lemma~\ref{lemma:1-3}.
\end{proof}

We have thus shown the existence of a low-temperature ordered phase of the lattice gauge theory Eq.~\eqref{eq:Z_clean_spin_model}, as characterized by the gauge-invariant boundary order parameters $L_J(\sigma_\tf)$ for all $J \subseteq [K]$.
While natural to error correction, they appear to be somewhat unconventional.
As we argue in the next section, the boundary order parameters also have a clear physical meaning when the partition function is interpreted as a path integral.

We note that Theorem~\ref{thm:ordering_of_clean_stat_mech_model} can also be proven by generalizing Peierls arguments for the Ising model~\cite{peierls1936ising, griffiths1964peierls}.
Such a proof will follow closely that of Lemma~\ref{lemma:1-1} (see Appendix~\ref{sec:pfail_upperbound_MW}), only reinterpretation of terms in the series expansion is needed.

\section{Ground state energy splitting and excitation gap from memory and stability experiments\label{sec:decoding_success_physical_consequences}}

We now explore physical consequences of the ordering of the lattice gauge theories 
on the perturbative stability of the code Hamiltonian.
The two can be related by discretizing the imaginary time path integral of the quantum Hamiltonian, as we detail in Sec.~\ref{sec:psucc_to_observable}.
We further discuss the ``memory experiment'' in Sec.~\ref{sec:memory_expr} and relate it to the energy splitting between ground states;
and the ``stability experiment'' in Sec.~\ref{sec:stability_hard_wall}, \ref{sec:stability_hard_wall_two_punctures}, and relate it to the gap to excited states.
The latter experiment involves a slight modification of the Hamiltonian (see Eq.~\eqref{eq:H_excluded_W}), which can be avoided for 2D toric codes in Euclidean or hyperbolic spaces, as we detail in Appendix~\ref{sec:stability_experiment_special_to_surface_codes}.

\subsection{Transfer matrix representation and imaginary time evolution under a quantum Hamiltonian \label{sec:psucc_to_observable}}

To connect classical stat mech and the quantum code Hamiltonian,
it is suggestive to write Eq.~\eqref{eq:Z_clean_spin_model} in terms of a transfer matrix acting on an auxiliary Hilbert space of $\sigma$ spins,
\begin{align}
    Z(\tf) \propto&\, \sum_{\sigma_{\tf}}  \bra{\sigma_\tf} 
    \Pi_\pauliz \cdot T(\tf) \ket{\mathbf{0}}.
\end{align}
Here, the boundary conditions come from Eq.~\eqref{eq:cleanSPAM_boundary_condition}.
At $\tf$, the condition $\km \to \infty$ is imposed through the projector $\Pi_\pauliz$, where we define
\begin{align}
    \label{eq:projector_Z}
    \Pi_\pauliz =&\ \prod_{\qs} \frac{1+\pauliz_\qs}{2}, \\
    \label{eq:projector_X}
    \Pi_\paulix =&\ \prod_{\widetilde{\qs}} \frac{1+\paulix_{\widetilde{\qs}}}{2}.
\end{align}
Also observe that 
\begin{align}
    \sum_{\sigma_{\tf}}  \ket{\sigma_\tf} = \Pi_\paulix \sum_{\sigma_{\tf}}  \ket{\sigma_\tf} = \Pi_\paulix^2 \sum_{\sigma_{\tf}}  \ket{\sigma_\tf}.
\end{align}
With these, we have the following rewriting
\begin{equation}
\begin{split}
\label{eq:discrete_time_path_integral}
    Z(\tf) \propto&\, \sum_{\sigma_{\tf}}  \bra{\sigma_\tf} 
    \Pi_\pauliz
    \Pi_\paulix
     \cdot T(\tf) 
    \cdot
    \Pi_\paulix
    \ket{\mathbf{0}} \\
    T(\tf) =&\, \prod_{t=\ti}^{\tf-1} e^{\km \sum_\qs \prod_{\bfr \in \qs} \pauliz_{\bfr, t} } \cdot e^{\overline{\kbf} \sum_\bfr \paulix_{\bfr, t} } \\
    \equiv&\ (T)^{\tf},
\end{split}
\end{equation}
where $\tanh \overline{\kbf} = e^{-2\kbf}$.
Here, we used that $[T(\tf), \Pi_\paulix] = 0$.
Written this way, the transfer matrix takes the same form of the discrete imaginary-time path integral of the following quantum Hamiltonian,
\begin{align}
    \label{eq:H_quantum}
    H_{\rm q} \equiv 
    - \sum_\qs \prod_{\bfr \in \qs} \pauliz_{\bfr} 
    -\sum_{\widetilde{\qs}} \prod_{\bfr \in \widetilde{\qs}} \paulix_{\bfr} 
    -
    h  \sum_\bfr \paulix_{\bfr}.
\end{align}
When $h=0$, we recover the unperturbed code Hamiltonian Eq.~\eqref{eq:stabilizer_code_hamiltonian}; and we are interested the stability of this Hamiltonian for a small but finite $h$.
Its path integral takes the form
\begin{align}
    \label{eq:continuous_time_path_integral}
    e^{-\beta H_{\rm q}} =&\ 
    \Pi_\paulix
    \left[
    \lim_{\delta \tau \to 0}\prod_{t = 0}^{\beta/\delta \tau} 
    e^{\delta \tau \sum_\qs \prod_{\bfr \in \qs} \pauliz_{\bfr, t} } \cdot e^{h \delta \tau \sum_\bfr \paulix_{\bfr, t} } \right]
    \Pi_\paulix,
\end{align}
Here, as the $X$ stabilizers $\widetilde{\qs}$ commute with all other terms, we treat them as constraints (realized by projectors $\Pi_\paulix$).

Comparing Eqs.~(\ref{eq:discrete_time_path_integral}, \ref{eq:continuous_time_path_integral}), we see that the continuous time path integral is 
\YL{precisely}
a limit of $Z(\tf)$, where $\km \to 0, \kbf \to \infty$, while $h^{-1} = e^{2\kbf} \km$ and $\beta \equiv \km \tf$ are kept fixed.
\YL{Therefore, our results in the previous section on the thermal stability of $Z(\tf)$ can be borrowed to argue for the perturbative stability of $H_{\rm q}$.}

In the remainder of this section, we translate stability results for $Z(\tf)$ to statements about the spectrum of 
\YL{the transfer matrix (or more precisely, of the quantum Hamiltonina $\widetilde{H}_{\rm q} \equiv - \frac{1}{\tf} \ln T(\tf)$)}.
With the assumption that these results persist when the continuous time limit is taken,  stability of ordering in the classical stat mech model \YL{at nonzero temperature} provides evidence to the stability of the $H_{\rm q}$ under the perturbation of a transverse field.

The correspondence between the imaginary-time path integral and the stat mech model in one higher dimension is known as ``quantum-to-classical mapping''~\cite{fradkinsusskind1978, kogut1979rmp, kogut1983rmp, polyakov1987gauge, Sondhi1997}.
In many examples, including repetition codes and 2D toric codes (see Sec.~\ref{sec:examples_and_gauge_invariance}), the phase diagrams of $Z(\tf)$ and of $H_{\rm q}$ are both known~\cite{fradkinsusskind1978, trebst2007, tupitsyn2010}, and our assumption about the continuous time limit is justified.
In general, however, rigorous justifications will have to be made separately: \YL{in the extreme anisotropic limit corresponding to the continuous time limit where $\km \to 0, \kbf \to \infty$, conditions in Theorem~\ref{thm:ordering_of_clean_stat_mech_model} (requiring both $\km,  \kbf \geq \mathsf{K}_{\rm th}$)  are no longer satisfied.}
\YL{We discuss these difficulties in detail in Appendix~\ref{sec:continuous_time_limit}.}

The connection we make here is therefore heuristic but not rigorous.
Nevertheless, it is still possible that perturbative stability of the quantum Hamiltonian can follow from the existence of a low temperature phase of the stat mech model (Theorem~\ref{thm:ordering_of_clean_stat_mech_model}), without going through the random-signed models on the Nishimori line.
We leave this possibility for future work.

We will put aside the issue of the continuous time limit, and treat $\widetilde{H}_{\rm q} = H_{\rm q}$ for the rest of this section.

\subsection{Memory experiment and ground state splitting \label{sec:memory_expr}}

We consider physical consequences of the success of memory experiments.
\YL{Recall that the memory experiment comprises precisely of the error model and the inference problems~\ref{IP1},\ref{IP2} we discussed in Sec.~\ref{sec:stat_mech_model}, see also footnote~\ref{fn:intro_memory_stability_expr}.}
We first note that
\begin{align}
\begin{split}\label{eq:boundary_state_projection}
    \Pi_\paulix
    \ket{\mathbf{0}} \propto&\ \ket{\mathbf{0}^\pauliz_L}, \\
    \Pi_\paulix
    \Pi_\pauliz
    \sum_{\sigma_{\tf}}  \ket{\sigma_\tf} \propto&\ \sum_{\bm{\ell}\in \{0,1\}^{[K]}} \ket{\bm{\ell}^\pauliz_L}.
\end{split}
\end{align}
Here, $\ket{\bm{\ell}^\pauliz_L}$ is a ground state of the unperturbed Hamiltonian, which is simultaneously an eigenstate of all logical $\pauliz$ operators.
It satisfies
\begin{align}
    L_J \ket{\bm{\ell}^\pauliz_L} = 
    (-1)^{\sum_{j \in J} \bm{\ell}_j}
    \ket{\bm{\ell}^\pauliz_L}.
\end{align}
Using this representation, $Z(\tf)$ from Eq.~\eqref{eq:discrete_time_path_integral} can be rewritten as
\begin{align}
    Z(\tf) \propto 
    \sum_{\bm{\ell}\in \{0,1\}^{[K]}}  
    \bra{\bm{\ell}^\pauliz_L}
    T(\tf)
    \ket{\mathbf{0}^\pauliz_L}.
\end{align}
It follows from
Appendix~\ref{sec:zbasis} that
\begin{align}
    \frac{
    \bra{\bm{\ell}^\pauliz_L}
    T(\tf)
    \ket{\mathbf{0}^\pauliz_L}}
    {
    \sum_{\bm{\ell'} \subseteq [K]}
    \bra{\bm{\ell'}^\pauliz_L}
    T(\tf)
    \ket{\mathbf{0}^\pauliz_L}}
    \to
    \begin{cases}
        1, \quad&\bm{\ell} = \mathbf{0}\\
        0, \quad&\bm{\ell} \neq \mathbf{0}
    \end{cases}.
\end{align}
The tunneling amplitude between different unperturbed logical $\pauliz$ states is therefore suppressed.
This can be contrasted with the soluble limit $h \to \infty$, where this ratio is independent of $\bm{\ell}$.

A more explicit dependence on the spectrum of the perturbed Hamiltonian can be drawn by going to the logical $\paulix$ basis.
This is because the perturbations are in the $\paulix$ basis, and the logical $\paulix$ operators still commute with the perturbed Hamiltonian.
Starting from Eq.~\eqref{eq:discrete_time_path_integral}, we represent $\avg{L_J}$ with transfer matrices as follows,
\begin{align}
    \avg{L_J} = \frac{
        \sum_{\sigma_{\tf}}  \bra{\sigma_\tf} 
        \Pi_\pauliz
        \Pi_\paulix
        L_J
         \cdot T(\tf) 
        \cdot
        \Pi_\paulix
        \ket{\mathbf{0}} 
    }
    {
        \sum_{\sigma_{\tf}}  \bra{\sigma_\tf} 
        \Pi_\pauliz
        \Pi_\paulix
         \cdot T(\tf) 
        \cdot
        \Pi_\paulix
        \ket{\mathbf{0}} 
    }.
\end{align}
Here, we noted that $L_J$ commutes with $\Pi_{\pauliz, \paulix}$.
With
\begin{align}
\begin{split}\label{eq:boundary_state_projection}
    \ket{\mathbf{0}^\pauliz_L} \propto&\ 
    \sum_{\bm{\ell}\in \{0,1\}^{[K]}} \ket{\bm{\ell}^\paulix_L}, \\
    \ket{\mathbf{0}^\paulix_L}
    \propto&\ 
    \sum_{\bm{\ell}\in \{0,1\}^{[K]}} \ket{\bm{\ell}^\pauliz_L},
\end{split}
\end{align}
we arrive at
\begin{align}
\label{eq:LJ_X_logical_basis}
    \avg{L_J} =&\ 
    \frac{
    \sum_{\bm{\ell}\in \{0,1\}^{[K]}}
    \bra{\mathbf{0}^\paulix_L} 
    L_J
    T(\tf)
    \ket{\bm{\ell}^\paulix_L}
    }
    {
    \sum_{\bm{\ell}\in \{0,1\}^{[K]}}
    \bra{\mathbf{0}^\paulix_L} 
    T(\tf)
    \ket{\bm{\ell}^\paulix_L}
    } \nn
    =&\ 
    \frac{
    \sum_{\bm{\ell}\in \{0,1\}^{[K]}}
    \bra{\bm{J}^\paulix_L}
    T(\tf)
    \ket{\bm{\ell}^\paulix_L}
    }
    {
    \sum_{\bm{\ell}\in \{0,1\}^{[K]}}
    \bra{\mathbf{0}^\paulix_L} 
    T(\tf)
    \ket{\bm{\ell}^\paulix_L}
    } \nn
    =&\ 
    \frac{
    \bra{\bm{J}^\paulix_L}
    T(\tf)
    \ket{\bm{J}^\paulix_L}
    }
    {
    \bra{\mathbf{0}^\paulix_L} 
    T(\tf)
    \ket{\mathbf{0}^\paulix_L}
    }.
\end{align}
Here
\begin{align}
    \ket{\bm{J}^\paulix_L} \equiv L_J \ket{\mathbf{0}^\paulix_L} = \prod_{j \in J} L_j \ket{\mathbf{0}^\paulix_L}
\end{align}
is the logical state obtained from $\ket{\mathbf{0}^\paulix_L}$ by the $\pauliz$ logical operator $L_J$, and in the last step, we used that the $\paulix$ logical operators commute with $T(\tf)$, so only terms diagonal in the logical space are nonzero.
Recall that Theorem~\ref{thm:ordering_of_clean_stat_mech_model} states that $\avg{L_J} \to 1$ for all $J \subseteq [K]$.

We now expand Eq.~\eqref{eq:LJ_X_logical_basis} in an eigenbasis of $T(\tf)$ (or equivalently, that of $H_{\rm q}$).
Since $H_{\rm q}$ commutes with all $\paulix$ logical operators, its eigenstates are simultaneous eigenstates of $\paulix$ logicals.\footnote{With zero perturbation, these ground states are precisely $\ket{\bm{J}_L^\paulix}$, separated by a gap of $O(1)$ to the first excited states in the same sector.}
We denote by $\Pi^J_\Omega$ the projector onto the perturbed ground state of $H_{\rm q}$ labeled by $J$, and obtain
\begin{align}
    \label{eq:LJ_contributions_ground_excited_states}
    \avg{L_J}
    =&\ 
    \frac{
    \bra{\bm{J}^\paulix_L}
    e^{-\beta E_0^J} \Pi^J_\Omega
    + (1-\Pi^J_\Omega) e^{-\beta H_q} (1-\Pi^J_\Omega)
    \ket{\bm{J}^\paulix_L}
    }
    {
    \bra{\bm{0}^\paulix_L}
    e^{-\beta E_0^{0}} \Pi^0_\Omega
    + (1-\Pi^0_\Omega) e^{-\beta H_q} (1-\Pi^0_\Omega)
    \ket{\bm{0}^\paulix_L}
    }.
\end{align}
Here $E_0^{0,J}$ denote the ground state energy of the perturbed ground states in sectors $0$ and $J$, respectively.
We now make the important but \textit{nonrigorous} \textit{approximation} that contributions to the matrix elements from excited states can be neglected at large $\beta$, particularly for $\beta = \km \tf = O(e^{D(N)})$.
Although we do not have a rigorous justification for this,
we note that the approximation is often valid when there is a constant gap $\Delta$ to all excitations as $N\to \infty$, as their contributions will be suppressed by factors of $e^{-\beta \Delta}$.
In the next subsection, we discuss how evidence for a gap can be obtained from decoding experiments.

As we established earlier 
for codes that satisfy LDPC conditions and with a sufficiently large code distance (see Eq.~\eqref{eq:LDPC_condition}), we have for sufficiently large $\km, \kbf$
\begin{align}
    \lim_{N\to\infty} \avg{L_J} \to 1, \quad \text{for } \beta \propto \tf = O(e^{D(N)}).
\end{align}
With the above assumptions, for $\beta \propto \tf = O(e^{D(N)})$ we have
\begin{align}
    \avg{L_J}
    \approx 
    e^{-\beta \delta E} 
    \frac{
        \bra{\bm{J}^\paulix_L}
        \Pi^J_\Omega
        \ket{\bm{J}^\paulix_L}
    }
    {
        \bra{\bm{0}^\paulix_L}
        \Pi^0_\Omega
        \ket{\bm{0}^\paulix_L}
    }
    \to 1,
\end{align}
\YL{where $\delta E \equiv E_0^J - E_0^0$.} 
We have the same limiting behavior upon replacing $\beta \to \beta/2$.
Taking the ratio between the two, we obtain
\begin{align}
    \label{eq:GS_splitting_upperbound}
    e^{-\beta \delta E}  \to 1 \quad \Rightarrow \quad |\delta E| = O(\beta^{-1}) = O(e^{-D(N)}) 
\end{align}
Therefore, success of the memory experiment corresponds to a small splitting between ground states.
The memory time $\tf \propto \beta = O(e^{D(N)})$ of the quantum code has the intuitive physical meaning of inverse the splitting.
Eq.~\eqref{eq:GS_splitting_upperbound} agrees what one would naively expect from perturbation theory.

It is useful to appreciate this scaling in the above examples of the repetition code and the toric code.
The RBIM and RPGM become the usual 2D Ising model and 3D $\mathbb{Z}_2$ gauge theory once we remove the random sign disorder.
In both cases, the free energy of the disordering defect (domain wall in the Ising model, flux loop in the gauge theory, see Fig.~\ref{fig:Ising},\ref{fig:toric-code-flux-loop}) has a nonzero line tension, and therefore an extensive free energy cost.
\YL{Formally, we have $e^{-\delta F} \equiv Z_{J=-1} / Z_{J=+1}$, and}
\begin{align}
    \avg{L_J} = \frac{1-e^{-\delta F}}{1+e^{-\delta F}} \to 1, \ \text{ when } \delta F = O(D(N)).
\end{align}
For these codes, $D(N)$ is proportional to the linear dimension of the system.
However, as $\tf$ increases, the defect gains a translational entropy in the temporal direction \YL{(which has extent $\beta$)}
\begin{align}
    \delta F \to \delta F - \ln \beta.
\end{align}
The free energy eventually vanishes when $\beta \propto \tf$ exceeds $O(e^{D(N)})$, whence the defects proliferate, and the quantum memory is lost.
Indeed, in the $\tf \to \infty$ limit the stat mech model becomes effectively 1-dimensional, and no long-range correlation can be established between its two temporal boundaries at finite temperature.

\subsection{``Hard wall'' stability experiment: example with the 2D toric code \label{sec:stability_hard_wall}}

We have seen that the memory experiment probes the proliferation of space-like defects that form logical errors.
Here we introduce a closely related decoding experiment as a spacetime dual to the memory experiment, which probes the proliferation of time-like defects.
This is in the same spirit to the so-called ``stability experiment'' proposed by Gidney~\cite{Gidney_2022}, see also~\cite{kesselring2022anyon, campbell2024stability}.
This experiment can be readily carried out numerically whenever decoding algorithms are available for the code.

We first introduce a ``hard wall'' version of the stability experiment within the familiar example of the 3D $\mathbb{Z}_2$ gauge theory from the 2D toric code to illustrate the ideas before generalizing to other cases.
Let $W$ be a loop of qubits in the code patch, and define the following Pauli string operator (abusing notation)
\begin{align}
    \PRXQ{L_W} = \prod_{\bfr \in W} \pauliz_\bfr.
\end{align}
Notice that we choose $\PRXQ{L_W}$ so that it is within the stabilizer group.
Consider our error model in Sec.~\ref{sec:stat_mech_model}, which we execute for a long time $(-\infty, 0)$.
At time $t=0$, we turn off the bitflip errors on every qubit in $W$, but everywhere else the bitflip and measurement error rates are not changed.
Then we run the modified dynamics for $t \in [0, \tf]$.\footnote{Here we take $\tf$ to be a free parameter that can be varied. In particular, we do not require it to scale as $e^{D(N)}$.}
Afterwards, we restore bitflip errors on $W$, and run the original dynamics indefinitely, for $(\tf, \infty)$.
The task is to predict the value of $W$ at any point $t \in [0, \tf]$ given the measurement results for all times and prior knowledge of the error rates. 
Note that its value is conserved during $t \in [0, \tf]$, as errors are turned off on $W$.
As $\PRXQ{L_W}$ lies within the stabilizer group, it is in principle possible to recover its value in the limit of perfect measurements, by multiplying stabilizers that are enclosed by the loop.

Our discussions in Sec.~\ref{sec:psucc_to_observable} were developed for $L_J$, but they immediately carry over to $W$.
Using these correspondences, the success probability of predicting $W$ lower bounds $\langle \PRXQ{L_W} \rangle$ in the lattice gauge theory.
Turning off bitflip errors on $W$ amounts to setting the corresponding $\kbf$ to infinity in the lattice gauge theory, which forces $\sigma_{\bfr, t} = \sigma_{\bfr, t+1}$ for all $t \in [0, \tf]$, $\bfr \in W$.
This is in close parallel to the memory experiment, where $\km$ couplings at $t=\tf$ are sent to infinity (see Eq.~\eqref{eq:SPAM_boundary_condition}), \YL{with the effect of promoting a ``perimeter law'' of the Wilson loop to a ``1-law''.}
\YL{Similarly, the problem of predicting $\PRXQ{L_W}$ can also be thought of as evaluating an extension of the Wilson loop.
In the absence of the hard wall (as realized by $\tf \to 0$), $\langle \PRXQ{L_W} \rangle$ becomes the usual Wilson loop operator, and exhibits a perimeter law scaling, namely $\langle \PRXQ{L_W} \rangle \propto e^{-|W|}$.
Extending the Wilson loop to a ``hard wall'' of height $\tf$ suppresses the flux loops and promotes the perimeter law to a ``1-law'', allowing mapping to error correction success probabilities.
}

\begin{figure}
    \centering
    \includegraphics[width=1.0\linewidth]{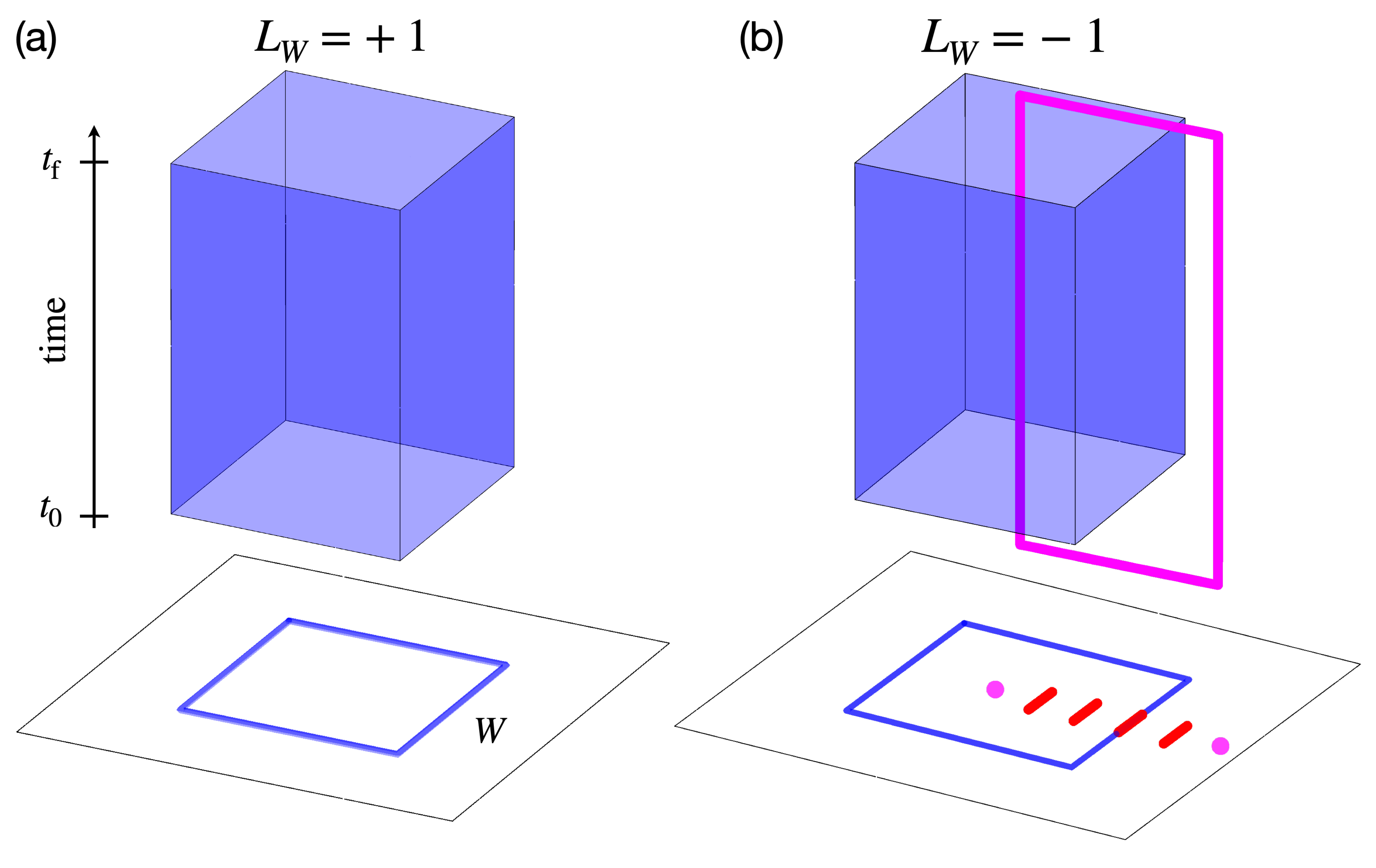}
    \caption{``Hard wall'' stability experiment for the 2D toric code.
    We focus on a large but finite loop $W$, and assume the code patch itself is infinite.
    Bitflip errors on $W$ are turned off for $t \in [\ti, \tf]$, and the decoder is tasked to predict its value, which is conserved during this time.
    The success probability translates into $\avg{\PRXQ{L_W}}$ in the $\mathbb{Z}_2$ lattice gauge theory in 3D, where the couplings strengths $\kbf$ are sent to $+\infty$ on $\bfr \in W$ for $t \in [\ti, \tf]$, resulting in a ``hard wall'', which we represent with a blue surface.
    (a) shows a representative configuration where $\PRXQ{L_W} = +1$, and (b) shows one where $\PRXQ{L_W} = -1$, with a single flux loops wrapping around the hard wall of length at least $O(\tf)$, as the hardwall repulses flux loops. 
    The configuration in (b) can be understood as a pair of anyons straddling $W$ at $t < \ti$, and remain present for $t \in [\ti, \tf]$; they are forbidden by the hard wall to annihilate.
    In both figures we omit drawing contractible flux loops, which are generically present.
    }
    \label{fig:stability-experiment}
\end{figure}

As we illustrate in Fig.~\ref{fig:stability-experiment},
$\avg{\PRXQ{L_W}}$ receives contributions from two classes of configurations, namely those with an even or odd number of flux loops threading through the cylindrical ``hard wall'' of infinite couplings.
The infinite $\kbf$ coupling on the wall prevents the flux loops to pass through it, and they have to go around the wall.
We can again write
\begin{align}
    \avg{\PRXQ{L_W}} = \frac{1 - e^{-\delta F}}{1 + e^{-\delta F}},
\end{align}
with $\delta F$ the free energy difference between the two classes.
In the deconfined phase, flux lines again have a finite line tension, and we have $\delta F \propto \tf$.
The line tension are conventionally identified with the gap to anyons in the associated quantum Hamiltonian $\widetilde{H}_{\rm q}$. \YL{(We make this point more explicit in Sec.~\ref{sec:stability_hard_wall_two_punctures}.)}
The line tension (or gap) vanishes at a finite temperature critical point, where flux loops proliferate and $\langle \PRXQ{L_W} \rangle \to 0$, signifying the failure of decoding $\PRXQ{L_W}$.
Comparing Fig.~\ref{fig:toric-code-flux-loop} and Fig.~\ref{fig:stability-experiment}, we see that the space-like and time-like defects are not so different for the 2D toric code: due to the isotropy of 3D model, the latter is simply a contractible version of the former, and which stretches in time.
Thus, one expects that the memory and stability experiments fail at the same phase transition.

In general, \YL{the stat mech model is no longer isotropic (cf. Fig.~\ref{fig:summary}), and the two transitions need not coincide.}
In particular, they can be separated by introducing columnar or planar disorder, such as in the McCoy-Wu model~\cite{McCoy-Wu, dsfisher1992, dsfisher1995}, where a Griffiths phase is present.
Successes of the memory and stability experiments are therefore expected to be independent of each other in sufficiently general models, and one should not hope to infer one from the other under general assumptions.
The small energy splitting between ground states and the nonzero gap to excitations of the perturbed model will have to be established separately.

More generally, beyond the 2D toric code example, an explicit upper bound to the failure probability (similar to Eq.~\eqref{eq:pfail_upperbound_memory_maintext}) of the stability experiment can also be obtained (see Eq.~\eqref{eq:pfail_hard_wall_stability_dennis} of Appendix~\ref{sec:pfail_upperbound_MW}), 
\begin{align}
    \label{eq:pfail_upperbound_stability_maintext}
    \epsilon(\tf) = O\left( \mathrm{area}(W)  \cdot [(w+1) \cdot \widetilde{q}]^{\tf} \right),
\end{align}
where $\mathrm{area}(W)$ is defined to be the minimum number of terms in a decomposition of $W$ into stabilizers.
The equation holds when $\widetilde{q}$ (defined in Eq.~\eqref{eq:def_q_tilde}) is sufficiently small.
Our reasoning in Sec.~\ref{sec:psucc_to_observable} implies that in the 3D gauge theory with a hard wall on $W \times [0, \tf]$, 
\begin{align}
\label{eq:temporal_decoding_lower_bound_2DTC}
    &\langle \PRXQ{L_W} \rangle \geq 1-2 \epsilon(\tf)\nn
    &\quad \Rightarrow \quad
    e^{-\delta F} = O(\epsilon(\tf)) = O(\mathrm{area}(W) \cdot e^{-\xi^{-1} \tf}).
\end{align}
Here, the prefactor $\mathrm{area}(W)$ comes from the entropy of the location of the flux loop, but is independent of $\tf$.
The ``inverse correlation length'' $\xi^{-1} = -\ln[(w+1) \cdot \widetilde{q}] > 0$ can be understood (see Sec.~\ref{sec:stability_hard_wall_two_punctures} for more details) as a line tension, or a finite gap to the magnetic charge, as is consistent with the form we expect from our discussions above of the deconfined phase.

\subsection{``Hard wall with two punctures'' and temporal correlation functions for general CSS codes \label{sec:stability_hard_wall_two_punctures}}

\begin{figure}[t]
    \centering
    \includegraphics[width=0.48\linewidth]{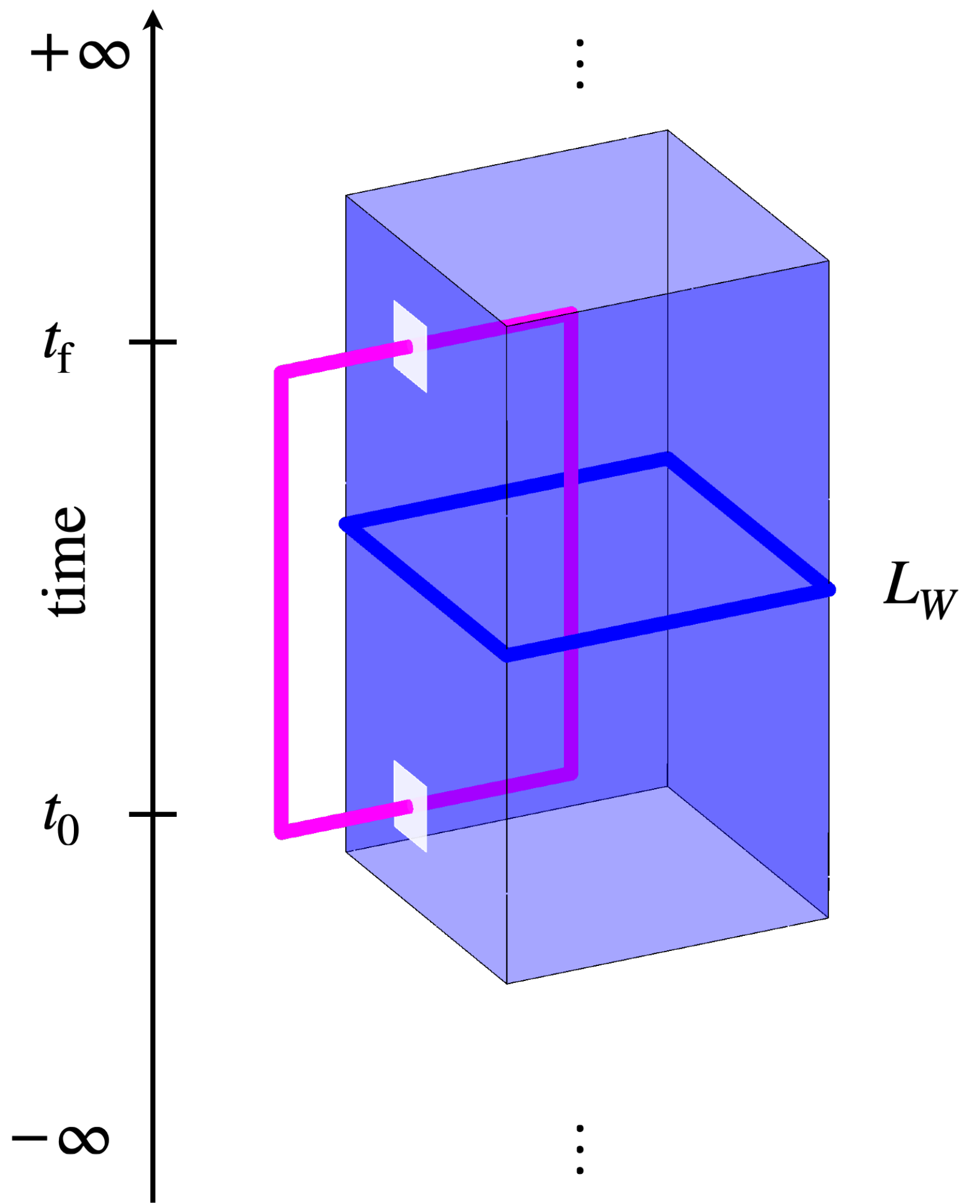}
    \caption{``Hard wall with two punctures'' stability experiment for the 2D toric code.
    The hard wall is present from $(-\infty, \infty)$, except that at times $t=0$ and $t=\tf$ the bitflip error (with small but constant probability $p_0$) is turned on at a chosen qubit $0 \in W$.
    The leading contribution that leads to $\PRXQ{L_W} = -1$ is flux loop that go through both punctures.}
    \label{fig:stability-experiment-2pt}
\end{figure}

We now move on to general discussion of LDPC CSS codes.
For the moment, we only require that $W$ is an element of the stabilizer group, without specifying the size of $W$.
To facilitate generalization, we consider a slightly different version of the stability experiment, which we call ``hard wall with two punctures''.
Consider running the modified dynamics (with bitflip errors on $W$ turned off) from $(-\infty, \infty)$, except that at times $t=0$ and $t=\tf$ the bitflip error (with a small but constant probability $p_0$) is turned on at a given qubit $0 \in W$.
We illustrate this for the 2D toric code in Fig.~\ref{fig:stability-experiment-2pt}.

The success probability of predicting $W$ can again be associated with the expectation value $\langle \PRXQ{L_W} \rangle$ in a stat mech model without disorder, along the lines of Sec.~\ref{sec:ordering_of_LGT} and Sec.~\ref{sec:psucc_to_observable}.
The transfer matrix of the stat mech model is generated by $H^W_{\rm q}$, where
\begin{align}
\label{eq:H_excluded_W}
    H^W_{\rm q} 
    \equiv&\,
    - \sum_\qs \prod_{\bfr \in \qs} \pauliz_{\bfr} 
    -\sum_{\widetilde{\qs}} \prod_{\bfr \in \widetilde{\qs}} \paulix_{\bfr} 
    -
    h  \sum_{\bfr \notin W} \paulix_{\bfr} \nn
    =&\ H_{\rm q} + h \sum_{\bfr \in W} \paulix_\bfr.
\end{align}
By design, $\PRXQ{L_W}$ is a symmetry of $H^W_{\rm q}$, and in particular ground states of $H^W_{\rm q}$ are even under $\PRXQ{L_W}$.

\YL{This setup is advantageous as compared to the one in Sec.~\ref{sec:stability_hard_wall}, as we can directly relate $\langle \PRXQ{L_W} \rangle$ to certain two-point correlation functions of operators separated in imaginary time. 
With details in Appendix~\ref{sec:two_punctures_two_points}, we show that
\begin{align}
\label{eq:W_temporal_corr}
    \langle \PRXQ{L_W} \rangle
    \leq&\, 1 - (\tanh \lambda)^2 \langle \paulix_0(\tf) \paulix_0(0) \rangle,
\end{align}
where $\tanh \lambda \equiv \frac{p_0}{1-p_0}$, and
\begin{align}
\label{eq:spectral_decomp_2pt_corr_fcn}
    \langle \paulix_0(\tf) \paulix_0(0) \rangle
    \equiv
    \frac{\tr (\Pi_\Omega^{W} \cdot \paulix_0 \cdot  e^{-H^W_{\rm q} \tf} \cdot \paulix_0)}{\tr (\Pi_\Omega^{W} \cdot e^{-H^W_{\rm q} \tf})}.
\end{align}
By construction, $\paulix_0$ anticommutes with $\PRXQ{L_W}$, and creates an excited state with $\PRXQ{L_W} = -1$.
The decay rate of this two-point function is therefore generically determined by the lowest-lying state within the $\PRXQ{L_W} = -1$ sector, whose gap we denote as $\Delta$.

Meanwhile, we can explicitly lower bound $\avg{\PRXQ{L_W}}$ using success probabilities of decoding for general LDPC codes (see Eq.~\eqref{eq:pfail_hard_wall_stability_two_punctures_dennis})
\begin{align}
    \label{eq:W_temporal_Psucc}
    1 - \langle \PRXQ{L_W} \rangle =  O(e^{-\xi^{-1} \cdot \tf}).
\end{align}
Comparing the previous two equations, we find a nonzero decay rate of the two point function,
\begin{align}
    \langle \paulix_0(\tf) \paulix_0(0) \rangle = O( e^{-\xi^{-1} \cdot \tf}).
\end{align}
Therefore, $\xi^{-1}$ provides a nonzero lower bound to $\Delta$.}

\begin{figure}
    \centering
    \includegraphics[width=.9\linewidth]{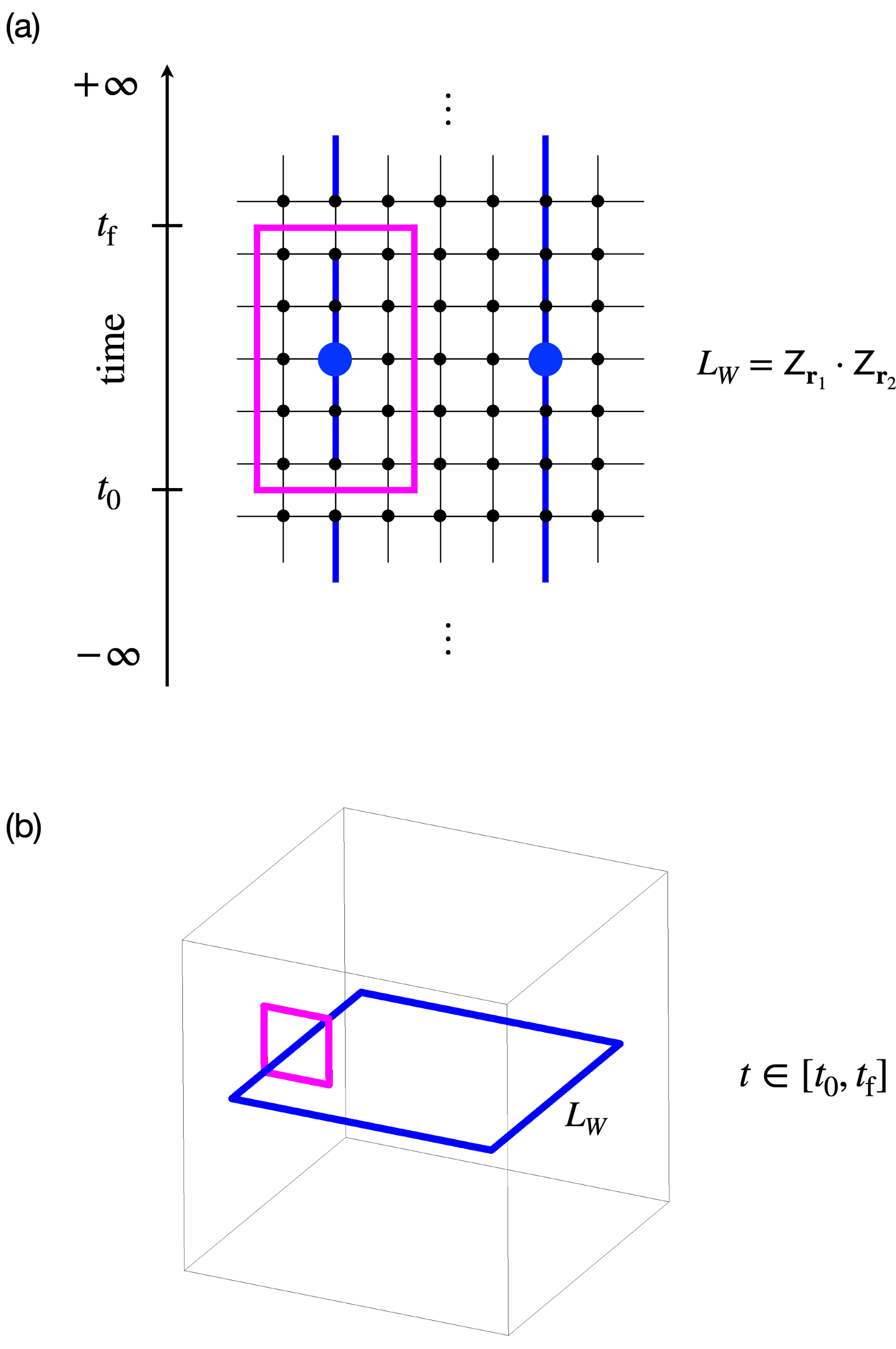}
    \caption{``Hard wall with two punctures'' stability experiment for (a) the 1D repetition code and (b) the 3D toric code.
    For (a), we take $\PRXQ{L_W} = \pauliz_{\bfr_1} \pauliz_{\bfr_2}$, which is a product of stabilizers between $[\bfr_1, \bfr_2]$.
    The disordering defect is a closed domain wall (magenta).
    For (b), we are only drawing the configuration of the system at a given time instant $t \in [\ti, \tf]$.
    $W$ is taken as a product of $\pauliz$-plaquette stabilizers, which form a closed loop.
    It detects the gap of magnetic loop-like excitations (magenta) created by Pauli $\paulix$ and braid nontrivially with $\PRXQ{L_W}$.
    }
    \label{fig:hardwall-2pt-1DREP-3DTC}
\end{figure}

We note that in the above discussions, $\Delta$ is defined for the Hamiltonian $H_{\rm q}^{W}$ where perturbations are turn\PRXQ{ed} off on $W$, rather than for $H_{\rm q}$.
What can we say about the spectrum of $H_{\rm q}$ from the above results for $H^W_{\rm q}$?
We point out two general difficulties in giving a satisfactory answer to this question for general CSS codes.

(i)
With a fixed $W$, the decoding problem probes those excited states with $\PRXQ{L_W} = -1$.
Thus, different choices of $W$ and different creation operators (say, those acting on more than one qubit) are needed in order to probe different excitations across the entire system, particularly for codes where translation symmetry is lacking.
The situation is further complicated if quenched disorder are present in the transfer matrix.
These are beyond the current formulation of this experiment.

(ii)
There are conflicting requirements we want to impose on $\PRXQ{L_W}$ and $W$.
On the one hand, we want $|W|$ to be small compared to $N$, and ideally $H^W_{\rm q}$ and $H_{\rm q}$ are only different on a vanishing fraction of qubits, so that lessons from $H^W_{\rm q}$ are hopefully informative of $H_{\rm q}$.
On the other hand, $\mathrm{area}(W)$ needs to be sufficiently large, so that the experiment is not probing local ``bound states'' near the defect $W$, while strong perturbations elsewhere have already closed the gap.\footnote{The simplest example of a bound state is a transverse field Ising model.
This is the relevant model to the 1D repetition code.
As we illustrate in Fig.~\ref{fig:hardwall-2pt-1DREP-3DTC}, $\PRXQ{L_W}$ takes the form of $\pauliz_{\bfr_1} \pauliz_{\bfr_2}$.
In the extreme case, when $\bfr_1, \bfr_2$ are only one lattice spacing apart, the absence of a transverse field on $\bfr_1, \bfr_2$ leads to formation of a locally ferromagnetic region, which is gapped regardless of the rest of the system.}
For example, for $\mathrm{area}(W) = O(1)$, a majority vote decoder gives a finite lower bound to $\xi^{-1}$ everywhere in the phase diagram, even when the perturbation strength/error rate is large.\footnote{\label{fn:majority_vote}As $\tf$ increases one eventually enters a regime where decoding has a success probability approaching $1$, even when the physical error rates are approaching $1/2$.
Suppose $W$ can be written as the product of by as few as $\mathrm{area}(W)$ checks.
We consider a decoder that takes the product of these checks at every time step, and then takes the majority vote of the products for $t \in [0, \tf]$.
It is easy to see that
\begin{align}
\label{eq:majority_vote_decoder}
    \langle \PRXQ{L_W} \rangle \geq \mathbb{P}_{\rm succ} \sim 1 - e^{-e^{-\mathrm{area}(W)} \tf}, \quad
    \text{ when } \tf \gg e^{\mathrm{area}(W)}.
\end{align}
This puts a lower bound on the inverse correlation length, $\xi^{-1} \geq e^{-\mathrm{area}(W)}$.
In the case of 2D toric code, $\mathrm{area}(W) \propto |W|^2$, and this is lower bound that vanishes with increasing $W$: despite the success of decoding, there is no finite gap to anyon excitations, as consistent with our expectation as the temperature goes to infinity.
However, when there are local redundancies, we have $\mathrm{area}(W) = O(1)$, there is an $O(1)$ lower bound to $\xi^{-1}$, which we associate to the formation of bound states near $W$.
}
\YL{To quantify these requirements, we propose the following general \textit{necessary} condition that should be satisfied by $W$,
\begin{align}
\label{eq:condition_on_W}
\exists W_N \in \mathcal{S}_N, \quad
    \lim_{N \to \infty} |W_N| / N = 0, \quad \lim_{N \to \infty} \mathrm{area}(W_N) = \infty.
\end{align}
As a simple concrete proposal, we may take $\PRXQ{L_W}$ as the product of $\Theta(N^\gamma)$ independent stabilizers, with $\gamma \in (0,1)$.
This way, $\mathrm{area}(W) = \Theta(N^\gamma)$ and $|W|$ grows no faster than $O(N^\gamma)$, and the above condition is satisfied.
As one can see, this proposal is general, and works in particular for ``good'' LDPC codes.}

\YL{Although the above issues exist when considering general CSS codes,}
for many examples of interest, there are natural choices of $W$ that satisfy Eq.~\eqref{eq:condition_on_W}.
For topological codes in low dimensions, either in flat or curved spaces, we can choose $\PRXQ{L_W}$ to be the product of two $\pauliz$-type logicals that are separated by a distance proportional to the code distance, and whose product lies within the stabilizer group.
We illustrate a few examples in Fig.~\ref{fig:hardwall-2pt-1DREP-3DTC}.
The stability experiment now looks similar to ``lattice surgery'', and has the intuitive interpretation of transporting logical information in space~\cite{Gidney_2022} (rather than in time, as in the memory experiment).
(In the special case of 2D toric code in Euclidean or hyperbolic spaces, the introduction of $\PRXQ{L_W}$ can be avoided, due to a global redundancy, as we discuss in Appendix~\ref{sec:stability_experiment_special_to_surface_codes}.)
More generally, ideas from materialized symmetry~\cite{benbrowndomwilliamson2020parallelized, brown2022conservation} might be useful in constructing $\PRXQ{L_W}$.

Indeed, \YL{for topological codes,} Eq.~\eqref{eq:condition_on_W} is
\YL{naturally satisfied by}
``spatial boundaries'' of a code.
When these conditions are met, one can expect to generally define the code on a lattice with ``open'' boundary conditions on $W$.
The stability experiment would then probe excitations that ``condense'' on such boundaries.
For graphs in finite dimensions with local connectivity, it is often possible to define such open boundary conditions, and $W$ will often be a lower-dimensional manifold, as exemplified by Figs.~\ref{fig:stability-experiment-2pt},\ref{fig:hardwall-2pt-1DREP-3DTC}.
When there is no natural notion of spatial boundary (\YL{such as for codes on graphs}), \YL{abstract constructions of $W$, such as the one proposed below Eq.~\eqref{eq:condition_on_W}, might be necessary}.
\YL{The lack of spatial open boundary conditions} also arises in defining thermodynamic limits for stat mech models on these graphs~\cite{Kovalev_2018, Jiang_2019, rakovszky2023gauge}, and is not unique to our problem.

\section{Discussions \label{sec:discussion}}

With a combination of exact and heuristic results, we argue that the existence of error correction thresholds of CSS codes offers quantitative evidence to their perturbative stabilities as Hamiltonians, thereby connecting the two notions of stability.
Our results suggest that error correction may provide means to define topological ordered phases beyond Euclidean lattices in finite dimensions.
This definition makes use of boundary conditions that are motivated by error correction considerations.
This can be particularly useful when thermodynamic limits are not well defined, whereas finite thresholds are well defined and can be readily established.

Our approach to stability of the perturbed code Hamiltonians is by turning them into classical Ising gauge theories, whose stability is implied by that of their disordered versions, when introducing random signs in the couplings satisfying the Nishimori conditions.
The Nishimori conditions permits an interpretation of the stability in terms of success of inference problems (decoding), which can be established using other methods, without directly bounding order parameters in the stat mech models.
We have thus far proven a low-temperature ordered / deconfined phase of the Ising gauge theories, as defined by the boundary order parameters, see Theorem~\ref{thm:ordering_of_clean_stat_mech_model}.
We are hopeful that these results can be generalized to any finite abelian group.
To make statements about the perturbed Hamiltonian, we make \textit{nontrivial} assumptions that (i) the continuous time limit can be taken without destroying stability (see Appendix~\ref{sec:continuous_time_limit} for further discussions), (ii) contributions to $\avg{L_J}$ in Eq.~\eqref{eq:LJ_contributions_ground_excited_states} from excited states can be neglected, and (iii) a finite gap in the spectrum of the \textit{modified} perturbed Hamiltonian Eq.~\eqref{eq:H_excluded_W} can be used to establish a finite gap in the \textit{unmodified} one.
Stronger mathematical results may be necessary in closing these logical gaps, but are beyond the scope of this work.

Our results here are instead based upon physical observations that the non-unitary time evolution of the code (under a stochastic noise model) ``emulates'' imaginary time evolution.
This allows us to build intuition of QEC from stat mech and vice versa, and to ``measure'' observables of the stat mech models by designing different decoding experiments.
With these, we show that success of \textit{memory} and \textit{stability} experiments corresponds to suppression of defects in space and time, and can be used to argue for (i) small energy splitting within the ground state manifold and (ii) exponentially decaying correlation functions associated with a nonzero gap to local excitations, respectively.

\PRXQ{Quantum imaginary time evolutions are richer than classical stochastic processes, and there are certainly cases where the analogy between error correction threshold and perturbative stability is no longer expected to hold.
For non-CSS codes and more general types of perturbations, the corresponding imaginary time evolution may involve negative weights and Berry's phase terms that are beyond the description of conventional classical statistical mechanics models.\footnote{On a positive note, the classical stat mech models we obtain for maximal-likelihood decoding of these more general cases may be of independent interest, and whose stability results may be readily established~\cite{ChubbFlammia2018}.}
We expect a general stability theorem to be more closely aligned with~\cite{Klich_2010, Bravyi_2010, Bravyi_2011, Michalakis_2013, lavasani2024klocal}, involving very different intuitions and techniques.}

On the other hand, there are also reasons to believe that the analogy we found may have extensions beyond CSS codes under incoherent errors.
Recently, Ref.~\cite{lavasani2024stability} showed that fixed-point wavefunctions can be recovered from the ground state of the perturbed Hamiltonian in the same phase via error correction.


On a technical level, our derivation based on transfer matrices may also be extended beyond CSS codes under uncorrelated $\paulix$ and $\pauliz$ noises. 
We consider two simple examples in Appendix~\ref{sec:examples_beyond}, namely (i) the toric code under pure $\pauliy$ noise and (ii) the Bacon-Shor code (which is a subsystem code with noncommuting checks).
We derive their stat mech models and re-derive known results about their thresholds.
These examples show that a stat-mech picture may be widely applicable to error correction codes and circuits, in particular ``dynamical'' ones with an evolving codespace in time~\cite{gidney2023bacon, hastings2021FloquetCode, davydova2022floquet, bombin2015gaugecolorcode, bombin2015singleshot, KubicaVasmer, bridgeman2023, YL2023STC, stahl2023singleshot, bacon2017sparse, delfosse2023spacetime, McEwen_2023, kesselring2022anyon, Townsend_Teague_2023, bombin2023faulttolerant, bombin2023unifying, beverland2024fault, fugottesman2024error}.
They may also be helpful in systematic identifications of materialized symmetries~\cite{benbrowndomwilliamson2020parallelized, brown2022conservation}.

\section*{Acknowledgements}

We are grateful to Shankar Balasubramanian, Dolev Bluvstein, Nikolas Breuckmann, Sky Cao, Arpit Dua, Michael Gullans, Jeongwan Haah, Israel Klich, Alex Kubica, Ali Lavasani, Pavel Nosov, Akshat Pandey, Benedikt Placke, Leonid Pryadko, Shivaji Sondhi, Adithya Sriram and Dominic Williamson for helpful discussions. We especially thank Tibor Rakovszky for many insightful discussions, and for collaboration on related work and during early stages of this work.

V.K. acknowledges support from the Alfred P. Sloan Foundation through a Sloan Research Fellowship, the Packard Foundation through a Packard Fellowship in Science and Engineering,  the US Department of Energy, Office of Science, Basic Energy Sciences, under Early Career Award Nos. DE-SC0021111 (support for N.O.D., and the study of phases of non-equilibrium error-correction processes), and the Office of Naval Research Young Investigator Program (ONR YIP) under Award Number N00014-24-1-2098 (formulation of information theoretic space-time order parameters to argue for phase stability in stat mech models). 
Y.L. was supported in part by the Gordon and Betty Moore Foundation’s EPiQS Initiative through Grant GBMF8686, in part by the US Department of Energy, Office of Science under Award No DE-SC0019380, and in part by a Q-FARM Bloch Postdoctoral Fellowship at Stanford University.
N.O.D. also acknowledges support from the ARCS Foundation for ARCS Scholar funding.

\bibliography{main}

@article{Arakawa_2005,
   title={Self-duality and phase structure of the 4D random-plaquette gauge model},
   volume={709},
   ISSN={0550-3213},
   DOI={10.1016/j.nuclphysb.2004.12.024},
   number={1–2},
   journal={Nuclear Physics B},
   publisher={Elsevier BV},
   author={Arakawa, Gaku and Ichinose, Ikuo and Matsui, Tetsuo and Takeda, Koujin},
   year={2005},
   month=mar, pages={296–306}
}

@misc{zhao2023errorthresholds,
      title={Extracting Error Thresholds through the Framework of Approximate Quantum Error Correction Condition}, 
      author={Yuanchen Zhao and Dong E. Liu},
      year={2023},
      eprint={2312.16991},
      archivePrefix={arXiv},
      primaryClass={quant-ph},
}

@misc{guoyi2024robust,
      title={Robust teleportation of a surface code and cascade of topological quantum phase transitions}, 
      author={Finn Eckstein and Bo Han and Simon Trebst and Guo-Yi Zhu},
      year={2024},
      eprint={2403.04767},
      archivePrefix={arXiv},
      primaryClass={quant-ph},
}

@article{guoyi2019gapless,
  title = {Gapless Coulomb State Emerging from a Self-Dual Topological Tensor-Network State},
  author = {Zhu, Guo-Yi and Zhang, Guang-Ming},
  journal = {Phys. Rev. Lett.},
  volume = {122},
  issue = {17},
  pages = {176401},
  numpages = {6},
  year = {2019},
  month = {Apr},
  publisher = {American Physical Society},
  doi = {10.1103/PhysRevLett.122.176401},
}

@article{vidal2009phase,
  title = {Low-energy effective theory of the toric code model in a parallel magnetic field},
  author = {Vidal, Julien and Dusuel, S\'ebastien and Schmidt, Kai Phillip},
  journal = {Phys. Rev. B},
  volume = {79},
  issue = {3},
  pages = {033109},
  numpages = {4},
  year = {2009},
  month = {Jan},
  publisher = {American Physical Society},
  doi = {10.1103/PhysRevB.79.033109},
}

@article{vidal2011phase,
  title = {Robustness of a Perturbed Topological Phase},
  author = {Dusuel, S\'ebastien and Kamfor, Michael and Or\'us, Rom\'an and Schmidt, Kai Phillip and Vidal, Julien},
  journal = {Phys. Rev. Lett.},
  volume = {106},
  issue = {10},
  pages = {107203},
  numpages = {4},
  year = {2011},
  month = {Mar},
  publisher = {American Physical Society},
  doi = {10.1103/PhysRevLett.106.107203},
}

@article{fengcheng2012phase,
  title = {Phase diagram of the toric code model in a parallel magnetic field},
  author = {Wu, Fengcheng and Deng, Youjin and Prokof'ev, Nikolay},
  journal = {Phys. Rev. B},
  volume = {85},
  issue = {19},
  pages = {195104},
  numpages = {7},
  year = {2012},
  month = {May},
  publisher = {American Physical Society},
  doi = {10.1103/PhysRevB.85.195104},
}

@misc{nahum2024patching,
      title={Worldsheet patching, 1-form symmetries, and "Landau-star" phase transitions}, 
      author={Pablo Serna and Andrés M. Somoza and Adam Nahum},
      year={2024},
      eprint={2403.04025},
      archivePrefix={arXiv},
      primaryClass={cond-mat.str-el}
}

@article{trebst2007,
  title="{Breakdown of a Topological Phase: Quantum Phase Transition in a Loop Gas Model with Tension}",
  author = {Trebst, Simon and Werner, Philipp and Troyer, Matthias and Shtengel, Kirill and Nayak, Chetan},
  journal = {Phys. Rev. Lett.},
  volume = {98},
  issue = {7},
  pages = {070602},
  numpages = {4},
  year = {2007},
  month = {Feb},
  publisher = {American Physical Society},
  doi = {10.1103/PhysRevLett.98.070602},
  url = {https://link.aps.org/doi/10.1103/PhysRevLett.98.070602}
}

@article{knill1997theory,
title = {Theory of quantum error-correcting codes},
  author = {Knill, Emanuel and Laflamme, Raymond},
  journal = {Phys. Rev. A},
  volume = {55},
  issue = {2},
  pages = {900--911},
  numpages = {0},
  year = {1997},
  month = {Feb},
  publisher = {American Physical Society},
  doi = {10.1103/PhysRevA.55.900},
  url = {https://link.aps.org/doi/10.1103/PhysRevA.55.900}
}

@misc{rakovszky2024product,
    title="{The Physics of (good) LDPC Codes II. Product constructions}",
    author={Tibor Rakovszky and Vedika Khemani},
    year={2024},
    eprint={2402.16831},
    archivePrefix={arXiv},
    primaryClass={quant-ph}
}

@misc{rakovszky2023gauge,
    title="{The Physics of (good) LDPC Codes I. Gauging and dualities}",
    author={Tibor Rakovszky and Vedika Khemani},
    year={2023},
    eprint={2310.16032},
    archivePrefix={arXiv},
    primaryClass={quant-ph}
}

@article{tupitsyn2010,
  title="{Topological multicritical point in the phase diagram of the toric code model and three-dimensional lattice gauge Higgs model}",
  author = {Tupitsyn, I. S. and Kitaev, A. and Prokof'ev, N. V. and Stamp, P. C. E.},
  journal = {Phys. Rev. B},
  volume = {82},
  issue = {8},
  pages = {085114},
  numpages = {5},
  year = {2010},
  month = {Aug},
  publisher = {American Physical Society},
  doi = {10.1103/PhysRevB.82.085114},
  url = {https://link.aps.org/doi/10.1103/PhysRevB.82.085114}
}

@ARTICLE{DKLP2001topologicalQmemory,
       author = {{Dennis}, Eric and {Kitaev}, Alexei and {Landahl}, Andrew and {Preskill}, John},
        title = "{Topological quantum memory}",
      journal = {Journal of Mathematical Physics},
         year = 2002,
        month = sep,
       volume = {43},
       number = {9},
        pages = {4452-4505},
          doi = {10.1063/1.1499754}
}

@article{Rujan_1993,
  title="{Finite temperature error-correcting codes}",
  author = {Ruj\'an, P\'al},
  journal = {Phys. Rev. Lett.},
  volume = {70},
  issue = {19},
  pages = {2968--2971},
  numpages = {0},
  year = {1993},
  month = {May},
  publisher = {American Physical Society},
  doi = {10.1103/PhysRevLett.70.2968},
  url = {https://link.aps.org/doi/10.1103/PhysRevLett.70.2968}
}

@article{Sourlas_1994,
	doi = {10.1209/0295-5075/25/3/001},
	url = {https://doi.org/10.1209/0295-5075/25/3/001},
	year = 1994,
	month = {jan},
	publisher = {{IOP} Publishing},
	volume = {25},
	number = {3},
	pages = {159--164},
	author = {N Sourlas},
	title="{Spin Glasses, Error-Correcting Codes and Finite-Temperature Decoding}",
	journal = {Europhysics Letters ({EPL})}
}

@article{YL2021z2decoing,
  title="{Decodable hybrid dynamics of open quantum systems with ${\mathbb{Z}}_{2}$ symmetry}",
  author = {Li, Yaodong and Fisher, Matthew P. A.},
  journal = {Phys. Rev. B},
  volume = {108},
  issue = {21},
  pages = {214302},
  numpages = {28},
  year = {2023},
  month = {Dec},
  publisher = {American Physical Society},
  doi = {10.1103/PhysRevB.108.214302},
  url = {https://link.aps.org/doi/10.1103/PhysRevB.108.214302}
}

@article{wang2003,
title="{Confinement-Higgs transition in a disordered gauge theory and the accuracy threshold for quantum memory}",
journal = {Annals of Physics},
volume = {303},
number = {1},
pages = {31-58},
year = {2003},
issn = {0003-4916},
doi = {https://doi.org/10.1016/S0003-4916(02)00019-2},
url = {https://www.sciencedirect.com/science/article/pii/S0003491602000192},
author = {Wang, Chenyang and Harrington, Jim and Preskill, John}
}

@article{kesselring2022anyon,
  title="{Anyon Condensation and the Color Code}",
  author = {Kesselring, Markus S. and Magdalena de la Fuente, Julio C. and Thomsen, Felix and Eisert, Jens and Bartlett, Stephen D. and Brown, Benjamin J.},
  journal = {PRX Quantum},
  volume = {5},
  issue = {1},
  pages = {010342},
  numpages = {56},
  year = {2024},
  month = {Mar},
  publisher = {American Physical Society},
  doi = {10.1103/PRXQuantum.5.010342},
  url = {https://link.aps.org/doi/10.1103/PRXQuantum.5.010342}
}

@article{Gidney_2022,
   title="{Stability Experiments: The Overlooked Dual of Memory Experiments}",
   volume={6},
   ISSN={2521-327X},
   url={http://dx.doi.org/10.22331/q-2022-08-24-786},
   DOI={10.22331/q-2022-08-24-786},
   journal={Quantum},
   publisher={Verein zur Forderung des Open Access Publizierens in den Quantenwissenschaften},
   author={Gidney, Craig},
   year={2022},
   month=aug, pages={786} }

@article{griffiths1967a_correlations,
  title="{Correlations in Ising ferromagnets. I}",
  author={Griffiths, Robert B},
  journal={Journal of Mathematical Physics},
  volume={8},
  number={3},
  pages={478--483},
  year={1967},
  publisher={American Institute of Physics},
  doi = {https://doi.org/10.1063/1.1705219}
}

@article{griffiths1967b_correlations,
  title="{Correlations in Ising ferromagnets. II. External magnetic fields}",
  author={Griffiths, Robert B},
  journal={Journal of Mathematical Physics},
  volume={8},
  number={3},
  pages={484--489},
  year={1967},
  publisher={American Institute of Physics},
  doi = {https://doi.org/10.1063/1.1705220}
}

@article{griffiths1967c_correlations,
  title="{Correlations in Ising ferromagnets. III: A mean-field bound for binary correlations}",
  author={Griffiths, Robert B},
  journal={Communications in Mathematical Physics},
  volume={6},
  number={2},
  pages={121--127},
  year={1967},
  publisher={Springer},
  doi = {https://doi.org/10.1007/BF01654128}
}

@article{kelly1968general,
  title="{General Griffiths' inequalities on correlations in Ising ferromagnets}",
  author={Kelly, Douglas G and Sherman, Seymour},
  journal={Journal of Mathematical Physics},
  volume={9},
  number={3},
  pages={466--484},
  year={1968},
  publisher={American Institute of Physics},
  doi = {https://doi.org/10.1063/1.1664600}
}

@book{friedli_velenik_2017,
place={Cambridge},
title="{Statistical Mechanics of Lattice Systems: A Concrete Mathematical Introduction}",
DOI={10.1017/9781316882603},
ISBN={978-1-107-18482-4},
publisher={Cambridge University Press},
author={Friedli, Sacha and Velenik, Yvan},
year={2017}
}

@misc{panteleev2021asymptotically,
    title="{Asymptotically Good Quantum and Locally Testable Classical LDPC Codes}",
    author={Pavel Panteleev and Gleb Kalachev},
    year={2021},
    eprint={2111.03654},
    archivePrefix={arXiv},
    primaryClass={cs.IT}
}

@misc{hastings2020fiber,
    title="{Fiber Bundle Codes: Breaking the $N^{1/2} \textrm{polylog}(N)$ Barrier for Quantum LDPC Codes}",
    author={Matthew B. Hastings and Jeongwan Haah and Ryan O'Donnell},
    year={2020},
    eprint={2009.03921},
    archivePrefix={arXiv},
    primaryClass={quant-ph}
}

@article{Breuckmann_2021,
   title="{Balanced Product Quantum Codes}",
   volume={67},
   ISSN={1557-9654},
   url={http://dx.doi.org/10.1109/TIT.2021.3097347},
   DOI={10.1109/tit.2021.3097347},
   number={10},
   journal={IEEE Transactions on Information Theory},
   publisher={Institute of Electrical and Electronics Engineers (IEEE)},
   author={Breuckmann, Nikolas P. and Eberhardt, Jens N.},
   year={2021},
   month=oct, pages={6653–6674} }

@misc{leverrier2022quantum,
      title="{Quantum Tanner codes}", 
      author={Anthony Leverrier and Gilles Zémor},
      year={2022},
      eprint={2202.13641},
      archivePrefix={arXiv},
      primaryClass={quant-ph}
}

@misc{dinur2022good,
    title="{Good Quantum LDPC Codes with Linear Time Decoders}",
    author={Irit Dinur and Min-Hsiu Hsieh and Ting-Chun Lin and Thomas Vidick},
    year={2022},
    eprint={2206.07750},
    archivePrefix={arXiv},
    primaryClass={quant-ph}
}

@misc{dinur2021locally,
      title="{Locally Testable Codes with constant rate, distance, and locality}", 
      author={Irit Dinur and Shai Evra and Ron Livne and Alexander Lubotzky and Shahar Mozes},
      year={2021},
      eprint={2111.04808},
      archivePrefix={arXiv},
      primaryClass={cs.IT}
}

@misc{gottesman2013faulttolerant,
    title="{Fault-Tolerant Quantum Computation with Constant Overhead}",
    author={Daniel Gottesman},
    year={2013},
    eprint={1310.2984},
    archivePrefix={arXiv},
    primaryClass={quant-ph}
}

@inproceedings{Anshu_2023, series={STOC ’23},
   title="{NLTS Hamiltonians from Good Quantum Codes}",
   url={http://dx.doi.org/10.1145/3564246.3585114},
   DOI={10.1145/3564246.3585114},
   booktitle="{Proceedings of the 55th Annual ACM Symposium on Theory of Computing}",
   publisher={ACM},
   author={Anshu, Anurag and Breuckmann, Nikolas P. and Nirkhe, Chinmay},
   year={2023},
   month=jun, collection={STOC ’23} }

@misc{freedman2013quantum,
    title="{Quantum Systems on Non-$k$-Hyperfinite Complexes: A Generalization of Classical Statistical Mechanics on Expander Graphs}",
    author={M. H. Freedman and M. B. Hastings},
    year={2013},
    eprint={1301.1363},
    archivePrefix={arXiv},
    primaryClass={quant-ph}
}

@article{Jiang_2019,
   title="{Duality and free energy analyticity bounds for few-body Ising models with extensive homology rank}",
   volume={60},
   pages={083302},
   ISSN={1089-7658},
   DOI={10.1063/1.5039735},
   number={8},
   journal={Journal of Mathematical Physics},
   publisher={AIP Publishing},
   author={Jiang, Yi and Dumer, Ilya and Kovalev, Alexey A. and Pryadko, Leonid P.},
   year={2019},
   month=aug }

@article{Kovalev_2018,
   title="{Numerical and analytical bounds on threshold error rates for hypergraph-product codes}",
   volume={97},
   pages={062320},
   ISSN={2469-9934},
   DOI={10.1103/physreva.97.062320},
   number={6},
   journal={Physical Review A},
   publisher={American Physical Society (APS)},
   author={Kovalev, Alexey A. and Prabhakar, Sanjay and Dumer, Ilya and Pryadko, Leonid P.},
   year={2018},
   month=jun }

@article{Bravyi_2010,
       author = {{Bravyi}, Sergey and {Hastings}, Matthew B. and {Michalakis}, Spyridon},
        title = "{Topological quantum order: Stability under local perturbations}",
      journal = {Journal of Mathematical Physics},
         year = 2010,
        month = sep,
       volume = {51},
       number = {9},
        pages = {093512-093512},
          doi = {10.1063/1.3490195},
}

@article{Bravyi_2011,
   title="{A Short Proof of Stability of Topological Order under Local Perturbations}",
   volume={307},
   ISSN={1432-0916},
   url={http://dx.doi.org/10.1007/s00220-011-1346-2},
   DOI={10.1007/s00220-011-1346-2},
   number={3},
   journal={Communications in Mathematical Physics},
   publisher={Springer Science and Business Media LLC},
   author={Bravyi, Sergey and Hastings, Matthew B.},
   year={2011},
   month=sep, pages={609–627} }

@article{Michalakis_2013,
   title="{Stability of Frustration-Free Hamiltonians}",
   volume={322},
   ISSN={1432-0916},
   url={http://dx.doi.org/10.1007/s00220-013-1762-6},
   DOI={10.1007/s00220-013-1762-6},
   number={2},
   journal={Communications in Mathematical Physics},
   publisher={Springer Science and Business Media LLC},
   author={Michalakis, Spyridon and Zwolak, Justyna P.},
   year={2013},
   month=jul, pages={277–302} }

@misc{lavasani2024klocal,
    title={On stability of k-local quantum phases of matter},
    author={Ali Lavasani and Michael J. Gullans and Victor V. Albert and Maissam Barkeshli},
    year={2024},
    eprint={2405.19412},
    archivePrefix={arXiv},
    primaryClass={quant-ph}
}

@misc{pryadko2013,
      title="{Spin glass reflection of the decoding transition for quantum error correcting codes}", 
      author={Alexey A. Kovalev and Leonid P. Pryadko},
      year={2013},
      eprint={1311.7688},
      archivePrefix={arXiv},
      primaryClass={quant-ph}
}

@article{pryadko2014,
  title="{Thresholds for Correcting Errors, Erasures, and Faulty Syndrome Measurements in Degenerate Quantum Codes}",
  author = {Dumer, Ilya and Kovalev, Alexey A. and Pryadko, Leonid P.},
  journal = {Phys. Rev. Lett.},
  volume = {115},
  issue = {5},
  pages = {050502},
  numpages = {5},
  year = {2015},
  month = {Jul},
  publisher = {American Physical Society},
  doi = {10.1103/PhysRevLett.115.050502},
  url = {https://link.aps.org/doi/10.1103/PhysRevLett.115.050502}
}

@book{nishimori2001book,
    author = {Nishimori, Hidetoshi},
    title = "{Statistical Physics of Spin Glasses and Information Processing: An Introduction}",
    publisher = {Oxford University Press},
    year = {2001},
    month = {07},
    isbn = {9780198509417},
    doi = {10.1093/acprof:oso/9780198509417.001.0001}
}

@book{mezardmontanari2009book,
  title="{Information, physics, and computation}",
  author={Mezard, Marc and Montanari, Andrea},
  year={2009},
  publisher={Oxford University Press},
  doi = {10.1093/acprof:oso/9780198570837.001.0001}
}

@article{Zdeborova_2016,
   title="{Statistical physics of inference: thresholds and algorithms}",
   volume={65},
   ISSN={1460-6976},
   url={http://dx.doi.org/10.1080/00018732.2016.1211393},
   DOI={10.1080/00018732.2016.1211393},
   number={5},
   journal={Advances in Physics},
   publisher={Informa UK Limited},
   author={Zdeborová, Lenka and Krzakala, Florent},
   year={2016},
   month=aug, pages={453–552} }

@article{ChubbFlammia2018,
   title="{Statistical mechanical models for quantum codes with correlated noise}",
   volume={8},
   ISSN={2308-5835},
   url={http://dx.doi.org/10.4171/AIHPD/105},
   DOI={10.4171/aihpd/105},
   number={2},
   journal={Annales de l’Institut Henri Poincaré D, Combinatorics, Physics and their Interactions},
   publisher={European Mathematical Society - EMS - Publishing House GmbH},
   author={Chubb, Christopher T. and Flammia, Steven T.},
   year={2021},
   month=may, pages={269–321} }

@article{Kitaev_1997,
   title="{Fault-tolerant quantum computation by anyons}",
   volume={303},
   ISSN={0003-4916},
   url={http://dx.doi.org/10.1016/S0003-4916(02)00018-0},
   DOI={10.1016/s0003-4916(02)00018-0},
   number={1},
   journal={Annals of Physics},
   publisher={Elsevier BV},
   author={Kitaev, A.Yu.},
   year={2003},
   month=jan, pages={2–30} }

@ARTICLE{wegner1971,
       author = {{Wegner}, Franz J.},
        title = "{Duality in Generalized Ising Models and Phase Transitions without Local Order Parameters}",
      journal = {Journal of Mathematical Physics},
         year = 1971,
        month = oct,
       volume = {12},
       number = {10},
        pages = {2259-2272},
          doi = {10.1063/1.1665530},
       adsurl = {https://ui.adsabs.harvard.edu/abs/1971JMP....12.2259W}
}

@article{Klich_2010,
   title="{On the stability of topological phases on a lattice}",
   volume={325},
   ISSN={0003-4916},
   url={http://dx.doi.org/10.1016/j.aop.2010.05.002},
   DOI={10.1016/j.aop.2010.05.002},
   number={10},
   journal={Annals of Physics},
   publisher={Elsevier BV},
   author={Klich, Israel},
   year={2010},
   month=oct, pages={2120–2131} }

@article{nishimori1981internal,
  title="{Internal energy, specific heat and correlation function of the bond-random Ising model}",
  author={Nishimori, Hidetoshi},
  journal={Progress of Theoretical Physics},
  volume={66},
  number={4},
  pages={1169--1181},
  year={1981},
  publisher={Oxford University Press},
  doi={10.1143/PTP.66.1169}
}

@article{nishimori1986geometry,
  title="{Geometry-induced phase transition in the$\pm$J Ising model}",
  author={Nishimori, Hidetoshi},
  journal={Journal of the Physical Society of Japan},
  volume={55},
  number={10},
  pages={3305--3307},
  year={1986},
  publisher={The Physical Society of Japan},
  doi={10.1143/JPSJ.55.3305}
}

@article{nishimori1993optimum,
  title="{Optimum decoding temperature for error-correcting codes}",
  author={Nishimori, Hidetoshi},
  journal={Journal of the Physical Society of Japan},
  volume={62},
  number={9},
  pages={2973--2975},
  year={1993},
  publisher={The Physical Society of Japan},
  doi={10.1143/JPSJ.62.2973}
}

@article{poulin2008bp_splitbelief,
   title="{On the iterative decoding of sparse quantum codes}",
   volume={8},
   ISSN={1533-7146},
   url={http://dx.doi.org/10.26421/QIC8.10-8},
   DOI={10.26421/qic8.10-8},
   number={10},
   journal={Quantum Information and Computation},
   publisher={Rinton Press},
   author={Poulin, D. and Chung, Y.},
   year={2008},
   month={Nov},
   pages={986–1000} }

@article{fradkinsusskind1978,
  title = "{Order and disorder in gauge systems and magnets}",
  author = {Fradkin, Eduardo and Susskind, Leonard},
  journal = {Phys. Rev. D},
  volume = {17},
  issue = {10},
  pages = {2637--2658},
  numpages = {0},
  year = {1978},
  month = {May},
  publisher = {American Physical Society},
  doi = {10.1103/PhysRevD.17.2637},
  url = {https://link.aps.org/doi/10.1103/PhysRevD.17.2637}
}

@article{kogut1979rmp,
  title = "{An introduction to lattice gauge theory and spin systems}",
  author = {Kogut, John B.},
  journal = {Rev. Mod. Phys.},
  volume = {51},
  issue = {4},
  pages = {659--713},
  numpages = {0},
  year = {1979},
  month = {Oct},
  publisher = {American Physical Society},
  doi = {10.1103/RevModPhys.51.659},
  url = {https://link.aps.org/doi/10.1103/RevModPhys.51.659}
}

@article{kogut1983rmp,
  title="{The lattice gauge theory approach to quantum chromodynamics}",
  author = {Kogut, John B.},
  journal = {Rev. Mod. Phys.},
  volume = {55},
  issue = {3},
  pages = {775--836},
  numpages = {0},
  year = {1983},
  month = {Jul},
  publisher = {American Physical Society},
  doi = {10.1103/RevModPhys.55.775},
  url = {https://link.aps.org/doi/10.1103/RevModPhys.55.775}
}

@article{Sondhi1997,
  title="{Continuous quantum phase transitions}",
  author = {Sondhi, S. L. and Girvin, S. M. and Carini, J. P. and Shahar, D.},
  journal = {Rev. Mod. Phys.},
  volume = {69},
  issue = {1},
  pages = {315--333},
  numpages = {0},
  year = {1997},
  month = {Jan},
  publisher = {American Physical Society},
  doi = {10.1103/RevModPhys.69.315},
  url = {https://link.aps.org/doi/10.1103/RevModPhys.69.315}
}

@book{polyakov1987gauge,
  title="{Gauge fields and strings}",
  author={Polyakov, Alexander Markovich},
  year={1987},
  publisher={Taylor \& Francis},
  doi={10.1201/9780203755082}
}

@misc{lavasani2024stability,
      title="{The Stability of Gapped Quantum Matter and Error-Correction with Adiabatic Noise}", 
      author={Ali Lavasani and Sagar Vijay},
      year={2024},
      eprint={2402.14906},
      archivePrefix={arXiv},
      primaryClass={cond-mat.str-el}
}

@misc{gidney2023bacon,
    title="{Less Bacon More Threshold}",
    author={Craig Gidney and Dave Bacon},
    year={2023},
    eprint={2305.12046},
    archivePrefix={arXiv},
    primaryClass={quant-ph}
}

@article{kubica-2018-PRL-statmech-3Dcolorcode,
  title="{Three-Dimensional Color Code Thresholds via Statistical-Mechanical Mapping}",
  author = {Kubica, Aleksander and Beverland, Michael E. and Brand\~ao, Fernando and Preskill, John and Svore, Krysta M.},
  journal = {Phys. Rev. Lett.},
  volume = {120},
  issue = {18},
  pages = {180501},
  numpages = {6},
  year = {2018},
  month = {May},
  publisher = {American Physical Society},
  doi = {10.1103/PhysRevLett.120.180501},
  url = {https://link.aps.org/doi/10.1103/PhysRevLett.120.180501}
}

@ARTICLE{hastings2021FloquetCode,
   title="{Dynamically Generated Logical Qubits}",
   volume={5},
   ISSN={2521-327X},
   url={http://dx.doi.org/10.22331/q-2021-10-19-564},
   DOI={10.22331/q-2021-10-19-564},
   journal={Quantum},
   publisher={Verein zur Forderung des Open Access Publizierens in den Quantenwissenschaften},
   author={Hastings, Matthew B. and Haah, Jeongwan},
   year={2021},
   month={Oct},
   pages={564}
}

@article{davydova2022floquet,
  title="{Floquet Codes without Parent Subsystem Codes}",
  author = {Davydova, Margarita and Tantivasadakarn, Nathanan and Balasubramanian, Shankar},
  journal = {PRX Quantum},
  volume = {4},
  issue = {2},
  pages = {020341},
  numpages = {20},
  year = {2023},
  month = {Jun},
  publisher = {American Physical Society},
  doi = {10.1103/PRXQuantum.4.020341},
  url = {https://link.aps.org/doi/10.1103/PRXQuantum.4.020341}
}

@ARTICLE{bombin2015gaugecolorcode,
       author = {{Bomb{\'\i}n}, H{\'e}ctor},
        title = "{Gauge Color Codes: Optimal Transversal Gates and Gauge Fixing in Topological Stabilizer Codes}",
   volume={17},
   ISSN={1367-2630},
   url={http://dx.doi.org/10.1088/1367-2630/17/8/083002},
   DOI={10.1088/1367-2630/17/8/083002},
   number={8},
   journal={New Journal of Physics},
   publisher={IOP Publishing},
   author={Bombín, Héctor},
   year={2015},
   month={Aug},
   pages={083002}
}

@article{bombin2015singleshot,
       author = {{Bomb{\'\i}n}, H{\'e}ctor},
        title = "{Single-Shot Fault-Tolerant Quantum Error Correction}",
      journal = {Phys. Rev. X},
         year = 2015,
        month = jul,
       volume = {5},
       number = {3},
          eid = {031043},
        pages = {031043},
          doi = {10.1103/PhysRevX.5.031043}
}

@article{KubicaVasmer,
       author = {{Kubica}, Aleksander and {Vasmer}, Michael},
        title = "{Single-shot quantum error correction with the three-dimensional subsystem toric code}",
  volume={13},
  eid={6272},
  pages={6272},
   ISSN={2041-1723},
   url={http://dx.doi.org/10.1038/s41467-022-33923-4},
   number={1},
   journal={Nature Communications},
   publisher={Springer Science and Business Media LLC},
   year={2022},
   month={Oct}
}

@misc{delfosse2023spacetime,
      title="{Spacetime codes of Clifford circuits}", 
      author={Nicolas Delfosse and Adam Paetznick},
      year={2023},
      eprint={2304.05943},
      archivePrefix={arXiv},
      primaryClass={quant-ph}
}

@article{bacon2017sparse,
  title="{Sparse quantum codes from quantum circuits}",
  author={Bacon, Dave and Flammia, Steven T and Harrow, Aram W and Shi, Jonathan},
  journal={IEEE Transactions on Information Theory},
  volume={63},
  number={4},
  pages={2464--2479},
  year={2017},
  publisher={IEEE},
  doi={10.1109/TIT.2017.2663199}
}

@article{McEwen_2023,
   title="{Relaxing Hardware Requirements for Surface Code Circuits using Time-dynamics}",
   volume={7},
   ISSN={2521-327X},
   url={http://dx.doi.org/10.22331/q-2023-11-07-1172},
   DOI={10.22331/q-2023-11-07-1172},
   journal={Quantum},
   publisher={Verein zur Forderung des Open Access Publizierens in den Quantenwissenschaften},
   author={McEwen, Matt and Bacon, Dave and Gidney, Craig},
   year={2023},
   month=nov, pages={1172} }

@article{nahum2020selfdual,
  title="{Self-Dual Criticality in Three-Dimensional ${\mathbb{Z}}_{2}$ Gauge Theory with Matter}",
  author = {Somoza, Andr\'es M. and Serna, Pablo and Nahum, Adam},
  journal = {Phys. Rev. X},
  volume = {11},
  issue = {4},
  pages = {041008},
  numpages = {32},
  year = {2021},
  month = {Oct},
  publisher = {American Physical Society},
  doi = {10.1103/PhysRevX.11.041008},
  url = {https://link.aps.org/doi/10.1103/PhysRevX.11.041008}
}

@article{fradkinshenker1979,
  title="{Phase diagrams of lattice gauge theories with Higgs fields}",
  author = {Fradkin, Eduardo and Shenker, Stephen H.},
  journal = {Phys. Rev. D},
  volume = {19},
  issue = {12},
  pages = {3682--3697},
  numpages = {0},
  year = {1979},
  month = {Jun},
  publisher = {American Physical Society},
  doi = {10.1103/PhysRevD.19.3682},
  url = {https://link.aps.org/doi/10.1103/PhysRevD.19.3682}
}

@misc{takeda2003selfdual,
    title="{Self-dual random-plaquette gauge model and the quantum toric code}",
    author={Koujin Takeda and Hidetoshi Nishimori},
    year={2003},
    eprint={hep-th/0310279},
    archivePrefix={arXiv},
    primaryClass={hep-th}
}

@article{Breuckmann_2016,
   title="{Constructions and Noise Threshold of Hyperbolic Surface Codes}",
   volume={62},
   ISSN={1557-9654},
   url={http://dx.doi.org/10.1109/TIT.2016.2555700},
   DOI={10.1109/tit.2016.2555700},
   number={6},
   journal={IEEE Transactions on Information Theory},
   publisher={Institute of Electrical and Electronics Engineers (IEEE)},
   author={Breuckmann, Nikolas P. and Terhal, Barbara M.},
   year={2016},
   month=jun, pages={3731–3744} }

@misc{kubica2018ungauging,
    title="{Ungauging quantum error-correcting codes}",
    author={Aleksander Kubica and Beni Yoshida},
    year={2018},
    eprint={1805.01836},
    archivePrefix={arXiv},
    primaryClass={quant-ph}
}

@article{Tuckett_2019,
  title="{Tailoring Surface Codes for Highly Biased Noise}",
  author = {Tuckett, David K. and Darmawan, Andrew S. and Chubb, Christopher T. and Bravyi, Sergey and Bartlett, Stephen D. and Flammia, Steven T.},
  journal = {Phys. Rev. X},
  volume = {9},
  issue = {4},
  pages = {041031},
  numpages = {22},
  year = {2019},
  month = {Nov},
  publisher = {American Physical Society},
  doi = {10.1103/PhysRevX.9.041031},
  url = {https://link.aps.org/doi/10.1103/PhysRevX.9.041031}
}

@misc{YL2023STC,
    title="{Phase diagram of the three-dimensional subsystem toric code}",
    author={Yaodong Li and C. W. von Keyserlingk and Guanyu Zhu and Tomas Jochym-O'Connor},
    year={2023},
    eprint={2305.06389},
    archivePrefix={arXiv},
    primaryClass={quant-ph}
}

@misc{stahl2023singleshot,
    title="{Single-Shot Quantum Error Correction in Intertwined Toric Codes}",
    author={Charles Stahl},
    year={2023},
    eprint={2307.08118},
    archivePrefix={arXiv},
    primaryClass={quant-ph}
}

@misc{bombin2023faulttolerant,
    title="{Fault-tolerant complexes}",
    author={Hector Bombin and Chris Dawson and Terry Farrelly and Yehua Liu and Naomi Nickerson and Mihir Pant and Fernando Pastawski and Sam Roberts},
    year={2023},
    eprint={2308.07844},
    archivePrefix={arXiv},
    primaryClass={quant-ph}
}

@misc{bombin2023unifying,
    title="{Unifying flavors of fault tolerance with the ZX calculus}",
    author={Hector Bombin and Daniel Litinski and Naomi Nickerson and Fernando Pastawski and Sam Roberts},
    year={2023},
    eprint={2303.08829},
    archivePrefix={arXiv},
    primaryClass={quant-ph}
}

@article{Townsend_Teague_2023,
   title="{Floquetifying the Colour Code}",
   volume={384},
   ISSN={2075-2180},
   url={http://dx.doi.org/10.4204/EPTCS.384.14},
   DOI={10.4204/eptcs.384.14},
   journal={Electronic Proceedings in Theoretical Computer Science},
   publisher={Open Publishing Association},
   author={Townsend-Teague, Alex and Magdalena de la Fuente, Julio and Kesselring, Markus},
   year={2023},
   month=aug, pages={265–303} }

@article{Andi_Bauer_2024,
   title="{Topological error correcting processes from fixed-point path integrals}",
   volume={8},
   ISSN={2521-327X},
   url={http://dx.doi.org/10.22331/q-2024-03-20-1288},
   DOI={10.22331/q-2024-03-20-1288},
   journal={Quantum},
   publisher={Verein zur Forderung des Open Access Publizierens in den Quantenwissenschaften},
   author={Bauer, Andreas},
   year={2024},
   month=mar, pages={1288} }

@article{McCoy-Wu,
  title="{Theory of a Two-Dimensional Ising Model with Random Impurities. I. Thermodynamics}",
  author = {McCoy, Barry M. and Wu, Tai Tsun},
  journal = {Phys. Rev.},
  volume = {176},
  issue = {2},
  pages = {631--643},
  numpages = {0},
  year = {1968},
  month = {Dec},
  publisher = {American Physical Society},
  doi = {10.1103/PhysRev.176.631},
  url = {https://link.aps.org/doi/10.1103/PhysRev.176.631}
}

@article{dsfisher1992,
  title="{Random transverse field Ising spin chains}",
  author = {Fisher, Daniel S.},
  journal = {Phys. Rev. Lett.},
  volume = {69},
  issue = {3},
  pages = {534--537},
  numpages = {0},
  year = {1992},
  month = {Jul},
  publisher = {American Physical Society},
  doi = {10.1103/PhysRevLett.69.534},
  url = {https://link.aps.org/doi/10.1103/PhysRevLett.69.534}
}

@article{dsfisher1995,
  title="{Critical behavior of random transverse-field Ising spin chains}",
  author = {Fisher, Daniel S.},
  journal = {Phys. Rev. B},
  volume = {51},
  issue = {10},
  pages = {6411--6461},
  numpages = {0},
  year = {1995},
  month = {Mar},
  publisher = {American Physical Society},
  doi = {10.1103/PhysRevB.51.6411},
  url = {https://link.aps.org/doi/10.1103/PhysRevB.51.6411}
}

@misc{fan2023diagnostics,
    title="{Diagnostics of mixed-state topological order and breakdown of quantum memory}",
    author={Ruihua Fan and Yimu Bao and Ehud Altman and Ashvin Vishwanath},
    year={2023},
    eprint={2301.05689},
    archivePrefix={arXiv},
    primaryClass={quant-ph}
}

@misc{jongyeonlee2024exact,
    title="{Exact Calculations of Coherent Information for Toric Codes under Decoherence: Identifying the Fundamental Error Threshold}",
    author={Jong Yeon Lee},
    year={2024},
    eprint={2402.16937},
    archivePrefix={arXiv},
    primaryClass={cond-mat.stat-mech}
}

@misc{su2024tapestry,
    title="{Tapestry of dualities in decohered quantum error correction codes}",
    author={Kaixiang Su and Zhou Yang and Chao-Ming Jian},
    year={2024},
    eprint={2401.17359},
    archivePrefix={arXiv},
    primaryClass={cond-mat.str-el}
}

@misc{lyons2024understanding,
    title="{Understanding Stabilizer Codes Under Local Decoherence Through a General Statistical Mechanics Mapping}",
    author={Anasuya Lyons},
    year={2024},
    eprint={2403.03955},
    archivePrefix={arXiv},
    primaryClass={quant-ph}
}

@misc{yimu2024,
    title={Information dynamics in decohered quantum memory with repeated syndrome measurements: a dual approach},
    author={Jacob Hauser and Yimu Bao and Shengqi Sang and Ali Lavasani and Utkarsh Agrawal and Matthew P. A. Fisher},
    year={2024},
    eprint={2407.07882},
    archivePrefix={arXiv},
    primaryClass={quant-ph}
}

@misc{hastings2023graph,
    title="{Quantum Codes on Graphs}",
    author={M. B. Hastings},
    year={2023},
    eprint={2308.10264},
    archivePrefix={arXiv},
    primaryClass={quant-ph}
}

@article{griffiths1964peierls,
  title="{Peierls Proof of Spontaneous Magnetization in a Two-Dimensional Ising Ferromagnet}",
  author = {Griffiths, Robert B.},
  journal = {Phys. Rev.},
  volume = {136},
  issue = {2A},
  pages = {A437--A439},
  numpages = {0},
  year = {1964},
  month = {Oct},
  publisher = {American Physical Society},
  doi = {10.1103/PhysRev.136.A437},
  url = {https://link.aps.org/doi/10.1103/PhysRev.136.A437}
}

@article{benbrowndomwilliamson2020parallelized,
  title="{Parallelized quantum error correction with fracton topological codes}",
  author = {Brown, Benjamin J. and Williamson, Dominic J.},
  journal = {Phys. Rev. Res.},
  volume = {2},
  issue = {1},
  pages = {013303},
  numpages = {15},
  year = {2020},
  month = {Mar},
  publisher = {American Physical Society},
  doi = {10.1103/PhysRevResearch.2.013303},
  url = {https://link.aps.org/doi/10.1103/PhysRevResearch.2.013303}
}

@article{vidal_2009_xumoore,
  title="{Self-duality and bound states of the toric code model in a transverse field}",
  author = {Vidal, Julien and Thomale, Ronny and Schmidt, Kai Phillip and Dusuel, S\'ebastien},
  journal = {Phys. Rev. B},
  volume = {80},
  issue = {8},
  pages = {081104},
  numpages = {4},
  year = {2009},
  month = {Aug},
  publisher = {American Physical Society},
  doi = {10.1103/PhysRevB.80.081104},
  url = {https://link.aps.org/doi/10.1103/PhysRevB.80.081104}
}

@article{xumoore_2004,
  title="{Strong-Weak Coupling Self-Duality in the Two-Dimensional Quantum Phase Transition of $p+ip$ Superconducting Arrays}",
  author = {Xu, Cenke and Moore, J. E.},
  journal = {Phys. Rev. Lett.},
  volume = {93},
  issue = {4},
  pages = {047003},
  numpages = {4},
  year = {2004},
  month = {Jul},
  publisher = {American Physical Society},
  doi = {10.1103/PhysRevLett.93.047003},
  url = {https://link.aps.org/doi/10.1103/PhysRevLett.93.047003}
}

@article{nussinov2005compass,
  title="{Discrete sliding symmetries, dualities, and self-dualities of quantum orbital compass models and $p+ip$ superconducting arrays}",
  author = {Nussinov, Zohar and Fradkin, Eduardo},
  journal = {Phys. Rev. B},
  volume = {71},
  issue = {19},
  pages = {195120},
  numpages = {9},
  year = {2005},
  month = {May},
  publisher = {American Physical Society},
  doi = {10.1103/PhysRevB.71.195120},
  url = {https://link.aps.org/doi/10.1103/PhysRevB.71.195120}
}

@article{bacon-shor-2006,
  title="{Operator quantum error-correcting subsystems for self-correcting quantum memories}",
  author = {Bacon, Dave},
  journal = {Phys. Rev. A},
  volume = {73},
  issue = {1},
  pages = {012340},
  numpages = {13},
  year = {2006},
  month = {Jan},
  publisher = {American Physical Society},
  doi = {10.1103/PhysRevA.73.012340},
  url = {https://link.aps.org/doi/10.1103/PhysRevA.73.012340}
}

@article{andrist2015pra,
  title="{Error thresholds for Abelian quantum double models: Increasing the bit-flip stability of topological quantum memory}",
  author = {Andrist, Ruben S. and Wootton, James R. and Katzgraber, Helmut G.},
  journal = {Phys. Rev. A},
  volume = {91},
  issue = {4},
  pages = {042331},
  numpages = {7},
  year = {2015},
  month = {Apr},
  publisher = {American Physical Society},
  doi = {10.1103/PhysRevA.91.042331},
  url = {https://link.aps.org/doi/10.1103/PhysRevA.91.042331}
}

@article{andrist2011tricolored,
doi = {10.1088/1367-2630/13/8/083006},
year = {2011},
month = {aug},
publisher = {},
volume = {13},
number = {8},
pages = {083006},
author = {Ruben S Andrist and Helmut G Katzgraber and H Bombin and M A Martin-Delgado},
title="{Tricolored lattice gauge theory with randomness: fault tolerance in topological color codes}",
journal = {New Journal of Physics}
}

@article{haosong2022fractonstatmech,
  title="{Optimal Thresholds for Fracton Codes and Random Spin Models with Subsystem Symmetry}",
  author = {Song, Hao and Sch\"onmeier-Kromer, Janik and Liu, Ke and Viyuela, Oscar and Pollet, Lode and Martin-Delgado, M. A.},
  journal = {Phys. Rev. Lett.},
  volume = {129},
  issue = {23},
  pages = {230502},
  numpages = {7},
  year = {2022},
  month = {Nov},
  publisher = {American Physical Society},
  doi = {10.1103/PhysRevLett.129.230502},
  url = {https://link.aps.org/doi/10.1103/PhysRevLett.129.230502}
}

@article{brown2022conservation,
   title="{Conservation Laws and Quantum Error Correction: Toward a Generalized Matching Decoder}",
   volume={2},
   ISSN={2692-4080},
   url={http://dx.doi.org/10.1109/MBITS.2023.3246025},
   DOI={10.1109/mbits.2023.3246025},
   number={3},
   journal={IEEE BITS the Information Theory Magazine},
   publisher={Institute of Electrical and Electronics Engineers (IEEE)},
   author={Brown, Benjamin J.},
   year={2022},
   month=dec, pages={5–19} }

@misc{beverland2024fault,
    title="{Fault tolerance of stabilizer channels}",
    author={Michael E. Beverland and Shilin Huang and Vadym Kliuchnikov},
    year={2024},
    eprint={2401.12017},
    archivePrefix={arXiv},
    primaryClass={quant-ph}
}

@misc{fugottesman2024error,
    title="{Error Correction in Dynamical Codes}",
    author={Xiaozhen Fu and Daniel Gottesman},
    year={2024},
    eprint={2403.04163},
    archivePrefix={arXiv},
    primaryClass={quant-ph}
}

@article{campbell2024stability,
  title="{Tangling Schedules Eases Hardware Connectivity Requirements for Quantum Error Correction}",
  author = {Geh\'er, Gy\"orgy P. and Crawford, Ophelia and Campbell, Earl T.},
  journal = {PRX Quantum},
  volume = {5},
  issue = {1},
  pages = {010348},
  numpages = {26},
  year = {2024},
  month = {Mar},
  publisher = {American Physical Society},
  doi = {10.1103/PRXQuantum.5.010348},
  url = {https://link.aps.org/doi/10.1103/PRXQuantum.5.010348}
}

@article{plackebreuckmann2023hyperbolicIsing,
  title="{Random-bond Ising model and its dual in hyperbolic spaces}",
  author = {Placke, Benedikt and Breuckmann, Nikolas P.},
  journal = {Phys. Rev. E},
  volume = {107},
  issue = {2},
  pages = {024125},
  numpages = {13},
  year = {2023},
  month = {Feb},
  publisher = {American Physical Society},
  doi = {10.1103/PhysRevE.107.024125},
  url = {https://link.aps.org/doi/10.1103/PhysRevE.107.024125}
}

@misc{fahimniya2023faulttolerant,
    title="{Fault-tolerant hyperbolic Floquet quantum error correcting codes}",
    author={Ali Fahimniya and Hossein Dehghani and Kishor Bharti and Sheryl Mathew and Alicia J. Kollár and Alexey V. Gorshkov and Michael J. Gullans},
    year={2023},
    eprint={2309.10033},
    archivePrefix={arXiv},
    primaryClass={quant-ph}
}

@inproceedings{peierls1936ising,
  title="{On Ising's model of ferromagnetism}",
  author={Peierls, Rudolf},
  booktitle="{Mathematical Proceedings of the Cambridge Philosophical Society}",
  volume={32},
  number={3},
  pages={477--481},
  year={1936},
  organization={Cambridge University Press},
  doi={10.1017/S0305004100019174}
}

@article{Aharonov_1999,
   title="{Fault-Tolerant Quantum Computation with Constant Error Rate}",
   volume={38},
   ISSN={1095-7111},
   DOI={10.1137/s0097539799359385},
   number={4},
   journal={SIAM Journal on Computing},
   publisher={Society for Industrial & Applied Mathematics (SIAM)},
   author={Aharonov, Dorit and Ben-Or, Michael},
   year={2008},
   month=jan, pages={1207–1282} }

@misc{breuckmann2018phd,
    title="{PhD thesis: Homological Quantum Codes Beyond the Toric Code}",
    author={Nikolas P. Breuckmann},
    year={2018},
    eprint={1802.01520},
    archivePrefix={arXiv},
    primaryClass={quant-ph}
}

@article{bpt2010,
  title="{Tradeoffs for Reliable Quantum Information Storage in 2D Systems}",
  author = {Bravyi, Sergey and Poulin, David and Terhal, Barbara},
  journal = {Phys. Rev. Lett.},
  volume = {104},
  issue = {5},
  pages = {050503},
  numpages = {4},
  year = {2010},
  month = {Feb},
  publisher = {American Physical Society},
  doi = {10.1103/PhysRevLett.104.050503},
  url = {https://link.aps.org/doi/10.1103/PhysRevLett.104.050503}
}

@article{bridgeman2023,
  title="{Lifting Topological Codes: Three-Dimensional Subsystem Codes from Two-Dimensional Anyon Models}",
  author = {Bridgeman, Jacob C. and Kubica, Aleksander and Vasmer, Michael},
  journal = {PRX Quantum},
  volume = {5},
  issue = {2},
  pages = {020310},
  numpages = {30},
  year = {2024},
  month = {Apr},
  publisher = {American Physical Society},
  doi = {10.1103/PRXQuantum.5.020310},
  url = {https://link.aps.org/doi/10.1103/PRXQuantum.5.020310}
}

@article{schumacher_nielsen_1996,
  title="{Quantum data processing and error correction}",
  author = {Schumacher, Benjamin and Nielsen, M. A.},
  journal = {Phys. Rev. A},
  volume = {54},
  issue = {4},
  pages = {2629--2635},
  numpages = {0},
  year = {1996},
  month = {Oct},
  publisher = {American Physical Society},
  doi = {10.1103/PhysRevA.54.2629},
  url = {https://link.aps.org/doi/10.1103/PhysRevA.54.2629}
}

@misc{colmenarez2023accurate,
    title="{Accurate optimal quantum error correction thresholds from coherent information}",
    author={Luis Colmenarez and Ze-Min Huang and Sebastian Diehl and Markus Müller},
    year={2023},
    eprint={2312.06664},
    archivePrefix={arXiv},
    primaryClass={quant-ph}
}

@article{nahum2018hybrid,
       author = {{Skinner}, Brian and {Ruhman}, Jonathan and {Nahum}, Adam},
        title = "{Measurement-Induced Phase Transitions in the Dynamics of Entanglement}",
      journal = {Phys. Rev. X},
         year = "2019",
        month = "Jul",
       volume = {9},
       number = {3},
          eid = {031009},
        pages = {031009},
          doi = {10.1103/PhysRevX.9.031009}
}

@ARTICLE{li2020cft,
  title="{Conformal invariance and quantum nonlocality in critical hybrid circuits}",
  author = {Li, Yaodong and Chen, Xiao and Ludwig, Andreas W. W. and Fisher, Matthew P. A.},
  journal = {Phys. Rev. B},
  volume = {104},
  issue = {10},
  pages = {104305},
  numpages = {28},
  year = {2021},
  month = {Sep},
  publisher = {American Physical Society},
  doi = {10.1103/PhysRevB.104.104305},
  url = {https://link.aps.org/doi/10.1103/PhysRevB.104.104305}
}

@ARTICLE{chenxiao2020nonunitary,
       author = {{Chen}, Xiao and {Li}, Yaodong and {Fisher}, Matthew P.~A. and {Lucas}, Andrew},
        title = "{Emergent conformal symmetry in nonunitary random dynamics of free fermions}",
      journal = {Phys. Rev. Research},
         year = 2020,
        month = jul,
       volume = {2},
       number = {3},
          eid = {033017},
        pages = {033017},
          doi = {10.1103/PhysRevResearch.2.033017}
}

@article{andreas2019hybrid,
      author = {{Jian}, Chao-Ming and {You}, Yi-Zhuang and {Vasseur}, Romain and
         {Ludwig}, Andreas W.~W.},
        title = "{Measurement-induced criticality in random quantum circuits}",
      journal = {\prb},
         year = 2020,
        month = mar,
       volume = {101},
       number = {10},
          eid = {104302},
        pages = {104302},
          doi = {10.1103/PhysRevB.101.104302}
}

@ARTICLE{choi2019spin,
       author = {{Bao}, Yimu and {Choi}, Soonwon and {Altman}, Ehud},
        title = "{Theory of the phase transition in random unitary circuits with measurements}",
      journal = {\prb},
         year = 2020,
        month = mar,
       volume = {101},
       number = {10},
          eid = {104301},
        pages = {104301},
          doi = {10.1103/PhysRevB.101.104301}
}

@ARTICLE{li2018hybrid,
       author = {{Li}, Yaodong and {Chen}, Xiao and {Fisher}, Matthew P.~A.},
        title = "{Quantum Zeno effect and the many-body entanglement transition}",
      journal = {\prb},
         year = 2018,
        month = Nov,
       volume = {98},
          eid = {205136},
        pages = {205136},
          doi = {10.1103/PhysRevB.98.205136}
}

@article{CalderbankShor,
  title = {Good quantum error-correcting codes exist},
  author = {Calderbank, A. R. and Shor, Peter W.},
  journal = {Phys. Rev. A},
  volume = {54},
  issue = {2},
  pages = {1098--1105},
  numpages = {0},
  year = {1996},
  month = {Aug},
  publisher = {American Physical Society},
  doi = {10.1103/PhysRevA.54.1098},
  url = {https://link.aps.org/doi/10.1103/PhysRevA.54.1098}
}

@article{steane1996multiple,
  title={Multiple-particle interference and quantum error correction},
  author={Steane, Andrew},
  journal={Proceedings of the Royal Society of London. Series A: Mathematical, Physical and Engineering Sciences},
  volume={452},
  number={1954},
  pages={2551--2577},
  year={1996},
  publisher={The Royal Society London},
  doi = {10.1098/rspa.1996.0136}
}

@article{sagar2016fracton,
  title = {Fracton topological order, generalized lattice gauge theory, and duality},
  author = {Vijay, Sagar and Haah, Jeongwan and Fu, Liang},
  journal = {Phys. Rev. B},
  volume = {94},
  issue = {23},
  pages = {235157},
  numpages = {9},
  year = {2016},
  month = {Dec},
  publisher = {American Physical Society},
  doi = {10.1103/PhysRevB.94.235157},
  url = {https://link.aps.org/doi/10.1103/PhysRevB.94.235157}
}

@article{williamson2016fractal,
  title = {Fractal symmetries: Ungauging the cubic code},
  author = {Williamson, Dominic J.},
  journal = {Phys. Rev. B},
  volume = {94},
  issue = {15},
  pages = {155128},
  numpages = {9},
  year = {2016},
  month = {Oct},
  publisher = {American Physical Society},
  doi = {10.1103/PhysRevB.94.155128},
  url = {https://link.aps.org/doi/10.1103/PhysRevB.94.155128}
}

@Article{shirley2019foliated,
	title={{Foliated fracton order from gauging subsystem symmetries}},
	author={Wilbur Shirley and Kevin Slagle and Xie Chen},
	journal={SciPost Phys.},
	volume={6},
	pages={041},
	year={2019},
	publisher={SciPost},
	doi={10.21468/SciPostPhys.6.4.041},
	url={https://scipost.org/10.21468/SciPostPhys.6.4.041},
}

@misc{english2024thresholds,
    title={Thresholds for post-selected quantum error correction from statistical mechanics},
    author={Lucas H. English and Dominic J. Williamson and Stephen D. Bartlett},
    year={2024},
    eprint={2410.07598},
    archivePrefix={arXiv},
    primaryClass={quant-ph}
}

\appendix

\section{Correlated bitflip errors and non-uniform error rates\label{sec:correlated_bitflip}}

In this appendix, we explain how to generalize the transfer matrix construction in the main text to handle correlated bitflip errors as well as location-dependent error rates. Thus, the discussion in the main text of how error correction properties under independent bitflip errors relate to properties of quantum Hamiltonians made of $\pauliz$-stabilizers and perturbed by a field directly generalizes to this wider class of correlated errors. In particular, a correlated bitflip error is in one-one correspondence to perturbation of the quantum Hamiltonian by an interaction made of a product of $\paulix$'s on those sites.
\subsection{Transfer matrix and quantum Hamiltonian from correlated bitflip errors}
In Sec.~\ref{sec:decoding_success_physical_consequences} of the main text, we introduced an auxiliary Hilbert space to re-express conditional expectation values (conditioned on the measurement outcomes) as matrix elements of operators. The key 
object in this rewriting was the transfer matrix, whose matrix elements are defined as 
\begin{equation}
\begin{split}\label{eq:appA_transfer_def}
T(s_t)_{\sigma_t, \sigma_{t-1}} &= \mathbb{P}(\sigma_{t}, s_{t}|\sigma_{t-1}, s_{t-1}) 
    \\&= \mathbb{P}(s_{t}|\sigma_{t}) \cdot \mathbb{P}(\sigma_{t}|\sigma_{t-1}).
\end{split}
\end{equation}
Here, 
$\mathbb{P}(s_{t}|\sigma_{t})$ 
describes the measurement errors, which are assumed to be independent of one another and the history; whereas 
$\mathbb{P}(\sigma_{t}|\sigma_{t-1})$ 
describes the bitflip errors; we will allow spatially-correlated bitflip events, generalizing the main text whose focus was an error model with only independent bitflip errors. A given stabilizer $\qs$ will give the wrong measurement outcome with probability $p_\qs$, while a given bitflip error $\qe$ will occur with probability $p_\qe$.\footnote{Notice that we have adopted a notation that makes non-uniform, location dependent error rates explicit, in contrast to those in Sec.~\ref{sec:stat_mech_model}.}
For notational and conceptual ease, and to make contact with a perturbed time-independent quantum model, we will assume these error rates are independent of time. 

The transfer matrix can be described directly as an operator in the auxiliary Hilbert space:
\begin{equation}
\begin{split}
    T(s) &\propto
    {\exp( \sum_\qs K_\qs s_{\qs, t} \pauliz_\qs)} 
    {\exp(\sum_\qe \overline{K}_\qe \paulix_\qe)}\,\, 
    \\ 
    {K_\qs} &
    {= \frac{1}{2} \log(\frac{1-p_\qs}{p_\qs})}, \quad
    {\overline{K}_\qe = \frac{1}{2} \log(\frac{1}{1-2p_\qe})}
\end{split}
\end{equation}
where $\pauliz_\qs \equiv \prod_{\bfr \in \qs} \pauliz_\bfr$ as before, and $\paulix_\qe = \prod_{\bfr \in \qe} \paulix_\bfr$.
This matches the definition in Eq.~\eqref{eq:appA_transfer_def}. To see this, consider first the measurement errors. Configurations with $s_{\qs, t} = -\pauliz_\qs$ (i.e. a measurement error has occurred) are weighted by $e^{-K_\qs}$, while configurations with $s_{\qs, t} = + \pauliz_\qs$ (i.e. no measurement error) are weighted by $e^{+K_\qs}$. 
The relative weights of these configurations is the correct $e^{-2K_\qs} = \frac{1-p_\qs}{p_\qs}$. For the bitflip errors, note that $\exp(\overline{K}_\qe \paulix_\qe) = \cosh(K_\qe) + \sinh(K_\qe) \paulix_\qe$.
That is, the relative weight of configurations where the correlated bitflip error $\qe$ has occurred relative to those configurations where it hasn't occurred is the correct $\tanh(\overline{K}_\qe) = \frac{p_\qe}{1-p_\qe}$.

Note that for $s=1$, the resulting quantum Hamiltonian in the continuous-time limit of $\tau \to 0$ where $K_\qs = J_\qs \tau$ and $\overline{K}_\qe = h_\qe \tau$ is $H_q = - \sum_\qs J_\qs \pauliz_\qs - \sum_\qe h_\qe \paulix_\qe$.

\subsection{Classical Hamiltonian for the classical stat mech model}
Note that the above discussion directly constructs the transfer matrix from the stochastic error model without first passing through the classical Hamiltonian description of the classical statistical mechanical model. For completeness, the interactions in the classical Hamiltonian within the time $t$ slice and between the times $t$ and $t-1$ slices can be found via $H(s_t)_{\sigma_t, \sigma_{t-1}} = -\log( T(s_t)_{\sigma_t, \sigma_{t-1}} )$. 

An alternative, equivalent, and more direct construction of the classical Hamiltonian is as follows. For a given errors-followed-by-measurements round, partition the set of errors $\{ \qe \}$ into sets $\{ \qe \}_m$ of completely non-overlapping errors, where $m$ is ordered from $1$ to the number $M$ of such sets. For each error $\qe$, give an ordering to the locations of spins on which $\qe$ has support, $\mathbf{r}_{\qe,n} \in \qe$, with $n$ running from $1$ to the weight of $\qe$.
For each such set $\{ \qe \}_m$ and errors-followed-by-measurements round at time $t$, introduce a layer $t_{m}$ of $\sigma$ spins in the classical stat mech model.
Between a given layer $t_{m}$ and the next layer $t_{m+1}$ ($t_{M+1} = (t+1)_{1}$), make between-layer two-spin ferromagnetic Ising couplings of strength $K_\qe= -\frac{1}{2}\log(\tanh(\overline{K}_\qe))$, one for each error $\qe$ in the set $\{ \qe \}_m$, of the form $-K_\qe \sigma_{\mathbf{r}_{\qe,1}, t_{m}} \sigma_{\mathbf{r}_{\qe,1}, t_{m+1}} $.
For each error $\qe \in \{\qe \}_m$ (note the $m$ index), introduce the four-body couplings $-\infty \sum_{n \neq 1}  \sigma_{\mathbf{r}_{\qe,1}, t_{m}} \sigma_{\mathbf{r}_{\qe,n}, t_{m}} \sigma_{\mathbf{r}_{\qe,1}, t_{m+1}} \sigma_{\mathbf{r}_{\qe,n}, t_{m+1}}$.
Note that this infinite coupling within the Hamiltonian can be alternatively realized via factors of $\prod_{n \neq 1} \frac{1+\sigma_{\mathbf{r}_{\qe,1}, t_{m}} \sigma_{\mathbf{r}_{\qe,n}, t_{m}} \sigma_{\mathbf{r}_{\qe,1}, t_{m+1}} \sigma_{\mathbf{r}_{\qe,n}, t_{m+1}}}{2}$ multiplying the Boltzmann weights. This infinite coupling enforces correlated flipping of spins on the support of $\qe$, while the between-layer $K_\qe$ coupling encodes the probability of a correlated flipping event.

This latter form of the classical Hamiltonian is nice, as removing the sign disorder caused by $s$ makes all couplings manifestly ferromagnetic (at the cost of increasing the number of slices of $\sigma$ spins by a multiplicative factor of $M$). This means that the hypotheses of the GKS inequalities (discussed in Appendix~\ref{sec:GKS}) are satisfied, and their consequences hold. That is, removing sign disorder (relevant for taking the quantum limit) cannot decrease expectation values of products of the classical spins $\sigma$.

\section{Further discussions of disordered stat mech models \label{sec:nishimori}}

\subsection{Change of basis on the Nishimori line and uncorrelated quenched disorder \label{sec:change_of_basis}}

The random-sign disorder $s$ in $Z[s]$ (Eq.~\eqref{eq:Z[s]}) can be correlated through space and time, and its distribution can be difficult to describe.
Fortunately, there exists a change of basis of the $\sigma$ spins, with which we obtain a disordered spin model with uncorrelated disorder.

Consider the disorder average of any observable $\mathcal{O}(\sigma) = \pm 1$ that is a product of some spins.
Based on its definition, we perform a rewriting by introducing \PRXQ{dummy spin variables $\widehat{\sigma}$},
\begin{align}
    \label{eq:def_dummy_spin_sigma_hat}
    [ \avg{\mathcal{O}} ]_s &\,
    \equiv \sum_s \mathbb{P}(s) \frac{\tr_\sigma [\mathcal{O}(\sigma) \exp{-H(\sigma, s)}]}{\tr_\sigma [\exp{-H(\sigma, s)}]} \nn
    =&\ \sum_{s, \widehat{\sigma}} \mathbb{P}(\widehat{\sigma}, s) 
    \frac{\tr_\sigma [\mathcal{O}(\sigma) \exp{-H(\sigma, s)}]}{\tr_\sigma [\exp{-H(\sigma, s)}]}.
\end{align}
\PRXQ{Here we recall the definition of $\mathbb{P}(\widehat{\sigma}, s)$ in Eq.~\eqref{eq:H(sigma,s)}, and the definition $\mathbb{P}(s)$ below Eq.~\eqref{eq:Z[s]}.
Note that the spins $\widehat{\sigma}$ spins still obey the perfect SPAM conditions Eq.~\eqref{eq:SPAM_boundary_condition}, in particular that $\widehat{\sigma}_{\bfr, \ti} = +1$ for all $\bfr$.}
Making the substitution $\sigma \to \widehat{\sigma} \sigma$, we obtain
\begin{align}
    [ \avg{\mathcal{O}} ]_s &\,
    = \sum_{s, \widehat{\sigma}} \mathbb{P}(\widehat{\sigma}, s) 
    \mathcal{O}(\widehat{\sigma})
    \frac{\tr_{\sigma} [\mathcal{O}(\sigma) \exp{-H(\widehat{\sigma}\sigma, s)}]}{\tr_{\sigma} [\exp{-H(\widehat{\sigma}\sigma, s)}]} \nn
    =&\, \sum_{s, \widehat{\sigma}} \mathbb{P}(\widehat{\sigma}, s) 
    \mathcal{O}(\widehat{\sigma})
    \frac{\tr_{\sigma} [\mathcal{O}(\sigma) \exp{-\widehat{H}(\sigma, \error)}]}{\tr_{\sigma} [\exp{-\widehat{H}(\sigma, \error)}]},
\end{align}
where
\begin{align}
\begin{split}
    \label{eq:def_eta}
    \error_{\bfr, t} &= \widehat{\sigma}_{\bfr, t} \cdot \widehat{\sigma}_{\bfr, t+1}, \\
    \error_{\qs, t} &= s_{\qs, t} \cdot \prod_{\bfr \in \qs} \widehat{\sigma}_{\bfr, t},
\end{split}
\end{align}
and
\begin{align}
\label{eq:def_hat_H}
    &\ \widehat{H}(\sigma, \error) = H(\widehat{\sigma}\sigma, s) \nn
    &\quad =
    -\sum_{t} \left[ \km \sum_{\qs} \error_{\qs, t} \prod_{\bfr \in \qs} \sigma_{\bfr, t} + \kbf \sum_{\bfr} \error_{\bfr, t} \sigma_{\bfr, t+1} \sigma_{\bfr, t} \right], \\
\label{eq:def_hat_Z}
    &\ \widehat{Z}[\error] = \tr_\sigma \exp {-\widehat{H}(\sigma, \error)}.
\end{align}
\PRXQ{In this new basis, our assumption of perfect SPAM Eq.~\eqref{eq:SPAM_boundary_condition} implies that $\error_{\qs, \ti} = \error_{\qs, \tf} = +1$.
Together with condition $\widehat{\sigma}_{\bfr, \ti} = +1$ (see comments below Eq.~\eqref{eq:def_dummy_spin_sigma_hat}), we see that} $s, \widehat{\sigma}$ are uniquely determined by $\error$, and we denote $\mathcal{O}(\error) \equiv \mathcal{O}(\widehat{\sigma}(\error))$.
By the definition of $\error$, it is precisely the error events that gives rise to the history labeled by $(\widehat{\sigma}, s)$.
We have immediately
\begin{align}
    \mathbb{P}(\error)
    \equiv
    \mathbb{P}(\widehat{\sigma}, s) 
    \propto \left(\frac{p_{\rm bf}}{1-p_{\rm bf}}\right)^{\wgt(\error_\bfr)} \cdot
    \left(\frac{p_{\rm m}}{1-p_{\rm m}}\right)^{\wgt(\error_\qs)}.
\end{align}
Here, $\wgt(\cdot)$ counts the number of $-1$ components of a vector.
Summarizing, we have
\begin{align}
    [ \avg{\mathcal{O}} ]_s = [ \mathcal{O}(\error) \cdot \avg{\mathcal{O}} ]_\error.
\end{align}
On the RHS we adopted $[\avg{...}]_\error$ as a shorthand notation for $[| \avg{...}_{\widehat{Z}[\error]} |]_\error$, without introducing confusion.
We have thus obtained a disordered spin model where the distribution of $\error$ is explicitly known.
By similar reasoning, we obtain
\begin{align}
\label{eq:LJ_Zs_equals_LJ_Zeta}
    [| \avg{L_J} |]_s
    =
    [| \avg{L_J} | ]_\error.
\end{align}

Meanwhile, the same quantity can be rewritten as
\begin{align}
    &[| \avg{L_J} |]_s \nn
    =&\, \sum_s \mathbb{P}(s) \cdot \sgn(\avg{L_J}_{Z[s]}) \cdot \frac{\tr_{\widehat{\sigma}} \exp{-H(\widehat{\sigma}, s)} \cdot L_J(\widehat{\sigma})}{\exp{-H(\widehat{\sigma}, s)}} \nn
    =&\, \sum_s \mathbb{P}(s) \cdot \sgn(\avg{L_J}_{Z[s]}) \cdot \frac{\sum_{\widehat{\sigma}} \mathbb{P}(\widehat{\sigma}, s) L_J(\widehat{\sigma})}{\mathbb{P}(s)} \nn
    =&\, \sum_s \sum_{\widehat{\sigma}} \mathbb{P}(\widehat{\sigma}, s) \sgn(\avg{L_J}_{Z[s]} \cdot L_J(\widehat{\sigma})) \nn
    =&\, \sum_s \sum_{\widehat{\sigma}} \mathbb{P}(\widehat{\sigma}, s) \sgn\left(
    \frac{\tr_{\sigma} \exp{-H(\sigma, s)} L_J(\sigma) L_J(\widehat{\sigma}) }{\tr_{\sigma} \exp{-H(\sigma, s)}}
    \right) \nn
    =&\, \sum_s \sum_{\widehat{\sigma}} \mathbb{P}(\widehat{\sigma}, s) 
    \sgn\left(L_J(\widehat{\sigma})\right) 
    \sgn\left(
    \frac{\tr_{\sigma} \exp{-H(\sigma, s)} L_J(\sigma)}{\tr_{\sigma} \exp{-H(\sigma, s)}}
    \right).
\end{align}
Here, one can replace the partition functions with the probabilities only on the Nishimori line.
This expression has an explicit operational interpretation: $s, \widehat{\sigma}$ represent samples of an error history (occuring with probability $\mathbb{P}(\widehat{\sigma}, s)$), and the decoder outputs $\sgn(\avg{L_J}_{Z[s]})$ as its prediction for $L_J$.
The decoder succeeds if it agrees with the ``ground truth'' $L_J(\widehat{\sigma})$, and fails otherwise.
Therefore, $[| \avg{L_J} |]_s$ computes the difference between success and failure probabilities, as we have seen in Eq.~\eqref{eq:2Ps-1_abs_avg_LJ}.

Again making the substitution $\sigma \to \widehat{\sigma} \sigma$, we obtain
\begin{align}
    &[| \avg{L_J} |]_s \nn
    =&\, \sum_s \sum_{\widehat{\sigma}} \mathbb{P}(\widehat{\sigma}, s) \sgn\left(
    \frac{\tr_{\sigma} \exp{-H(\widehat{\sigma}\sigma, s)} L_J(\sigma)}{\tr_{\sigma} \exp{-H(\widehat{\sigma}\sigma, s)}}
    \right) \nn
    =&\, \sum_\error \mathbb{P}(\error) \sgn\left(
    \frac{\tr_{\sigma} \exp{-\widehat{H}(\sigma, \error)} L_J(\sigma)}{\tr_{\sigma} \exp{-\widehat{H}(\sigma, \error)}}
    \right) \nn
    =&\, [ \sgn(\avg{L_J}) ]_\error.
\end{align}
We have thus established that on the Nishimori line
\begin{align}
    \label{eq:sign_LJ_equals_abs_LJ}
    [ \sgn(\avg{L_J}) ]_\error = [| \avg{L_J} | ]_\error.
\end{align}
This relation was first shown by Nishimori~\cite{nishimori1993optimum}.

\subsection{Divergence of free energy from success of Inference Problem~\ref{IP2} on index set $J$ \label{sec:divergence_of_free_energy}}

The success of 
Inference Problem~\ref{IP2} on index set $J$ implies $\Delta \mathbb{P}^J_{\decoder_J^{\rm ML}} \to 1$,
and $[ \sgn(\avg{L_J}) ]_\error = [| \avg{L_J}| ]_\error \to 1$, compare Eqs.~(\ref{eq:2Ps-1_abs_avg_LJ}, \ref{eq:LJ_Zs_equals_LJ_Zeta}, \ref{eq:sign_LJ_equals_abs_LJ}).
Define 
\begin{align}
    \widehat{Z}[\error](L_J=\pm 1) \equiv
    \frac{\tr_{\sigma} \exp{-\widehat{H}(\sigma, \error)} 
    \cdot \mathbb{1}_{L_J(\sigma) = \pm 1}}{\tr_{\sigma} \exp{-\widehat{H}(\sigma, \error)}}.
\end{align}
That $[| \avg{L_J} | ]_\error \to 1$ implies that for each $\error$, either $\widehat{Z}[\error](L_J=+ 1) / \widehat{Z}[\error] \to 1$ or $\widehat{Z}[\error](L_J=- 1) / \widehat{Z}[\error] \to 1$; and $[ \sgn(\avg{L_J}) ]_\error \to 1$ implies that $\widehat{Z}[\error](L_J=+1) \geq \widehat{Z}[\error](L_J=-1)$ with probability $1$.
Together, we conclude that with probability $1$, $\widehat{Z}[\error](L_J=+ 1) / \widehat{Z}[\error] \to 1$.
Consequently,
\begin{align}
    & [| \avg{L_J} | ]_\error - [ \avg{L_J} ]_\error \nn
    &= 2\cdot \sum_\error \mathbb{P}(\error) 
    \cdot 
    \mathbb{1}_{\avg{L_J} < 0}
    \cdot
    | \avg{L_J} |\nn
    &\leq 2\cdot \sum_\error \mathbb{P}(\error) 
    \cdot 
    \mathbb{1}_{\avg{L_J} < 0} \nn
    &= 1 - \left[\mathrm{sgn} \avg{L_J} \right]_\error \to 0.
\end{align}
From this, we see that whenever $\Delta \mathbb{P}^J_{\decoder_J^{\rm ML}} \to 1$, we also have
\begin{align}
    [ \avg{L_J} ]_\error
    = \left[
        \frac{1 - e^{-\Delta F_{J,\error}}}{1+e^{-\Delta F_{J,\error}}}
    \right]_\error \to 1,
\end{align}
where
\begin{align}
    e^{-\Delta F_{J,\error}} = \frac{\widehat{Z}[\error](L_J=-1)}{\widehat{Z}[\error](L_J=+1)}.
\end{align}
Finally, by Jensen's inequality,
\begin{align}
    \left[
       \Delta F_{J,\error}
    \right]_\error
    =
    \left[
        -\ln e^{-\Delta F_{J,\error}}
    \right]_\error 
    \geq
    -\ln
    \left[
       e^{-\Delta F_{J,\error}}
    \right]_\error \to \infty.
\end{align}

\section{GKS inequalities \label{sec:GKS}}

The GKS inequalities are due to Griffiths~\cite{griffiths1967a_correlations, griffiths1967b_correlations, griffiths1967c_correlations}, Kelly and Sherman~\cite{kelly1968general}.
They make precise the intuition that stronger ferromagnetic couplings correspond to more pronounced magnetic order.
Here, we follow the treatment in Ref.~\cite{friedli_velenik_2017}.

\begin{theorem}[Theorem 3.49 of~\cite{friedli_velenik_2017}]
\label{thm:GKS}
Let $\mathcal{G}$ be any graph.
To each vertex $j \in \mathcal{G}$ we associate an Ising spin $\sigma_j$, and to each subset $C \subseteq \mathcal{G}$ we associate a non-negative ``coupling constant'' $K_C \ge 0$.
Consider the following model of Ising spins
\begin{align}
    Z_K \coloneqq \tr_{\sigma} \exp \left[ \sum_C K_C \sigma_C \right],
\end{align}
where we define $\sigma_C = \prod_{j \in C} \sigma_j$.
Then we have 
\begin{align}
\label{eq:GKS_1}
    \forall A \subseteq \mathcal{G}, \quad \langle \sigma_A \rangle_{Z_K} \ge 0.
\end{align}
\end{theorem}

\begin{corollary}[Exercise 3.31 of~\cite{friedli_velenik_2017}]
Let $K$ and $K'$ be two sets of coupling constants satisfying $K_C \geq |K'_C|$ for all $C \subseteq \mathcal{G}$.
We have
\begin{align}
\label{eq:GKS_2}
    \forall A \subseteq \mathcal{G}, \quad \langle \sigma_A \rangle_{Z_K} \ge \left| \langle \sigma_A \rangle_{Z_{K'}} \right|.
\end{align}
\end{corollary}

\section{Tunneling matrix elements in the $\pauliz$ basis}\label{sec:zbasis}
In Sec.~\ref{sec:psucc_to_observable} of the main text, we noted that successful decoding corresponded to vanishing ``tunneling'' matrix elements and exponentially small (in distance) splittings between ground states in different logical-$\paulix$ sectors.
We argued for the small splittings within that section, and here we complete the details of the vanishing ``tunneling'' matrix elements in the $\pauliz$ basis.

First, consider the effects of successful decoding \YL{of Inference Problem~\ref{IP1}} on matrix elements in the disordered classical stat mech model.
Recalling the definition of the success probability of decoding for decoder $\decoder$:
\begin{equation}
\begin{split}\label{eq:appz_mldecoder}
    \mathbb{P}^{\rm all}_\decoder
    \equiv&\ 
    \mathbb{P}(\forall j\in [K],\ \decoder_j(s) = L_j(\sigma_{\tf})) \\
    =&\ \sum_{s} \mathbb{P}(s)  \sum_{\sigma} \mathbb{P}(\sigma|s) \prod_{j=1}^K \frac{1+\decoder_j(s) \cdot L_j(\sigma_{\tf})}{2}.
\end{split}
\end{equation}
Suppose that we are below threshold with decoder $\decoder$.
The statement $\mathbb{P}^{\rm all}_\decoder \to 1$ is equivalent 
that in the corresponding disordered stat mech model
\begin{equation}
    \mathbb{P}^{\rm all}_\decoder = \left[ \left \langle \prod_{j=1}^K \frac{1+\decoder_j(s) \cdot L_j(\sigma_{\tf})}{2} \right \rangle_{Z[s]} \right ]_s \to 1.
\end{equation}
However, notice that for any $s$ disorder realization, the inner thermal expectation value bounds the corresponding quantity in the clean model:
\begin{equation}
\begin{split}
    &\left \langle \prod_{j=1}^K \frac{1+\decoder_j(s) \cdot L_j(\sigma_{\tf})}{2} \right \rangle_{Z[s]}  \\&=  \frac{1}{2^K} \sum_{J \subseteq [K]} \left( \prod_{j \in J} \decoder_j(s) \right) \langle L_J \rangle_{Z[s]} 
    \\& \leq  \frac{1}{2^K} \sum_{J \subseteq [K]} \left| \langle L_J \rangle_{Z[s]} \right|
    \\ & \leq \frac{1}{2^K} \sum_{J \subseteq [K]} \langle L_J \rangle
    \\& = \left \langle \prod_{j=1}^K \frac{1+L_j(\sigma_{\tf})}{2} \right \rangle
\end{split}
\end{equation}
The first equality used linearity of expectation to expand the product as a sum over all $2^K$ choices of $J$. The following inequality is merely $\sum_i x_i \leq \sum_i |x_i|$. The following inequality is the GKS inequality noted in Appendix~\ref{sec:GKS}, and uses the notation in the main text that $\langle L_J \rangle$ without an outer disorder average (and without a $Z[s]$ or $Z[\error]$ label) corresponds to a thermal average in the clean model. The last equality is the reverse of the first. 

This chain of inequalities implies that whenever a decoder is below threshold (and in particular, whenever the maximal-likelihood decoder is below threshold), 
\begin{equation}
    \left \langle \prod_{j=1}^K \frac{1+L_j(\sigma_{\tf})}{2} \right \rangle \to 1
\end{equation}
Note also that $\sum_{ \{ \ell_j \} s.t. \ell_j \in  \{ 1,-1 \}} \prod_{j=1}^K \frac{1+ \ell_j L_j(\sigma_{\tf})}{2} = 1$. Since every term in that sum is non-negative, and the term with $\ell_j=1$ for all $j$ is asymptotically $1$, the remaining terms must all asymptotically vanish. That is,
\begin{equation}
    \left \langle \prod_{j=1}^K \frac{1+\ell_j L_j(\sigma_{\tf})}{2} \right \rangle     \to
    \begin{cases}
        1, \quad \forall j,\,\ell_j=1\\
        0, \quad\text{otherwise}
    \end{cases}.
\end{equation}
Rephrasing this result in the language of the auxiliary Hilbert space gives the desired claim:
\begin{align}
    \frac{
    \bra{\bm{\ell}^\pauliz_L}
    T(\tf)
    \ket{\mathbf{0}^\pauliz_L}}
    {
    \sum_{\bm{\ell'} \subseteq [K]}
    \bra{\bm{\ell'}^\pauliz_L}
    T(\tf)
    \ket{\mathbf{0}^\pauliz_L}}
    \to
    \begin{cases}
        1, \quad&\bm{\ell} = \mathbf{0}\\
        0, \quad&\bm{\ell} \neq \mathbf{0}
    \end{cases}.
\end{align}

\pagebreak
\begin{widetext}
\section{Relating $\avg{\PRXQ{L_W}}$ to two point correlation functions in the ``hard wall with two punctures'' stability experiment \label{sec:two_punctures_two_points}}

Here we derive Eq.~\eqref{eq:W_temporal_corr} in the main text.
Starting from Eq.~\eqref{eq:H_excluded_W} and defining $\tanh \lambda = \frac{p_0}{1-p_0}$,
we have
\begin{align}
    \langle \PRXQ{L_W} \rangle =&\, \frac{
        \tr{\Pi_\Omega^{W} \cdot e^{-\lambda \paulix_0} \cdot e^{-H^W_{\rm q} \tf} \cdot \PRXQ{L_W} \cdot e^{-\lambda \paulix_0}}
    }{
        \tr{\Pi_\Omega^{W} \cdot e^{-\lambda \paulix_0} \cdot e^{-H^W_{\rm q} \tf} \cdot e^{-\lambda \paulix_0}}
    }
    =
    1 - \frac{2 \cdot \sum_{n}' \bra{n} e^{-\lambda \paulix_0} \cdot \Pi_\Omega^{W} \cdot e^{-\lambda \paulix_0} \ket{n} \cdot e^{-E^W_n \tf}}{\sum_{n} \bra{n} e^{-\lambda \paulix_0} \cdot \Pi_\Omega^{W} \cdot e^{-\lambda \paulix_0} \ket{n}  \cdot e^{-E^W_n \tf}}.
\end{align}
Here, $\Pi_\Omega^{W} \propto \lim_{t \to \infty} e^{-H^W_{\rm q} \cdot t}$ is the projection operator onto the ground space of $H^W_{\rm q}$, $\sum_{n}'$ sums over eigenstates of $H^W_{\rm q}$ for which $\bra{n} W \ket{n} = -1$, whereas $\sum_{n}$ sums over all eigenstates of $H^W_{\rm q}$.
Expanding in $\tanh \lambda$ we get
\begin{align}
    \langle \PRXQ{L_W} \rangle
    =&\, \frac{\sum_{n} \bra{n} \Pi_\Omega^{W} \ket{n}  \cdot e^{-E^W_n \tf} - (\tanh \lambda)^2 \sum_{n}' \bra{n} \paulix_0 \cdot \Pi_\Omega^{W} \cdot \paulix_0 \ket{n} \cdot e^{-E^W_n \tf}}{\sum_{n} \bra{n} \Pi_\Omega^{W} \ket{n}  \cdot e^{-E^W_n \tf} +  (\tanh \lambda)^2 \sum_{n}' \bra{n} \paulix_0 \cdot \Pi_\Omega^{W} \cdot \paulix_0 \ket{n}  \cdot e^{-E^W_n \tf}} \nn
    =&\, \frac{1  - (\tanh \lambda)^2 \frac{ \sum_{n} \bra{n} \paulix_0 \cdot \Pi_\Omega^{W} \cdot \paulix_0 \ket{n} \cdot e^{-E^W_n \tf} }{\sum_{n} \bra{n} \Pi_\Omega^{W} \ket{n}  \cdot e^{-E^W_n \tf}}}{1 + (\tanh \lambda)^2 \frac{ \sum_{n} \bra{n} \paulix_0 \cdot \Pi_\Omega^{W} \cdot \paulix_0 \ket{n} \cdot e^{-E^W_n \tf} }{\sum_{n} \bra{n} \Pi_\Omega^{W} \ket{n}  \cdot e^{-E^W_n \tf}}}.
\end{align}
Here, we used the following superselection rules
\begin{align}
    \label{eq:superselection_W}
    W \ket{n} = -\ket{n} %
    \quad \Rightarrow& \quad
    \bra{n} \Pi_\Omega^{W} \ket{n} = \bra{n} \paulix_0 \Pi_\Omega^{W} \ket{n} = \bra{n} \Pi_\Omega^{W} \paulix_0 \ket{n} = 0, \\
    W \ket{n} = +\ket{n} %
    \quad \Rightarrow& \quad
    \bra{n} \paulix_0 \Pi_\Omega^{W} \paulix_0 \ket{n} = \bra{n} \paulix_0 \Pi_\Omega^{W} \ket{n} = \bra{n} \Pi_\Omega^{W} \paulix_0 \ket{n} = 0.
\end{align}
Using these identities we can identify a correlation function of the imaginary time path integral,
\begin{align}
    \label{eq:spectral_decomp_2pt_corr_fcn}
    \langle \paulix_0(\tf) \paulix_0(0) \rangle
    =&\,
    \frac{\tr (\Pi_\Omega^{W} \cdot \paulix_0 \cdot  e^{-H^W_{\rm q} \tf} \cdot \paulix_0)}{\tr (\Pi_\Omega^{W} \cdot e^{-H^W_{\rm q} \tf})} 
    =
    \frac{\sum_{n} \bra{n} \paulix_0 \cdot \Pi_\Omega^{W} \cdot \paulix_0 \ket{n} \cdot e^{-E^W_n \tf}}{\sum_{n} \bra{n} \Pi_\Omega^{W} \ket{n}  \cdot e^{-E^W_n \tf}} 
    \geq
    {\bra{1} \paulix_0 \cdot \Pi_\Omega^{W} \cdot \paulix_0 \ket{1}}
     \cdot e^{-\Delta_{01}^W \cdot \tf},
\end{align}
\end{widetext}
where $\ket{1}$ is the lowest-lying state in the sector $\PRXQ{L_W} = -1$, and $\Delta^W_{01}$ the difference in energy between this state and the ground space of $H^W_{\rm q}$.
By definition, the correlation function $\langle \paulix_0(\tf) \paulix_0(0) \rangle$ is non-negative.
Therefore,
\begin{align}
    \label{eq:W_temporal_corr_appendix}
    \langle \PRXQ{L_W} \rangle
    =&\, \frac{
        1 - (\tanh \lambda)^2 \langle \paulix_0(\tf) \paulix_0(0) \rangle
    }{
        1 + (\tanh \lambda)^2 \langle \paulix_0(\tf) \paulix_0(0) \rangle
    } \nn
    \leq&\, 1 - (\tanh \lambda)^2 \langle \paulix_0(\tf) \paulix_0(0) \rangle,
\end{align}
where we used $\langle \paulix_0(\tf) \paulix_0(0) \rangle \geq 0$. 
The analog of Eq.~\eqref{eq:temporal_decoding_lower_bound_2DTC} for general LDPC codes now reads (see Eq.~\eqref{eq:pfail_hard_wall_stability_two_punctures_dennis})
\begin{align}
    \label{eq:W_temporal_Psucc}
    1 - \langle \PRXQ{L_W} \rangle =  O(e^{-\xi^{-1} \cdot \tf}).
\end{align}
Notice that we no longer have the  entropic factor $\mathrm{area}(W)$ as compared to Eq.~\eqref{eq:temporal_decoding_lower_bound_2DTC}, since we have fixed two points of the flux tube, see Fig.~\ref{fig:stability-experiment-2pt}.
Comparing Eq.~\eqref{eq:W_temporal_corr_appendix} and Eq.~\eqref{eq:W_temporal_Psucc}, we have
\begin{align}
    \label{eq:2pt_correlation_function_exp_decay}
    \langle \paulix_0(\tf) \paulix_0(0) \rangle = O( e^{-\xi^{-1} \cdot \tf}).
\end{align}
Thus, the stability experiment \textit{generically} allows us to probe temporal correlation functions between $\paulix_0$ (which are creation operators of excitations that flip the sign of $\PRXQ{L_W}$), and with decoding success we conclude the exponential-in-time decay of the correlation function.
More precisely, 
assuming the $\tf$-independent matrix element $\bra{1} \paulix_0 \cdot \Pi_\Omega^{W} \cdot \paulix_0 \ket{1}$ in Eq.~\eqref{eq:spectral_decomp_2pt_corr_fcn} does not vanish for any finite system size,\footnote{In general, we expect $\bra{1} \paulix_0 \cdot \Pi_\Omega^{W} \cdot \paulix_0 \ket{1}$ may decrease as a function (inverse polynomial) of system size. This is because the transverse field gives a dispersion to excitations, and flipping a single site of an exact ground state will generically give a superposition of excited states within a band of excitations. Nevertheless, $\Delta_{10}^W = \Omega(\xi^{-1})$ will still follow. }
Eq.~\eqref{eq:2pt_correlation_function_exp_decay} implies that $\Delta_{10}^W = \Omega(\xi^{-1})$.
Therefore we have a finite gap
to excitations created by $\paulix_0$ (which necessarily have $\PRXQ{L_W} = -1$, see Eq.~\eqref{eq:superselection_W}), in the modified Hamiltonian $H^W_{\rm q}$ of Eq.~\eqref{eq:H_excluded_W}.\footnote{\label{fn:connected_correlation_function}More precisely, the gap is associated to the inverse correlation time of the \textit{connected} correlation function $\avg{\paulix_0(\tf) \paulix_0(0)}_c \equiv \avg{\paulix_0(\tf) \paulix_0(0)} - \avg{\paulix_0(\tf)} \avg{\paulix_0(0)}$.
(When we have analytic control over the correlation functions, Fourier transforming into the frequency domain and analytic continuation to imaginary frequencies allows direct access to the spectral function.)
Here, we have effectively set $\avg{\paulix_0(\tf)}=\avg{\paulix_0(0)}=0$ (therefore $\avg{\paulix_0(\tf) \paulix_0(0)}_c = \avg{\paulix_0(\tf) \paulix_0(0)}$) by imposing $W$ as an exact symmetry in the hard wall construction.
It is generally difficult to directly measure connected correlation functions with decoding experiments, and this difficulty necessitates the hard wall.
Other methods of subtracting the disconnected correlation function exist for the 2D toric code, as we discuss in Appendix~\ref{sec:stability_experiment_special_to_surface_codes}.}

\section{Stability experiments for 2D Euclidean and hyperbolic toric codes\label{sec:stability_experiment_special_to_surface_codes}}

Here we analyze a ``flux threading'' version of the stability experiment for the 2D toric code that are closer to its original formulation~\cite{Gidney_2022}.
They make use of a global redundancy special to the code, namely all $\pauliz$ stabilizers multiply to identity (i.e. the system is in the even-parity sector of anyons).
This way, the modification of the code Hamiltonian on $W$ in Eq.~\eqref{eq:H_excluded_W} can be avoided.

Our discussion here applies to 2D toric codes in both Euclidean and hyperbolic planes~\cite{Breuckmann_2016, plackebreuckmann2023hyperbolicIsing}, as we only require (i) global anyon parity conservation and (ii) LDPC conditions.
The hyperbolic toric code is particularly interesting to us, since it has a code distance $D(N) = \Theta(\log N)$ just enough to have a nonzero threshold (c.f. Theorem~\ref{thm:ordering_of_clean_stat_mech_model}), while having a finite code rate $K = \Theta(N)$.
The code is economical, and at the same time has exponentially many failure modes, therefore least stable.

\subsection{Threading a single flux\label{sec:stability_hyperbolic_surface_code}}

\begin{figure}[b]
    \centering
    \includegraphics[width=\linewidth]{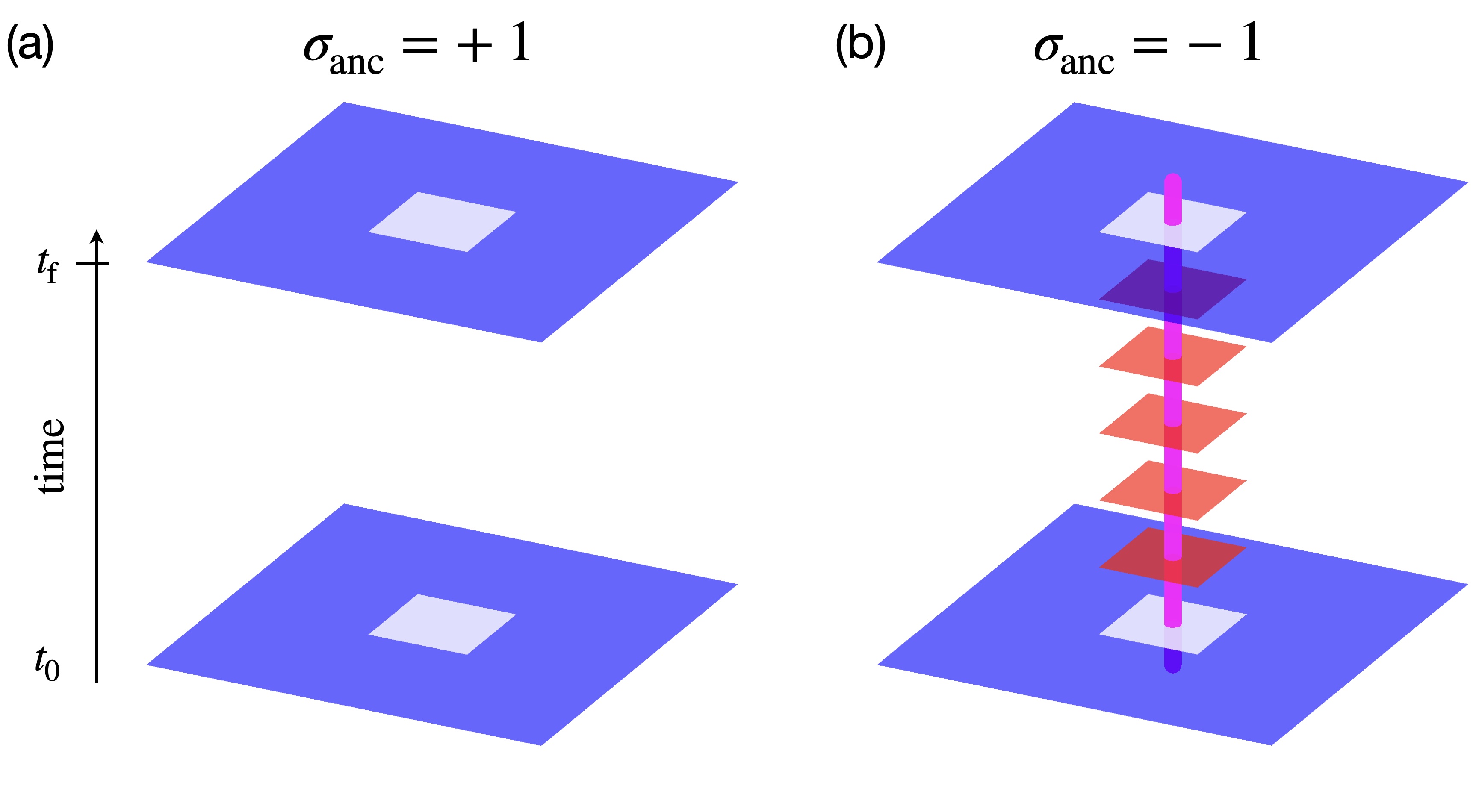}
    \caption{``Flux threading'' stability experiment for 2D toric codes.
    This construction works in Euclidean or hyperbolic spaces.
    The two panels represent $Z_{\sigma_{\rm anc}=\pm 1}$ in Eq.~\eqref{eq:def_Z_single_flux}, respectively.
    Here, the red plaquettes represent a column antiferromagnetic $\km$ couplings, introducing a flux tube (magenta) of unsatisfied plaquette terms (which cost energy).
    The flux tube is allowed wiggle (they do not necessarily track the red plaquettes), as long as its endpoints are pinned.
    }
    \label{fig:single-flux-threading}
\end{figure}

Consider initializing the qubits in the code space, and start running the error model for $t \in [0, \tf]$.
We introduce an additional ancilla qubit $\sigma_{\rm anc}$, which is initialized in the state $\ket{0}$ or $\ket{1}$ with equal probability.
This ancilla qubit does not experience any bitflip errors.
We choose a particular stabilizer $\qs_0$ and make the following modification,
\begin{align}
    \pauliz_{\qs_0} = \prod_{\bfr \in \qs_0} \pauliz_\bfr 
    \quad \to \quad
    \pauliz_{\sigma_{\rm anc}} \cdot \prod_{\bfr \in \qs_0} \pauliz_\bfr.
\end{align}
All other stabilizers are unchanged.
At $t=0$ and $t=\tf$, we assume that all stabilizer measurements are perfect, except for $\qs_0$.
Between $t=0$ and $t=\tf$, there are bitflip errors and syndrome errors can occur on any stabilizer measurement, as usual.
Now, the task of the decoder is to predict the state of $\sigma_{\rm anc}$.

The state of $\sigma_{\rm anc}$ is a conserved quantity, and can in principle be recovered by multiplying all the stabilizer measurements in the limit $\km \to \infty$.
Our mapping to stat mech models (see Sec.~\ref{sec:stat_mech_model}) immediately applies here, allowing us to relate the success probability of this stability experiment to $\avg{\sigma_{\rm anc}}$ as follows,
\begin{align}
    \mathbb{P}^{\sigma_{\rm anc}}_{\decoder^{\rm ML}} = 1-\epsilon(\tf) \quad \Rightarrow \quad 
    \avg{\sigma_{\rm anc}} \geq 1-2\epsilon(\tf),
\end{align}
where the expectation value is taken with respect to the following stat mech model
\begin{widetext}
\begin{align}
    \label{eq:def_Z_single_flux}
    Z =&\ \tr_{\sigma, \sigma_{\rm anc}} \exp{
        \sum_{t=0}^{\tf}\left[ \km
        \left(
            \sigma_{\rm anc} \cdot \prod_{\bfr \in \qs_0} \sigma_{\bfr, t} + \sum_{\qs \neq \qs_0} \prod_{\bfr \in \qs} \sigma_{\bfr, t}
        \right)
        + \kbf \sum_{\bfr} \sigma_{\bfr, t+1} \sigma_{\bfr, t}
        \right]
    } \nn
    \equiv&\ Z_{\sigma_{\rm anc}=+1} + Z_{\sigma_{\rm anc}=-1}.
\end{align}
\end{widetext}
Here, the boundary conditions are $\sigma_{\ti} = +1$, and $\km \to \infty$ at $t=\ti, \tf$ for all $\qs \neq \qs_0$.
The partition function $Z_{\sigma_{\rm anc}=+1}$ is the same as in Eq.~\eqref{eq:Z_clean_spin_model} and has transfer matrix generated by $H_{\rm q}$, whereas $Z_{\sigma_{\rm anc}=-1}$ has a column on $\qs_0$ of antiferromagnetic couplings $\km \to -\km$, as we illustrate in Fig.~\ref{fig:single-flux-threading}.
Compared to Figs.~\ref{fig:stability-experiment},\ref{fig:stability-experiment-2pt}, the antiferromagnetic coupling introduces a single flux line, rather than a flux loop.
Again defining $e^{-\delta F} = Z_{\sigma_{\rm anc}=-1}/Z_{\sigma_{\rm anc}=+1}$, we have $\avg{\sigma_{\rm anc}} = \frac{1-e^{-\delta F}}{1+e^{-\delta F}}$, and by similar reasoning (see Appendix~\ref{sec:pfail_upperbound_MW})
\begin{align}
    \label{eq:pfail_upperbound_single_flux_threading}
    \epsilon(\tf) = O\left( [(w+1) \cdot \widetilde{q}]^{\tf} \right)
    \quad \Rightarrow \quad 
    \delta F = \Omega(\xi^{-1} \tf).
\end{align}
We obtain a finite lower bound to the gap of the magnetic anyon, for the Hamiltonian $H_{\rm q}$ without modification.

When naively extending this construction  to the 3D toric code, we see that the prediction of $\sigma_{\rm anc}$ always succeeds after $O(1)$ rounds using a majority vote decoder as long as $p_{\rm m} < 1/2$, due to local relations of stabilizers, see footnote~\ref{fn:majority_vote}.
Constructions in Fig.~\ref{fig:hardwall-2pt-1DREP-3DTC}(b), where the Hamiltonian is modified on the hard wall, can be used to avoid such situations.
Therefore, the flux threading construction is somewhat specific to 2D codes, as it requires the presence of a global relation, as well as the absense of local relations.

\subsection{Threading two fluxes\label{sec:stability_hyperbolic_surface_code_two_fluxes}}

\begin{figure}[t]
    \centering
    \includegraphics[width=\linewidth]{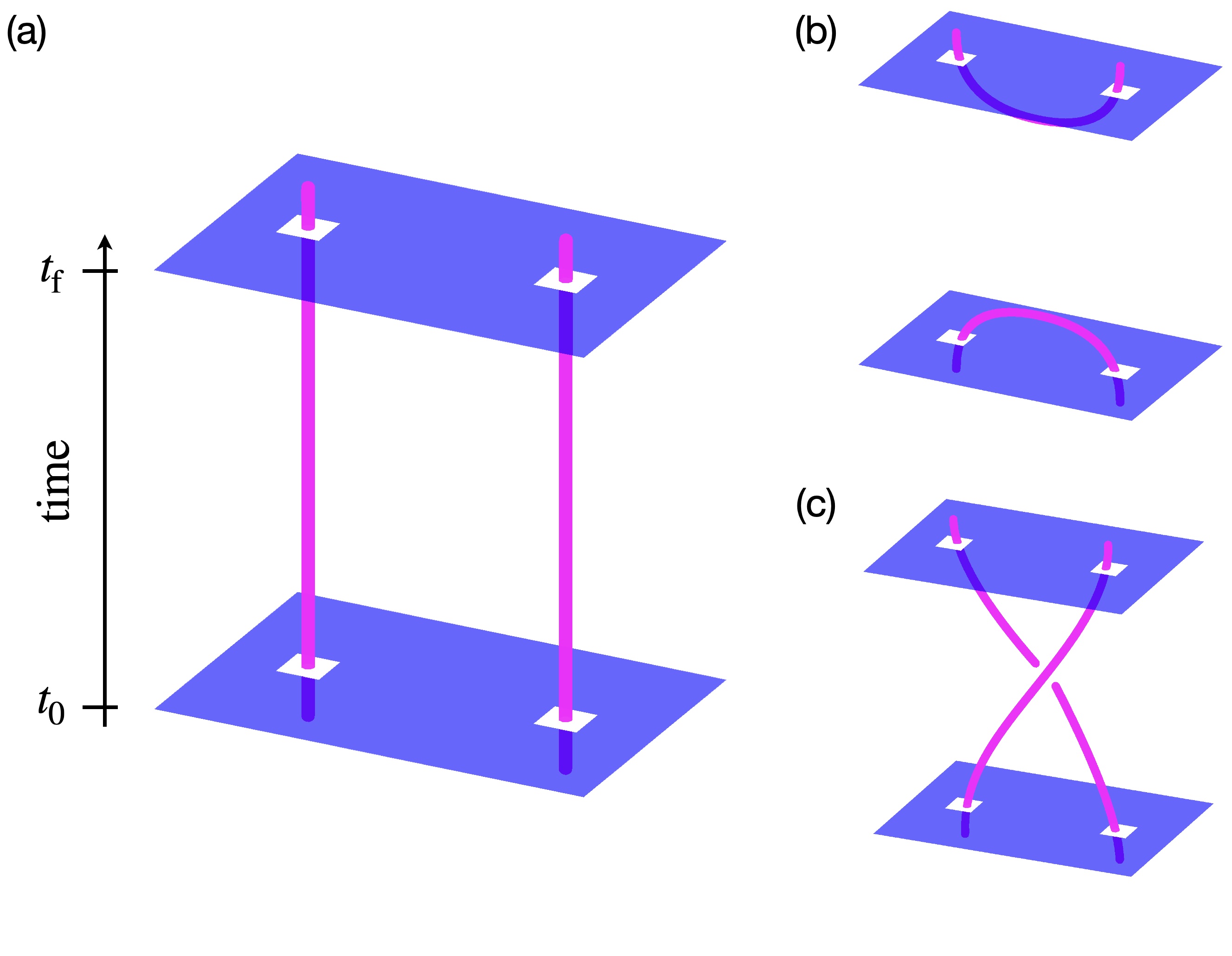}
    \caption{``Double flux threading'' stability experiment for 2D toric codes.
    Here, we are only plotting configurations of $Z_{\sigma_{\rm anc}=-1}$. The antiferromagnetic couplings always go straight in the temporal direction (similarly to Eq.~\eqref{eq:def_Z_single_flux} and Fig.~\ref{fig:single-flux-threading}(b)), which we omit drawing.
    The flux lines (magenta) represent unsatisfied plaquette terms (which cost energy).
    Contribution from (b) and (c) vanish when taking the distance between $\qs_{1,2}$ to infinity.
    }
    \label{fig:double-flux-threading}
\end{figure}

The single-anyon state discussed above is not within the Hilbert space of the qubits, as the anyons are guaranteed to have even total parity.
Here, we devise a different experiment that measures the gap of a two-anyon state, and directly connect the success of decoding to the spectrum of the transfer matrix.
The construction is almost identical to the one above, except that we now couple the ancilla $\sigma_{\rm anc}$ to two stabilizers $\qs_{1,2}$, as follows,
\begin{align}
    \pauliz_{\qs_{1,2}} = \prod_{\bfr \in \qs_{1,2}} \pauliz_\bfr 
    \quad \to \quad
    \pauliz_{\sigma_{\rm anc}} \cdot \prod_{\bfr \in \qs_{1,2}} \pauliz_\bfr.
\end{align}
The expectation value $\sigma_{\rm anc}$ can be similarly written as
\begin{align}
    \avg{\sigma_{\rm anc}} = \frac{1-e^{-\delta F}}{1+e^{-\delta F}}, 
    \quad 
    e^{-\delta F} = \frac{Z_{\sigma_{\rm anc}=-1}}{Z_{\sigma_{\rm anc}=+1}},
\end{align}
where $Z_{\sigma_{\rm anc}=+1}$ is identical to the one in Eq.~\eqref{eq:def_Z_single_flux}, and in $Z_{\sigma_{\rm anc}=-1}$ there are two columns of antiferromagnetic couplings $\km \to -\km$, as we illustrate in Fig.~\ref{fig:double-flux-threading}.
In a manner similar to Eq.~\eqref{eq:discrete_time_path_integral}, we can write
\begin{equation}
\begin{split}
    Z_{\sigma_{\rm anc}=+1} \propto&\, \sum_{\sigma_{\tf}}  \bra{\sigma_\tf} 
    \Pi_\pauliz
    \Pi_\paulix
     \cdot e^{-\tf H_{\rm q}}
    \cdot
    \Pi_\paulix
    \ket{\mathbf{0}}, \\
    Z_{\sigma_{\rm anc}=-1} \propto&\, \sum_{\sigma_{\tf}}  \bra{\sigma_\tf} 
    \Pi_\pauliz
    \Pi_\paulix
     \cdot e^{-\tf H'_{\rm q}}
    \cdot
    \Pi_\paulix
    \ket{\mathbf{0}},
\end{split}
\end{equation}
where $H'_{\rm q} = \mathcal{O} H_{\rm q} \mathcal{O}$, with $\mathcal{O}$ is an ``anyon creation'' operator composed of Pauli $\paulix$ matrices that anticommutes with $\pauliz_{\qs_{1,2}}$ but commutes with all other stabilizers.
Using Eq.~\eqref{eq:boundary_state_projection}, we obtain
\begin{align}
    e^{-\delta F} = 
    \frac{
        \sum_{\bm{\ell}\in \{0,1\}^{[K]}}
        \bra{\mathbf{0}^\paulix_L} 
            \mathcal{O}
            \cdot
            e^{-\tf H_{\rm q}}
            \cdot
            \mathcal{O}
        \ket{\bm{\ell}^\paulix_L}
    }
    {
        \sum_{\bm{\ell}\in \{0,1\}^{[K]}}
        \bra{\mathbf{0}^\paulix_L} 
            e^{-\tf H_{\rm q}}
        \ket{\bm{\ell}^\paulix_L}
    }.
\end{align}
Observing that all terms in the numerator commute with $\paulix$ logical operators, only terms diagonal in the code space survive,
\begin{align}
e^{-\delta F} &= 
    \frac{
        \bra{\mathbf{0}^\paulix_L} 
            \mathcal{O}
            \cdot
            e^{-\tf H_{\rm q}}
            \cdot
            \mathcal{O}
        \ket{\mathbf{0}^\paulix_L} 
    }
    {
        \bra{\mathbf{0}^\paulix_L} 
            e^{-\tf H_{\rm q}}
        \ket{\mathbf{0}^\paulix_L} 
    } \nn
    &\equiv
    \frac{
        \bra{\Psi} 
            e^{-\tf H_{\rm q}}
        \ket{\Psi} 
    }
    {
        \bra{\mathbf{0}^\paulix_L} 
            e^{-\tf H_{\rm q}}
        \ket{\mathbf{0}^\paulix_L} 
    }.
\end{align}
Here we define $\ket{\Psi} = \mathcal{O} \ket{\mathbf{0}^\paulix_L}$.
Performing a spectral decomposition of $H_{\rm q}$, we have
\begin{widetext}
\begin{align}
    e^{-\delta F} \geq \frac{1}{
        \bra{\mathbf{0}^\paulix_L} 
        \Pi_\Omega^{}
        \ket{\mathbf{0}^\paulix_L} }
        \(
            \bra{\Psi} 
                \Pi_\Omega^{}
            \ket{\Psi} 
            + 
            e^{-\Delta \cdot \tf}
            \bra{\Psi} 
                (1-\Pi_\Omega^{})
            \ket{\Psi} 
        \)
        \equiv A + B \cdot e^{-\Delta \cdot \tf}.
\end{align}
\end{widetext}
Here, $A, B \geq 0$ are constants that do not depend on $\tf$, but they may depend on $N$.
On the other hand, we show in Appendix~\ref{sec:pfail_upperbound_MW} that when the distance between $\qs_{1,2}$ diverges, the failure probability of predicting $\sigma_{\rm anc}$ has an explicit upper bound (see Eq.~\eqref{eq:pfail_upperbound_doule_flux_threading_dennis} of Appendix~\ref{sec:pfail_upperbound_MW})
\begin{align}
    \label{eq:pfail_upperbound_doule_flux_threading}
    \epsilon(\tf) = O\left( [(w+1) \cdot \widetilde{q}]^{2 \tf} \right).
\end{align}
Comparing this with Eq.~\eqref{eq:pfail_upperbound_single_flux_threading}, we see that the exponent is doubled (from $\tf$ to $2\tf$).
This can be intuitively understood as threading two fluxes instead of one, and when the distance between $\qs_{1,2}$ is infinite, the two fluxes are essentially independent of each other. 
Using
\begin{align}
    1 - e^{-\delta F} \geq \avg{\sigma_{\rm anc}} \geq 1-2\epsilon(\tf),
\end{align}
we conclude that as $\tf \to \infty$
\begin{align}
    A + B \cdot e^{-\Delta \cdot \tf} \leq 2\epsilon(\tf)
    \quad \Rightarrow \quad
    \Delta \geq \xi^{-1}.
\end{align}
We have thus shown a finite gap to excitations in the sector of $\ket{\mathbf{0}^\paulix_L}$ (where all logical $\paulix$ operators have eigenvalue $+1$).\footnote{\label{fn:gap_in_other_sectors}Additional assumptions must be made in order to conclude that there must also be a finite gap in all other sectors.}

As we have emphasized, it is necessary in this case to take the distance between $\qs_{1,2}$ to infinity, so that the two fluxes inserted do not interact with each other.
When the distance is finite, we get $O(1)$, $\tf$-independent contributions to $\epsilon(\tf)$ on top of the exponential-in-$\tf$ decay, making it difficult to extract the correlation time.
These contributions are precisely the disconnected part of the correlation function mentioned in footnote~\ref{fn:connected_correlation_function}, but they do not affect the gap in any way.
Taking the infinite-distance limit between $\qs_{1,2}$ is one way of subtracting off the disconnected part; the subtraction can also be achieved by imposing the hard wall (see footnote~\ref{fn:connected_correlation_function}) or by threading a single flux, as we have previously discussed.

\section{Explicit upper bound on decoding failure probabilities for LDPC codes 
\label{sec:pfail_upperbound_MW}}


We derive explicit upper bounds for failure probabilities of decoding experiments discussed in Sec.~\ref{sec:decoding_success_physical_consequences}.
Our proofs follow the combinatorial approaches in~\cite{DKLP2001topologicalQmemory, pryadko2014, breuckmann2018phd}: 
if a minimal weight decoder makes an error, it contains physical errors along at least half a logical error. By counting the types of physical errors, we can upper bound the failure probability of the minimum weight decoder.
We formulate the proofs in a way that is applicable to both memory and stability experiments.

Given the spacetime lattice of qubits (or classical spins of the stat mech models) $\mathcal{G} = [N] \times [\tf]$ ($\tf \in \mathbb{Z}$) and the $Z$ checks $\qs$ of the code, we define a classical linear code as follows.
\begin{itemize}
\item
To each $Z$ check $\qs$ we associate $\tf$ classical checks, labeled $\cc = (\qs, t+1/2)$ for $0 \leq t \leq \tf-1$.
\item 
To each $Z$ check $\qs$ we associate $\tf+1$ binary bits, labeled $\vertex = (\qs, t)$ for $0 \leq t \leq \tf$.
We say $\vertex=(\qs, t)$ is incident to $\cc=(\qs', t'+1/2)$ if $\qs = \qs'$ and $t=t'$ or $t=t'+1$.
We write this as $\vertex \in \cc$.
\item
To each qubit $\bfr$ we associate $\tf$ bits, labeled $\vertex = (\bfr, t+1/2)$ for $0 \leq t \leq \tf-1$.
We say $\vertex=(\bfr, t+1/2)$ is incident to $\cc=(\qs', t'+1/2)$ if $t=t'$ and $\bfr$ belongs to the check $\qs'$ in the quantum code.
We write this as $\vertex \in \cc$.
\end{itemize}
The bits of this classical code are in one-to-one correspondence with possible locations of bitflip and syndrome errors.
A real history can therefore be represented by a binary vector $\error$, where $\error_\vertex$ indicates whether an error occurs on location $\vertex$ ($\error_\vertex = 1$) or not ($\error_\vertex = 0$).\footnote{In this Appendix we switch our notation from Ising spins (taking value $\pm 1$) to binary vectors (taking value in $\{0, 1\}$).}

If the maximum degree of checks in the quantum code is $w$, then the maximum degree of checks in the classical code is $w+2$.

The syndrome of the quantum measurement at $(\qs, t)$ induced by error $\error$ is given by
\begin{align}
    s_{\qs, t} = \error_{\qs, t} + \sum_{\bfr \in \qs} \sigma(\error')_{\bfr, t},
\end{align}
where $\sigma(\error')_{\bfr, t}$ is the cumulative bitflip error on $\bfr$ up til time $t$,
\begin{align}
    \label{eq:def_proj_error}
    \sigma(\error')_{\bfr, t} = \sum_{t=0}^{\tf-1} \error'_{\bfr, t+1/2}.
\end{align}
The additions are performed in $\mathbb{Z}_2$.
Up to our notational change from $\{+1,-1\}$ to $\{0,1\}$,
these definitions of $\sigma, s$ agree with those in Eq.~\eqref{eq:def_sigma_s}, and $\error$ is the same as defined in Eq.~\eqref{eq:def_eta}.
With boundary conditions $\error_{\qs, t=\ti} = \error_{\qs, t=\tf} = 0$, the measurement syndromes are in correspondences with syndromes of the classical code.
As one can check, for $\cc = (\qs, t+1/2)$, we have
\begin{align}
    \label{eq:classical_syndrome}
    \synd_\cc = \sum_{\vertex  \in \cc} 
    \error_\vertex
    = s_{\qs, t} + s_{\qs, t+1}.
\end{align}
We write this as $\synd = \partial \error$.
Therefore, two errors $\error$, $\error'$ are compatible to the same syndrome if  and only if $\partial \error = \partial \error'$.

We define a ``minimal weight'' (MW) decoder to be one taking as input the syndrome $\synd$, and outputs the an error history $\error$ with minimal Hamming weight,
\begin{align}
    \label{eq:def_MW_decoder}
    \widehat{\error}_{\rm MW}(\synd) = \mathrm{argmin}_{\error: \partial \error = \synd} |\error|.
\end{align}
Notice that this is weaker than what is usually meant by a (properly weighted) minimal weight decoder(which is itself already sub-optimal), where the errors are weighted by their log likelihood, so that the decoder maximizes the error probability of its output.
Despite it being weaker, we prove it has a threshold.
Below, we denote by $\mathbb{P}_{\rm fail}$ the failure probability of the MW decoder (see Lemma~\ref{lemma:1-1})
\begin{align}
    \mathbb{P}_{\rm fail} \equiv \epsilon^{\rm MW}(N) \equiv 1 - \mathbb{P}^{\rm all}_{\decoder^{\rm MW}}.
\end{align}

\subsection{Memory experiment}


We upper bound the failure probability of our MW decoder on Inference Problem~\ref{IP1}.
We say the decoder fails if \textit{any} of the $K$ logical qubit fails,
\begin{align}
\label{eq:P_fail_spatial_MW}
    \mathbb{P}_{\rm fail}
    =&\, \sum_{\error} \mathbb{P}(\error)  \cdot \mathbb{1}_{\exists j \in [K], L_j (\error + \widehat{\error}_{\rm MW}(\partial \error)) = 1}. 
\end{align}
Define $\error' = \error + \widehat{\error}_{\rm MW}(\partial \error)$, where the addition is understood to be element-wise.
By definition, we have $\partial \error' = 0$.
We also have by definition in Eq.~\eqref{eq:def_LJ}
\begin{align}
    \label{eq:LJ_error_prime}
    L_j(\error') = \sum_{\bfr \in L_j} \sigma(\error')_{\bfr, \tf}
    = \sum_\alpha \sum_{\bfr \in L_j} \sigma(\error'_\alpha)_{\bfr, \tf}
    \equiv
    \sum_\alpha L_j(\sigma(\error'_\alpha)_\tf).
\end{align}

For errors that satisfy $\partial \error' = 0$, we write 
\begin{align}
    \label{eq:error_disjoint_union}
    \error' = \bigsqcup_\alpha \error'_\alpha
\end{align}
if (i) $\error' = \sum_\alpha \error'_\alpha$, (ii) for all $\vertex$ there is at most one $\alpha$ for which $(\error'_{\alpha})_\vertex = 1$, and (iii) $\partial \error'_\alpha = 0$ for all $\alpha$.
We see that $\error'$ is a disjoint union of $\error'_\alpha$, when viewed as sets of vertices.
We say $\error'$ is \textit{reducible} if it admits a decomposition in Eq.~\eqref{eq:error_disjoint_union} with at least two terms, and 
irreducible otherwise.
For any finite length $\error'$ with $\partial \error' = 0$, there exists a decomposition of $\error'$ into disjoint, irreducible ones as in Eq.~\eqref{eq:error_disjoint_union}.

\begin{lemma}
\label{lemma:D1}
Suppose $\partial \error' = 0$ and $L_j(\error') = 1$ for some $j \in [K]$.
Consider any decomposition of $\error'$ into the form of Eq.~\eqref{eq:error_disjoint_union} where each $\error'_\alpha$ is irreducible.
There exists an $\alpha$ for which $|\error'_\alpha| \geq D$, where $D$ is the code distance of the quantum code.
\end{lemma}
\begin{proof}
We prove by contradiction.
Assuming there is a decomposition such that $|\error'_\alpha| < D$ for all $\alpha$.
For any $\qs$ and $\alpha$, we have by Eqs.~(\ref{eq:def_proj_error},\ref{eq:classical_syndrome}), boundary conditions $\error'_{\qs, t=\ti} = \error'_{\qs, t=\tf} = 0$, and the condition $\partial \error'_\alpha = 0$, that
\begin{align}
    \sum_{\bfr \in \qs} \sigma(\error'_\alpha)_{\bfr,\tf}
    = \sum_{t=0}^{\tf-1} (\partial \error'_\alpha)_{\cc=(\qs, t+1/2)} = 0.
\end{align}
Therefore, $\sigma(\error'_\alpha)_\tf$ forms a logical $X$ operator.
As $|\sigma(\error'_\alpha)| \leq |\error'_\alpha| < D$, it must be a trivial logical operator. 
We have $L_j(\sigma(\error'_\alpha)_\tf) = 0$ for all $j$, and hence $L_j(\error') = 0$ for all $j$ by Eq.~\eqref{eq:LJ_error_prime}.
This contradicts with our assumption.
\end{proof}


\begin{lemma}
\label{lemma:D2}
Let $\error' = \error + \widehat{\error}_{\rm MW}(\partial \error)$, and let $\cycle \subseteq \error'$ be an irreducible error guaranteed by Lemma~\ref{lemma:D1}, for which $|\cycle| \geq D$.
We have (i) $\error \cap \widehat{\error}_{\rm MW} \cap \cycle = \emptyset$, and (ii) $|\error \cap \cycle| \geq |\widehat{\error}_{\rm MW} \cap \cycle|$.
Therefore, $|\error \cap \cycle| + |\widehat{\error}_{\rm MW} \cap \cycle| = |\cycle|$, and $|\error \cap \cycle| \geq |\cycle| / 2$.
\end{lemma}
\begin{proof}
For any $\vertex \in \error \cap \widehat{\error}_{\rm MW}$, we have $\error'_{\vertex} = 0$, so $\vertex \notin \cycle \subseteq \error'$.
To show $|\error \cap \cycle| \geq |\widehat{\error}_{\rm MW} \cap \cycle|$, define 
\begin{align}
    \error''_\vertex = \begin{cases}
        (\widehat{\error}_{\rm MW})_\vertex, &\vertex \notin \cycle \\
        \error_\vertex,       &\vertex \in \cycle
    \end{cases}.
\end{align}
We can verify that $\error'' = \widehat{\error}_{\rm MW} + \cycle$, so that $\partial \error'' = \partial \widehat{\error}_{\rm MW} = \partial \error$.
Therefore, $\error''$ is a valid correction given the syndrome of $\error$.
In addition,
\begin{align}
    |\error''| = |\widehat{\error}_{\rm MW}\setminus \cycle| + |\error \cap \cycle| = |\widehat{\error}_{\rm MW}| - |\widehat{\error}_{\rm MW} \cap \cycle| + |\error \cap \cycle|.
\end{align}
As $\widehat{\error}_{\rm MW}$ is the minimal weight error by construction, we must have $|\error''|  \leq |\widehat{\error}_{\rm MW}|$, thus $|\error \cap \cycle| \geq |\widehat{\error}_{\rm MW} \cap \cycle|$.    
\end{proof}

Puttting together Lemmas~\ref{lemma:D1},\ref{lemma:D2}, we can bound Eq.~\eqref{eq:P_fail_spatial_MW} as follows
\begin{align}
    \label{eq:D11}
    \mathbb{P}_{\rm fail}
    =&\, \sum_{\error} \mathbb{P}(\error)  \cdot \mathbb{1}_{\exists j \in [K], L_j (\error + \widehat{\error}_{\rm MW}(\partial \error)) = 1} \nn
     \leq&\, \sum_{\error} \mathbb{P}(\error) \cdot  \mathbb{1}_{\exists \text{ irreducible } \cycle \subseteq \error', |\cycle| \geq D, \partial \cycle = 0} \nn
     \leq& 
     \sum_{\text{ irreducible } \cycle, |\cycle| \geq D, \partial \cycle = 0}
     \nn
     &\quad\quad\quad\cdot
     \sum_{\error \cap \cycle: |\error \cap \cycle| \geq |\cycle|/2} q^{|\error \cap \cycle|} (1-q)^{|\cycle| - |\error \cap \cycle|} \nn
     \leq&\, 
     \sum_{\text{ irreducible } \cycle, |\cycle| \geq D, \partial \cycle = 0}
     2^{|\cycle|} \cdot [q(1-q)]^{|\cycle|/2} \nn
     \equiv&\, 
     \sum_{\text{ irreducible } \cycle, |\cycle| \geq D, \partial \cycle = 0}
     \widetilde{q}^{|\cycle|}.
\end{align}
In the last step, we enumerate all subsets of $\cycle$ with size greater than $|\cycle|/2$, which we identify with $\error \cap \cycle$.
The number of such subsets is no more than $2^{|\cycle|}$.
We have defined $q \equiv \mathrm{max}\{p_{\rm bp}, p_{\rm m}\}$, $\widetilde{q} \equiv 2\sqrt{q(1-q)}$, and assumed $q < 1/2$.

Finally, we have the following Lemma from~\cite{pryadko2014}, based on a cluster-enumeration algorithm.
\begin{lemma}
\label{lemma:D3}
\begin{align}
    \label{eq:D12}
     \sum_{\text{ irreducible } \cycle, \partial \cycle = 0}
     \mathbb{1}_{|\cycle| \leq \ell} 
     \leq
     (N \cdot \tf) (w+1)^{\ell-1}.
\end{align}
\end{lemma}
\begin{proof}
We fix an ordering of the classical checks $\cc$.
Consider the following recursive algorithm searching for valid $\cycle$'s.
The state of the search is uniquely specified by an ordered array of vertices, which we call $\mathsf{V}$.
\begin{itemize}
\item 
For each $\vertex = (\bfr, t+1/2)$ of the classical code, call the procedure $\mathtt{search}$ with argument $\mathsf{V} = \{\vertex\}$.
\item
Within the procedure $\mathtt{search}$ with argument $\mathsf{V}$,
\begin{itemize}
\item
If $\partial \mathsf{V} = \emptyset$, return $\mathsf{V}$ as a solution of $\cycle$.
\item
Otherwise, if $|\mathsf{V}| = \ell$, return failure.
\item 
Otherwise, pick the first violated check in $\partial \mathsf{V}$ (say $\cc$) according to our chosen ordering.
For each $\vertex' \in \cc \setminus \mathsf{V}$, call the procedure $\mathtt{search}$ recursively with argument $\mathsf{V} \cup \{\vertex'\}$.
\end{itemize}
\end{itemize}
It is easy to see that the algorithm explores at most $(N \cdot \tf) (w+1)^{\ell-1}$ possible $\mathsf{V}$'s, as there are at most $(N \cdot \tf)$ possible different initial conditions, and there are at most $w+1$ choices of $\vertex'$ at each step (recall that the maximum degree of the classical checks is $w+2$).

Finally, we note that any valid $\cycle$ will necessarily contain at least a classical bit of the form $\vertex = (\bfr, t+1/2)$.
Furthermore, the following properties are preserved throughout the algorithm for any valid $\cycle$,
\begin{itemize}
\item
    $\mathsf{V} \subseteq \cycle$.
    Hence, $\partial \mathsf{V} = \emptyset$ if and only if $\mathsf{V} = \cycle$, due to irreducibility of $\cycle$.
\item
    When $\mathsf{V} \subsetneq \cycle$, $\partial \mathsf{V} \neq \emptyset$.
    For the first violated check $\cc \in \partial \mathsf{V}$, we have
    $\cycle \cap (\cc \setminus \mathsf{V}) = \cc \cap (\cycle \setminus \mathsf{V}) \neq \emptyset$.
\end{itemize}
These guarantee that any valid $\cycle$ can be found by the search algorithm after at most $|\cycle|$ levels of recursion.
Therefore, there are fewer valid solutions than all solutions explored by the algorithm, which is Eq.~\eqref{eq:D12}.
\end{proof}

Eqs.~(\ref{eq:D11},\ref{eq:D12}) together imply
\begin{align}
\label{eq:pfail_memory_dennis}
    \mathbb{P}_{\rm fail}
    \leq &\, 
     \sum_{\ell = D}^{\infty} 
     (\widetilde{q})^{\ell}
     \sum_{\text{ irreducible } \cycle, \partial \cycle = 0}
     \mathbb{1}_{|\cycle| = \ell} \nn
     \leq &\, 
     (N \cdot \tf) (w+1)^{-1}.
     \sum_{\ell = D}^{\infty} 
     [(w+1)\widetilde{q}]^{\ell} \nn
     =&\, O((N \cdot \tf) \cdot [(w+1)\widetilde{q}]^D).
\end{align}

\subsection{Stability experiments}

We consider the ``hard wall'' stability experiment described in Sec.~\ref{sec:stability_hard_wall}.
The failure probability can be similarly evaluated as follows,
\begin{align}
\label{eq:P_fail_temporal_MW}
    \mathbb{P}_{\rm fail} 
    =&\, \sum_{\error} \mathbb{P}(\error)  \cdot \mathbb{1}_{\PRXQ{L_W}(\error \cdot \widehat{\error}_{\rm MW}(\error)) = 1},
\end{align}
where $\PRXQ{L_W}$ is now a product of stabilizers.
Here, the MW decoder only consider those errors obeying the hard wall constraint, as those violating the condition has a zero probability of occuring.
Equivalently, this can be formally realized by ``dropping out'' corresponding classical bits on the hard wall in the classical code.
They take the form $\vertex=(\bfr,t+1/2)$ for $\bfr \in W, t \in [0, \tf-1]$.
We similarly define $\error' = \error \cdot \widehat{\error}_{\rm MW}(\error)$.
By definition of the MW decoder Eq.~\eqref{eq:def_MW_decoder}, we have $\partial \error' = 0$.

Similarly to Eq.~\eqref{eq:def_proj_error}, we define the cumulative error at $t$ to be
\begin{align}
    \label{eq:def_proj_error_hard_wall}
    \sigma(\error')_{\bfr, t} = \sum_{t'=-\infty}^{t-1} \error'_{\bfr, t'+1/2}.
\end{align}
We have, by definition, for any $t \in [0, \tf]$
\begin{align}
    \label{eq:def_W_interms_of_proj_error}
    \PRXQ{L_W}(\error') = \sum_{\bfr \in W} \sigma(\error')_{\bfr, t}.
\end{align}
Its value is conserved during $t \in [0, \tf]$ (and therefore independent of $t$) due to the hard wall constraint.

\begin{lemma}
\label{lemma:D4}
Suppose $\partial \error' = 0$ and $\PRXQ{L_W}(\error') = 1$.
Consider any decomposition of $\error'$ into the form of Eq.~\eqref{eq:error_disjoint_union} where each $\error'_\alpha$ is irreducible and satisfy $\partial \error'_\alpha = 0$.
Consider a decomposition of $\PRXQ{L_W}$ into a minimal number of stabilizers,
\begin{align}
    \PRXQ{L_W} = \prod_{\ell=1}^{\mathrm{area}(W)} \pauliz_{\qs_\ell}
    \equiv
    \prod_{\qs \in W} \pauliz_\qs.
\end{align}
There exists an $\error'_\alpha$ such that for all $t\in[\tf]$, there exists $\qs \in W$ where $(\error'_\alpha)_{\qs, t} = 1$.
\end{lemma}
\begin{proof}
We can write
\begin{align}
   \PRXQ{L_W}(\error') = \sum_\alpha \PRXQ{L_W}(\error'_\alpha).
\end{align}
By assumption, $\error'$ does not contain any bitflips on the hard wall.
Consequently, none of the $\error'_\alpha$ contains bitflips on the hard wall.
Therefore, $\PRXQ{L_W}(\error'_\alpha)$ is also conserved for $t \in [0, \tf]$ for each $\alpha$, in the sense of Eq.~\eqref{eq:def_W_interms_of_proj_error}.
For $\PRXQ{L_W}(\error') = 1$ to hold, there is at least one $\error'_\alpha$ for which $\PRXQ{L_W}(\error'_\alpha) = 1$ for $t \in [\tf]$.


For any $t \in [\tf]$, we have according to Eqs.~(\ref{eq:def_proj_error_hard_wall}, \ref{eq:def_W_interms_of_proj_error}) and the condition $\partial \error'_\alpha=0$ that
\begin{align}
    &\PRXQ{L_W}(\error'_\alpha) + \sum_{\qs \in W} (\error'_\alpha)_{\qs, t}\nn
    =&\,
    \sum_{\bfr \in W} \sigma(\error'_\alpha)_{\bfr,t} + 
    \sum_{\qs \in W} (\error'_\alpha)_{\qs, t} \nn
    =&\,
    \sum_{\bfr \in W} \sum_{t'=-\infty}^{t-1} (\error'_\alpha)_{\bfr, t'+1/2} +
    \sum_{\qs \in W} (\error'_\alpha)_{\qs, t} \nn
    =&\,
    \sum_{t'=-\infty}^{t-1} \sum_{\qs \in W} \synd_{\cc=(\qs, t'+1/2)} \nn
    =&\, 0.
\end{align}
As $\PRXQ{L_W}(\error'_\alpha) = 1$ for each $t \in [\tf]$, this implies that 
\begin{align}
    \forall t \in [\tf], \quad
    \sum_{\qs \in W} (\error'_\alpha)_{\qs, t} = 1.
\end{align}
Therefore, for each $t \in [\tf]$, we must have at least one $\qs \in W$ such that $(\error'_\alpha)_{\qs, t} = 1$.
In particular, this means that $|\error'_\alpha| \geq \tf$.
\end{proof}

Let $\cycle \subseteq \error'$ be the $\error'_\alpha$ guaranteed by Lemma~\ref{lemma:D4}. 
Lemma~\ref{lemma:D2} still applies to $\cycle$, whereas an analog of Lemma~\ref{lemma:D3} for this case reads
\begin{align}
     \sum_{\text{ irreducible } \cycle, \partial \cycle = 0}
     \mathbb{1}_{|\cycle| \leq \ell} 
     \leq
     \mathrm{area}(W) \cdot (w+1)^{\ell-1}.
\end{align}
The prefactor $\mathrm{area}(W)$ counts the number of initial conditions of the recursive search algorithm, as any valid $\cycle$ will contain at least a bit of the form $\vertex = (\qs, \tf)$ where $\qs \in W$.
With these, we upper bound the failure probability of the ``hard wall'' stability experiment as
\begin{align}
\label{eq:pfail_hard_wall_stability_dennis}
    \mathbb{P}_{\rm fail}
    \leq&\, \mathrm{area}(W) \cdot 
    (w+1)^{-1} \sum_{\ell = \tf}^{\infty} [(w+2)\cdot \widetilde{q}]^{\ell} \nn
    =&\, 
    O(\mathrm{area}(W) \cdot [(w+1) \cdot \widetilde{q}]^{\tf}).
\end{align}

For ``hard wall with two punctures'', $\mathbb{P}_{\rm fail}$ receives contributions from only those $\cycle \subseteq \error'$ that contains both $\vertex=(\bfr = 0, t =\ti-1)$ and $\vertex=(\bfr = 0, t =\tf)$, as illustrated in Fig.~\ref{fig:stability-experiment-2pt}.
Microscopically, due to infinite rounds of measurement before $t=\ti$ and after $t = \tf$, the decoder knows with certainly whether the two punctures contributes an even or odd number of bitflips, in total.
The only failure mode is for the decoder to predict bitflip events on both punctures incorrectly.
Lemmas~\ref{lemma:D2},\ref{lemma:D4} still apply, whereas Lemma~\ref{lemma:D3} needs to be modified as
\begin{align}
     \sum_{\text{ irreducible } \cycle, \partial \cycle = 0}
     \mathbb{1}_{|\cycle| \leq \ell} 
     \leq
    (w+1)^{\ell-1}.
\end{align}
There is no prefactor, as $\cycle$ always contains $\vertex=(\bfr = 0, t =\tf)$, and the recursive search can always start with this initial condition.
The result reads
\begin{align}
    \label{eq:pfail_hard_wall_stability_two_punctures_dennis}
    \mathbb{P}_{\rm fail}
    =
    O([(w+1) \cdot \widetilde{q}]^{\tf}).
\end{align}

For 2D toric codes (in either Euclidean or hyperbolic planes) with periodic boundary conditions, each qubit has degree $2$.
Therefore, bitflip errors and measurement errors can be viewed as edges in a graph in spacetime, connecting ``vertices'' which are checks of the classical code.
The condition $\partial \error'=0$ imples that $\error'$ must be a collection of Euler cycles/paths.

With a single flux line insertion (Sec.~\ref{sec:stability_hyperbolic_surface_code}, Fig.~\ref{fig:single-flux-threading}), failure to predict $\sigma_{\rm anc}$ necessarily requires $\error'_{\qs_0, \ti} = \error'_{\qs_0, \tf} = 1$.
This means that classical checks $\cc = (\qs_0, \ti-1/2)$ and $\cc = (\qs_0, \tf+1/2)$ are the only vertices in the graph that have an odd degree.
They must therefore be connected by a Euler path $\cycle \subseteq \error'$.
This is the analog of Lemma~\ref{lemma:D4}.
Combining this with Lemma~\ref{lemma:D2} and a simplification of Lemma~\ref{lemma:D3}, we obtain an upper bound of the same form as in
Eq.~\eqref{eq:pfail_hard_wall_stability_two_punctures_dennis}.

With double flux line insertion (Sec.~\ref{sec:stability_hyperbolic_surface_code_two_fluxes}), failure to predict $\sigma_{\rm anc}$ means that $\error'_{\qs_1, \ti} = \error'_{\qs_1, \tf} = \error'_{\qs_2, \ti} = \error'_{\qs_2, \tf} =  1$.
There are three different ``failure modes'', corresponding to different possible connectivity patterns among the four odd-degree checks, as shown in Fig.~\ref{fig:double-flux-threading}.
A union bound gives
\begin{align}
    \label{eq:pfail_upperbound_doule_flux_threading_dennis}
    & \mathbb{P}_{\rm fail}\nn
    =&\
    O([(w+1) \cdot \widetilde{q}]^{2\tf}
    +
    [(w+1) \cdot \widetilde{q}]^{2d}
    +
    [(w+1) \cdot \widetilde{q}]^{2d+2\tf}),
\end{align}
where $d$ is the length of the shortest path between $\qs_{1,2}$.
Taking $d \to \infty$ we obtain Eq.~\eqref{eq:pfail_upperbound_doule_flux_threading}.


\section{The continuous time limit \label{sec:continuous_time_limit}}

In Sec.~\ref{sec:decoding_success_physical_consequences} we assumed that the continuous time limit can be taken, and the existence of a stable phase in the discrete-time path integral Eq.~\eqref{eq:discrete_time_path_integral} also leads to a stable phase of the continuous time path integral Eq.~\eqref{eq:continuous_time_path_integral}.
This assumption is mathematically nontrivial.
Here, we argue that additional evidence (beyond the existence of decoding thresholds) are needed in order to justify this assumption.

We first explain why the assumption is nontrivial.
Recall the perturbed quantum Hamiltonian and the continuous time path integral Eq.~\eqref{eq:continuous_time_path_integral}
\begin{align}
    &H_{\rm q} \equiv 
    - \sum_\qs \prod_{\bfr \in \qs} \pauliz_{\bfr} 
    -\sum_{\widetilde{\qs}} \prod_{\bfr \in \widetilde{\qs}} \paulix_{\bfr} 
    -
    h  \sum_\bfr \paulix_{\bfr}, \\
    &e^{-\beta H_{\rm q}} = 
    \Pi_\paulix
    \left[
    \lim_{\delta \tau \to 0}\prod_{t = 0}^{\beta/\delta \tau} 
    e^{\delta \tau \sum_\qs \prod_{\bfr \in \qs} \pauliz_{\bfr, t} } \cdot e^{h \delta \tau \sum_\bfr \paulix_{\bfr, t} } \right]
    \Pi_\paulix.
\end{align}
On the other hand, transfer matrices for discrete time path integral reads Eq.~\eqref{eq:discrete_time_path_integral}
\begin{equation}
    T(\tf) = \prod_{t=\ti}^{\tf-1} e^{\km \sum_\qs \prod_{\bfr \in \qs} \pauliz_{\bfr, t} } \cdot e^{\overline{\kbf} \sum_\bfr \paulix_{\bfr, t} },
\end{equation}
where $\tanh \overline{\kbf} = e^{-2\kbf}$.
It defines a clean stat mech model with couplings $(\km, \kbf)$.
In the limit $\delta \tau \to 0$, the two path integrals agree if
\begin{align}
\begin{split}
    \label{eq:couplings_small_delta_tau_limit}
    e^{-2\kbf} \to&\ h \delta \tau,\\
    \km \to&\ \delta \tau.
\end{split}
\end{align}
Recall that our strategy to establishing stability in the clean stat model is via the disordered stat mech model with random sign disorders, where the disorder strengths $p_{\rm m}, p_{\rm bf}$ are related to the couplings via the Nishimori conditions
\begin{align}
\begin{split}
    e^{-2\km} = \frac{p_{\rm m}}{1-p_{\rm m}}, \\
    e^{-2\kbf} = \frac{p_{\rm bf}}{1-p_{\rm bf}}.
\end{split}
\end{align}
Under the limit in Eq.~\eqref{eq:couplings_small_delta_tau_limit}, we have
\begin{align}
\begin{split}
    p_{\rm bf} \to&\ h \delta \tau,\\
    p_{\rm m} \to&\ \frac{1}{2} (1- \delta \tau) \approx \frac{1}{2} - \frac{1}{2h} p_{\rm bf}.
\end{split}
\end{align}
The stability of the disordered model with disorder strengths $(p_{\rm m}, p_{\rm bf})$\footnote{Here we suppress the explicit dependence on $\km, \kbf$, assuming they are set by $p_{\rm m}, p_{\rm bf}$ via Nishimori conditions.}
is in turn established by bounding the failure probabilities of decoding, e.g. Eq.~\eqref{eq:pfail_upperbound_memory_maintext}.
These bounds come from Appendix~\ref{sec:pfail_upperbound_MW}, and are only useful when $p_{\rm m}, p_{\rm bf}$ are both sufficiently small.
These requirements are clearly not met when taking the limit $\delta \tau \to 0$, where $p_{\rm bf} \to 0$ and $p_{\rm m} \to 1/2$.

Below, we argue that the stability of the disordered model generically breaks down in the continuous time limit.
In particular, on the Nishimori surface of the phase diagram, we find no ordering along the line $p_{\rm m} = \frac{1}{2} - \frac{1}{2h} p_{\rm bf}$ near $(p_{\rm bf}, p_{\rm m}) = (0, 1/2)$ for any $h > 0$.
Therefore, even if a stable phase of the quantum Hamiltonian exists, it does not rigorously follow from stability of its disordered version with the same coupling strengths on the Nishimori surface.
Additional considerations or assumptions are needed.
Our approach to stability is therefore  not conclusive, although suggestive.

Our argument is based on a physical error model with $p_{\rm bf} \to 0$ and $p_{\rm m} = \frac{1}{2} - \frac{1}{2h} p_{\rm bf}$.
Suppose bitflip errors are described by a Poisson point process, and suppose the syndromes are extracted frequently, in intervals of $\delta \tau$.
In the ``continuous measurement'' limit $\delta \tau \to 0$ we have $p_{\rm bf} \propto \delta \tau$, simply because it takes time for errors to accumulate.
On the other hand, in this limit the syndromes are also less reliable (recall Heisenberg uncertainty principle).
In our model, the measurement error rate is $p_{\rm m} = \frac{1}{2} - \frac{1}{2h} p_{\rm bf}$, becoming complete noise in the limit $\delta \tau \to 0$.

Reliable syndrome information can be built by repeating the measurements for $m$ rounds, collecting enough information before trying to decode.
Meanwhile, bitflip errors will accumulate.
When both ``effective'' error rates are finite but small, a decoding threshold can be established as in Appendix~\ref{sec:pfail_upperbound_MW}.
More precisely, define $\Delta \tau = m \delta \tau$, we can calculate an effective ``renormalized'' bitflip error rate as follows,
\begin{align}
    p_{\rm bf}(\Delta \tau) =&\, \sum_{k = 0}^{m} \mathbb{1}_{\text{$k$ odd}} \binom{m}{k} (p_{\rm bf})^{k} (1-p_{\rm bf})^{m-k} \nn
    =&\, \frac{1}{2}(1-(1-2 p_{\rm bf})^m).
\end{align}
For $p_{\rm bf}(\Delta \tau)$ to remain small, we require $m = O((p_{\rm bf})^{-1})$.
In this case, one can simply take the majority vote of all $m$ measurements to infer the correct syndrome, and we can calculate an effective measurement error rate 
\begin{align}
    p_{\rm m}(\Delta \tau) =&\, \sum_{k = m/2+1}^{m} \binom{m}{k} (p_{\rm m})^{k} (1-p_{\rm m})^{m-k} \nn
    \leq&\, [p_{\rm m}(1-p_{\rm m})]^{m/2} \cdot 2^{m-1} \nn
    =&\, \left[\frac{1}{4}-\frac{(p_{\rm bf})^2}{4h^2}\right]^{m/2} \cdot 2^{m-1}\nn
    =&\, O(e^{-\frac{m}{2h^2} (p_{\rm bf})^2}).
\end{align}
For $p_{\rm m}(\Delta \tau)$ to be small, we will need $m = \Omega((p_{\rm bf})^{-2})$.
The two conditions we wish to impose on $m$ are in conflict with each other, in the limit $p_{\rm bf} \to 0$.
We therefore conclude that no ``ordered phase'' can be established along the line $p_{\rm m} = \frac{1}{2} - \frac{1}{2h} p_{\rm bf}$ near $(p_{\rm bf}, p_{\rm m}) = (0, 1/2)$ for any $h > 0$.

On the other hand, if in our physical error model $p_{\rm m} = \frac{1}{2} - \frac{1}{2h} \sqrt{p_{\rm bf}}$, stability of this model can be established by choosing $m = \Theta((p_{\rm bf})^{-1})$ for sufficiently small $h$.
These arguments therefore suggest that the phase boundary between ``ordered'' and ``disordered'' phases should generically be given by
\begin{align}
    \label{eq:phase_boundary_pm_pbf}
    p_{\rm m} = \frac{1}{2} - \frac{1}{2h_c} \sqrt{p_{\rm bf}}, 
\end{align}
near $(p_{\rm bf}, p_{\rm m}) = (0, 1/2)$, where $h_c = O(1)$.

To test our ``rescaling'' argument, we provide a numerical lower bound to the ordered phase of the 2D random-bond Ising model on the Nishimori surface, by performing a memory experiment on the 1D repetition code, choosing minimal-weight perfect matching as the decoder.
The results are shown in Fig.~\ref{fig:RBIM_matching}, where we find consistency between the phase boundary and Eq.~\eqref{eq:phase_boundary_pm_pbf}.

\begin{figure}[t]
    \centering
    \includegraphics[width=.8\linewidth]{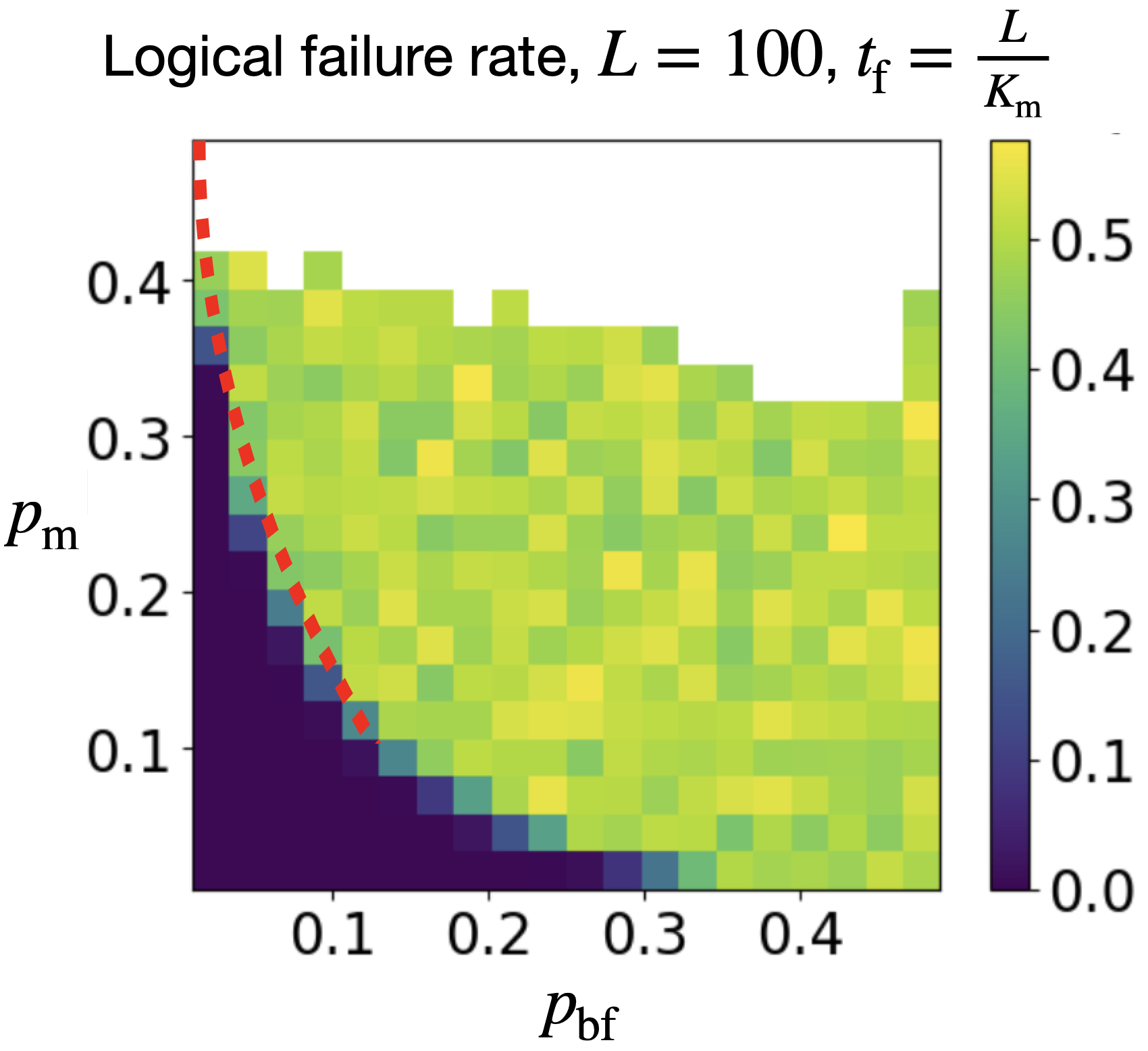}
    \caption{Numerical results from decoding the repetition code with minimal-weight perfect matching.
    We include both bitflip errors and syndrome errors.
    We approximate the phase boundary by sharp jumps in the failure probability at $L = 100$.
    The shape of the approximate phase boundary is in agreement with Eq.~\eqref{eq:phase_boundary_pm_pbf}.
    }
    \label{fig:RBIM_matching}
\end{figure}

\section{Example stat mech models beyond CSS stabilizer codes under $\paulix$ and $\pauliz$ noise \label{sec:examples_beyond}}

\subsection{Bacon-Shor code \label{sec:Bacon-Shor}}

We consider repeated faulty measurements of the Bacon-Shor code~\cite{bacon-shor-2006}.
As it is a subsystem code, we distinguish the terms ``checks'' and ``stabilizers''. 
The former refers to the operators being measured which may not themselves commute, and the latter refers to elements of the center of the group generated by the checks.

The qubits live on vertices of a 2D square lattice.
The $\pauliz$ checks live on vertical pairs $B_{\bf r} = \pauliz_\mathbf{r} \pauliz_{\mathbf{r}+\hat{\mathbf{z}}}$, and the $\paulix$ checks live on horizontal edges $A_{\bf r} = \paulix_\mathbf{r} \paulix_{\mathbf{r}+\hat{\mathbf{x}}}$.
In this convention, the $\pauliz$ logicals are straight horizontal lines, whereas the $\paulix$ logicals are straight vertical lines,
\begin{align}
    \overline{\pauliz} =& \pauliz_{(0,0)} \pauliz_{(1,0)} \pauliz_{(2,0)} \ldots \\
    \overline{\paulix} =& \paulix_{(0,0)} \paulix_{(0,1)} \paulix_{(0,2)} \ldots
\end{align}
The $\pauliz$ stabilizers are generated by the following representatives,
\begin{align}
\label{eq:Bacon-Shor-Z-stabilizer}
     \forall z, \quad \pauliz_{(0,z)} \pauliz_{(0,z+1)} \pauliz_{(1,z)} \pauliz_{(1,z+1)} \pauliz_{(2,z)} \pauliz_{(2,z+1)} \ldots
\end{align}
and similarly for the $\paulix$ stabilizers.

We consider a measurement sequence where the $\paulix$ and $\pauliz$ checks are measured in alterative rounds.
Due to the CSS nature of the code, we focus on $\paulix$ errors and $\pauliz$ check syndromes.
This allows us to again use $\sigma_{\bfr, t}$ to denote the cumulative $\paulix$ error on qubit ${\bf r}$.
The transition probabilities is identical to Eq.~\eqref{eq:tempphyserr},
\begin{align}
    \mathbb{P}(\sigma_{t+1} | \sigma_{t}) 
    \propto 
    e^{\kbf \sum_{\mathbf{r}}
    \sigma_{t+1, \mathbf{r}} \sigma_{t,\mathbf{r}} }.
\end{align}
We make a distinction between the \textit{true eigenvalue} of $B_\mathbf{r}$ on the post-measurement state (denoted $S_\mathbf{r} = \pm 1$), and the \textit{syndrome} of $B_\mathbf{r}$ (denoted $s_\mathbf{r} = \pm 1$).
We say a measurement error occurs if $s_\mathbf{r} \neq S_\mathbf{r}$.
Thus, 
\begin{align}
     \mathbb{P}(s_t | S_t)
     \propto  
      \exp{ \km \sum_{\mathbf{r}} s_{\mathbf{r},t} S_{\mathbf{r},t} }.
\end{align}
Since $\paulix$ measurements are performed in a previous round, \textit{most} $S_\mathbf{r}$ are going to be equally likely $\pm 1$ --- except that the $\pauliz$ stabilizers Eq.~\eqref{eq:Bacon-Shor-Z-stabilizer} need to be preserved in the absense of errors.
This means that
\begin{align}
\label{eq:S_sigma_product_equality}
    \forall z, \quad \prod_x S_{(x, z), t} = \prod_x \sigma_{(x,z), t} \sigma_{(x,z+1),t}.
\end{align}
Formally, the following conditional probability takes a simple form,
\begin{align}
    \mathbb{P}(S_t | \sigma_t) \propto \prod_z   \mathbb{1}_{\prod_x S_{(x, z), t} = \prod_x \sigma_{(x,z), t} \sigma_{(x,z+1), t}}.
\end{align}
The analog of Eq.~\eqref{eq:measerr} is obtained by integrating out $S$,
\begin{align}
    &   \mathbb{P}(s_t | \sigma_t) \nn
    =&\ \sum_{S_t} \mathbb{P}(s_t | S_t) \mathbb{P}(S_t | \sigma_t) \nn
    \propto&\ 
    \sum_{S_t} 
    \exp{ \km \sum_{\mathbf{r}, t} s_{\mathbf{r},t} S_{\mathbf{r},t} } 
    \prod_z   \mathbb{1}_{\prod_x S_{(x, z), t} = \prod_x \sigma_{(x,z), t} \sigma_{(x,z+1), t}} \nn
    \propto&\  \prod_z f_z
\end{align}
where
\begin{widetext}
\begin{align}
    f_z = 
    \begin{cases}
        \sum_{\text{even $k$}} \binom{L}{k} e^{-2k \km} = \frac{1}{2} \( (1+e^{-2 \km})^L + (1-e^{-2 \km})^L \) &\quad \text{ if } \prod_x s_{(x, z),t} \cdot \prod_x \sigma_{(x,z),t} \sigma_{(x,z+1),t} = +1\\ 
        \sum_{\text{odd $k$}}\ \binom{L}{k} e^{-2k\km} = \frac{1}{2} \( (1+e^{-2\km})^L - (1-e^{-2\km})^L \)   
        &\quad \text{ if } \prod_x s_{(x, z),t} \cdot \prod_x \sigma_{(x,z),t} \sigma_{(x,z+1), t} = -1
    \end{cases}.
\end{align}
\end{widetext}
Define
\begin{align}
    e^{-2\widetilde{\km}} = \frac{(1+e^{-2\km})^L - (1-e^{-2\km})^L}
    {(1+e^{-2\km})^L + (1-e^{-2\km})^L}, 
\end{align}
where
\begin{align}
    \label{eq:km_asymptotic_bacon_shor}
    \widetilde{\km} \sim 
    \begin{cases}
        \km &\text{ as } L \ll e^{2\km} \\
        (\tanh \km)^L
        &\text{ as } L \gg e^{2\km}
    \end{cases},
\end{align}
we can write
\begin{align}
    \mathbb{P}(s_t | \sigma_t) 
    \propto \exp{ \sum_z \widetilde{\km} \prod_x s_{(x,z),t} \sigma_{(x,z),t} \sigma_{(x,z+1),t} }.
\end{align}
The weights of the stat mech model are given by the Hamiltonian
\begin{align}
    &H(\sigma, s) \nn
    &= \sum_{t,z} 
        \widetilde{\km} \prod_x (s_{(x,z),t} \sigma_{(x,z),t} \sigma_{(x,z+1),t}) + \kbf 
    \sum_{t, \mathbf{r}}
        \sigma_{\mathbf{r},t+1} \sigma_{\mathbf{r},t}.
\end{align}
Under the change of basis in Sec.~\ref{sec:change_of_basis} we obtain  a model with uncorrelated disorder,
\begin{align}
    \label{eq:H_hat_bacon_shor}
    &\widehat{H}(\sigma, \error) \nn
    &= \sum_{t,z} 
        \widetilde{\km} \prod_x (\error_{(x,z),t} \sigma_{(x,z),t} \sigma_{(x,z+1),t}) \nn
    &\quad + \kbf 
    \sum_{t, \mathbf{r}}
        \error_{t, \mathbf{r}}
        \sigma_{\mathbf{r},t+1} \sigma_{\mathbf{r},t}.
\end{align}
Decoding success is signified by the same condition as before, namely
\begin{align}
    [|\avg{L_J}|]_\error \to 1.
\end{align}
From Eq.~\eqref{eq:km_asymptotic_bacon_shor} we can readily see that the system should have no phase transition, as $\lim_{L \to \infty} \widetilde{\km} = 0$, and Eq.~\eqref{eq:H_hat_bacon_shor} reduces to decoupled 1D Ising chains.

\subsection{Toric code under pure $\pauliy$ noise \label{sec:toric_code_Y}}

We adopt the ``rotated'' representation of the toric code, which is more convenient in this case.
The simplicity of the error model allows us to still apply the transfer matrix method.
Similarly to Eq.~\eqref{eq:def_sigma_s} we define
\begin{align}
\begin{split}
    \sigma_{\bfr,t} =&\ (-1)^{\text{\# of $\pauliy$ errors on site }\bfr\text{ by time }t}, \\
    s_{\qs,t} =&\ \text{measurement outcome of stabilizer }\qs\text{ at time } t.
\end{split}
\end{align}
Here, $\qs$ can represent $\paulix$ or $\pauliz$ stabilizers, as $\pauliy$ errors can flip both types.
According to Eq.~\eqref{eq:Z_eta_maintext}, the stat mech model we obtain is a disordered version of the plaquette Ising model in a transverse field.
More precisely, the model is defined on a cubic lattice, with spins living on vertices.
There are disordered Ising couplings of strength $\pm \kbf$ between spins on neighboring time steps (representing bit flips), and there are disordered 4-spin couplings $\pm \km$ in each plaquette.
There are no gauge invariance in this model.
Its clean limit has a self-duality~\cite{xumoore_2004, nussinov2005compass}, from which its critical temperature can be determined.

Interestingly, the model we obtain by applying an $\pauliy$ field in the toric code has the same phase diagram as the one above.
They are related by a duality transformation, which changes the ground state degeneracy~\cite{vidal_2009_xumoore} but not the phase diagram.

We now turn to the limit of perfect measurements.
This amounts to taking $\km \to \infty$ and setting $\error_{\qs, t} = +1$ everywhere, as we explained in Sec.~\ref{sec:perfect_measurements}.
It imposes the following constraint on the spins,
\begin{align}
    \forall \Box, t, \quad \prod_{\bfr \in \Box} \sigma_{\bfr, t} = 1.
\end{align}
Despite the notational similarity, this constraint is different from Eq.~\eqref{eq:perfect_measurement_contraint_RPGM}, as the spins live on the vertices rather than the edges of the squares.
With periodic boundary conditions the general solution to these constraints take the form
\begin{align}
\label{eq:toric_code_y_noise_constraint}
    \sigma_{\bfr=(i,j), t} = \tau^{(0)}_t \cdot \tau^{(x)}_{i,t} \cdot \tau^{(y)}_{j,t}.
\end{align}
Here, $\tau^{(0)}_t, \tau^{(x)}_{i,t}, \tau^{(y)}_{j,t}$ take values in $\{-1, +1\}$.
At each $t$, there is a $\tau^{(x)}_{i,t}$ for each row $i$, and there is a $\tau^{(y)}_{j,t}$ for each column $j$.
(At $t=0$, we have $\tau^{(0)}_t = \tau^{(x)}_{i,t} = \tau^{(y)}_{j,t} = +1$ for all $i$, $j$.)
The remaining part of the action reads
\begin{align}
    &\lim_{\km \to \infty} 
    \widehat{H}(\sigma, \error)
    =
    -\sum_{t} [ \kbf 
    \sum_{\bfr = (i,j)}
        \error_{\bfr, t} 
        \cdot
     \widetilde{\tau}^{(0)}_t
     \cdot
     \widetilde{\tau}^{(x)}_{i,t}
     \cdot
     \widetilde{\tau}^{(y)}_{j,t}
    ],
\end{align}
where $\widetilde{\tau}_{t} \equiv \tau_{t}  \tau_{t+1}$.
As spins at different $t$ decouple in this model, we can discuss different $t$ independently of each other.

At each $t$, $\widetilde{\tau}^{(0)}$ is a constant, which can be absorbed into the definition of $\widetilde{\tau}^{(x)}_{i}$.

\begin{proposition}
When $p_{\rm bf} < 1/2$ the ground state has 
\begin{align}
    \widetilde{\tau}_{i,t}^{(x)} = +1, \widetilde{\tau}_{j,t}^{(y)} = +1, \quad \forall t, i, j.
\end{align}
The energy difference between the ground state and any excited state is $\Omega(L)$.
\end{proposition}
\begin{proof}
We will suppress the subscript $t$ in the following.
Let $\mathsf{R}$ be the subset of row indices $i$ for which $\widetilde{\tau}^{(x)}_i = -1$, and denote its complement $\overline{\mathsf{R}}$.
Similarly, let $\mathsf{C}$ be the subset of column indices $j$ for which $\widetilde{\tau}^{(y)}_j = -1$, and denote its complement $\overline{\mathsf{C}}$.
The energy of this state, as compared to that of the all $+1$ state, is
\begin{align}
\label{eq:I21}
    \Delta E 
    =&\, (-\kbf) 
        \sum_{\bfr=(i,j)} \error_{\bfr} (\widetilde{\tau}_i^{(x)} \widetilde{\tau}_j^{(y)}  - 1)\nn
    =&\, (-\kbf) 
        \sum_{\bfr=(i,j)} \error_{\bfr} (\widetilde{\tau}_i^{(x)} \widetilde{\tau}_j^{(y)}  - 1)\nn
    =&\, (2\kbf) 
        \sum_{\bfr=(i,j)} \error_{\bfr} 
        (\mathbb{1}_{i \in \mathsf{R}, j \in \ovl{\mathsf{C}}}
        +
        \mathbb{1}_{i \in \ovl{\mathsf{R}}, j \in \mathsf{C}}
        ) \nn
    =&\, (2\kbf) 
        \left[
        \sum_{i \in \mathsf{R}, j \in \ovl{\mathsf{C}}} \error_{\bfr = (i,j)} 
        +
        \sum_{i \in \ovl{\mathsf{R}}, j \in \mathsf{C}} \error_{\bfr = (i,j)} \right] 
\end{align}
Note that the two sums are over two disjoint sets of independent random variables.
The two sets are empty if $\mathsf{R} = \mathsf{C} = \emptyset$, corresponding to the proposed ground state; or $\overline{\mathsf{R}} = \overline{\mathsf{C}} = \emptyset$, which gives the same spin configuration as the proposed ground state, see Eq.~\eqref{eq:toric_code_y_noise_constraint}.
In any other case, we get a state different from the proposed ground state, \PRXQ{where the number of summands on the RHS of Eq.~\eqref{eq:I21}, namely $M = |R| \cdot |\overline{C}| + |\overline{R}| 
\cdot |C|$, is nonzero.
Therefore, $\Delta E / (2 \kbf)$ converges to a random variable distributed identically to  $\binom{M}{(1-2\pbf)M}$, which by $M \geq L$ is positive with probability approaching $1$ for $\pbf < 1/2$.}
\end{proof}
As a consequence, the partition function is dominated by $\sigma_{\bfr, t} = +1$  (with the $t=0$ boundary condition) as long as $p_{\rm bf} < 1/2$, corresponding to a 50\% threshold.
The case with open boundary conditions is more involved, and was studied in~\cite{Tuckett_2019}.

\end{document}